\documentclass[11pt]{article}

 \usepackage{amsmath}
 \usepackage{amssymb}
\usepackage{amsthm}
\usepackage{dsfont} 
\usepackage{color}
\usepackage{xcolor}
\usepackage{array}
\usepackage{tcolorbox}

\usepackage[breaklinks=true, colorlinks, bookmarks=false, citecolor=C1]{hyperref}

\usepackage{graphicx,subcaption}
\usepackage{multirow} 
\usepackage{multicol}
\usepackage{natbib}
\setcitestyle{citesep={,}}
\usepackage{booktabs}
\usepackage{color}
\usepackage[ruled,vlined]{algorithm2e}

\usepackage{comment}

\newtheorem{lemma}{Lemma}

\newtheorem{proposition}{Proposition}
\newtheorem{definition}{Definition}

\usepackage[margin=1.35in]{geometry}
\DeclareMathOperator*{\argmin}{argmin}
\DeclareMathOperator*{\argmax}{argmax}
\DeclareMathOperator*{\sign}{sign}

\newcommand{\toP}{\stackrel{\mathbb{P}}{\longrightarrow}}
\newcommand{\simP}{\stackrel{\mathbb{P}}{\sim}}

\newcommand{\tod}{\overset{d}{\longrightarrow}}

\newcommand{\betaml}{\hat\beta_{{\rm{ml}}}}
\newcommand{\betamlj}{\hat\beta_{{\rm{ml}},j}}

\newcommand{\betarmap}{\hat\beta^{(r)}_{{\rm{map}}}}
\newcommand{\betacmap}{\hat\beta^{{\rm{c}}}_{{\rm{map}}}}
\newcommand{\betasmap}{\hat\beta^{{\rm{s}}}_{{\rm{map}}}}
\newcommand{\betasmapj}{\hat\beta^{{\rm{s}}}_{{\rm{map}},j}}

\newcommand{\dd}{{\rm{d}}}

\newcommand{\vMFS}{\rm{vMF}_{\rm{S}}}
\newcommand{\vMFB}{\rm{vMF}_{\rm{B}}}

\newcommand{\Au}{$\mathbf{(A_1)}$}
\newcommand{\Ad}{$\mathbf{(A_2)}$}
\newcommand{\At}{$\mathbf{(A_3)}$}
\newcommand{\MM}{$\mathbf{(A_0)}$}

\newcommand{\bigOP}{\mathcal{O}_{\mathbb{P}}}

\definecolor{viri1}{rgb}{0.26562500, 0.00390625, 0.32812500}
\definecolor{viri2}{rgb}{0.1914062, 0.4062500, 0.5546875}
\definecolor{viri3}{rgb}{0.2070312 ,0.7148438 ,0.4726562}
\definecolor{viri4}{rgb}{0.9882812 ,0.9023438 ,0.1445312}
\definecolor{forestgreen}{rgb}{0.13, 0.55, 0.13}

\usepackage{xparse}
\let\realItem\item % save a copy of the original item
\makeatletter
\NewDocumentCommand\myItem{ o }{%
   \IfNoValueTF{#1}%
      {\realItem}% add an item
      {\realItem[#1]\def\@currentlabel{#1}}% add an item and update label
}
\makeatother

\usepackage{enumitem}    
\setlist[enumerate]{
    before=\let\item\myItem,       % use \myItem in enumerate
    label=\textnormal{(\arabic*)}, % format the label
    widest=(2')                    % set the widest label
}

\usepackage{xspace}
\usepackage{soul}

\definecolor{forestgreen}{rgb}{0.13, 0.55, 0.13}

\def\similarity{R\xspace}

\usepackage{xcolor}
\definecolor{C1}{HTML}{3182ce} % blue
\definecolor{C2}{HTML}{dd6b20} % orange
\definecolor{Fushia}{HTML}{8C368C} % Fushia
\definecolor{C4}{HTML}{008B72} % 
\definecolor{C5}{HTML}{792500} % brown
\definecolor{ForestGreen}{RGB}{34, 139, 34}
\definecolor{bananayellow}{rgb}{0.94, 0.88, 0.19}

\begin{document}

%--------------------------------------------------------------------------
\title{Shrinkage for Extreme Partial Least-Squares}

\author{Julyan Arbel$^{(1)}$, Stéphane Girard$^{(1,\star)}$ \& Hadrien Lorenzo$^{(2)}$}

\date{\small $^{(1)}$ Univ. Grenoble Alpes, Inria, CNRS, Grenoble INP, LJK, 38000 Grenoble, France. \\
$^{(2)}$ {Aix Marseille Univ, CNRS, I2M, Marseille, France.}\\
$^\star$ Corresponding author,
{\tt stephane.girard@inria.fr}
}
\maketitle

% \begin{abstract}

\begin{center}
    \textbf{Abstract}
\end{center}

This work focuses on dimension-reduction techniques for modelling conditional extreme values. Specifically, we investigate the idea that extreme values of a response variable can be explained by nonlinear functions derived from linear projections of an input random vector. In this context, the estimation of projection directions is examined, as approached by the Extreme Partial Least Squares (EPLS) method--an adaptation of the original Partial Least Squares (PLS) method tailored to the extreme-value framework. 
Further, a novel interpretation of EPLS directions as maximum likelihood estimators is introduced, utilizing the von Mises--Fisher distribution applied to hyperballs. The dimension reduction process is enhanced through the Bayesian paradigm, enabling the incorporation of prior information into the projection direction estimation. The maximum a posteriori estimator is derived in two specific cases, elucidating it as a regularization or shrinkage of the EPLS estimator. We also establish its asymptotic behavior as the sample size approaches infinity. 
A simulation data study is conducted in order to assess the practical utility of our proposed method. This clearly demonstrates its effectiveness even in moderate data problems within high-dimensional settings.
Furthermore, we provide an illustrative example of the method's applicability using French farm income data, highlighting its efficacy in real-world scenarios.
\\

%\begin{keywords} 


% \end{abstract}

%%%%%%%%%%%%%%%%%%%%%%%%%%%%%%%%%%%%%%%%%%%%%%%%%%%%%%%%%%%%%%%%%%%%%%%%%%%%%%%%%%%%%%%%%%%%%%

\section{Introduction}
%%%%%%%%%%%%%%%%%%%%%%%%%%%%%%%%%%%%%%%%%%%%%%%%%%%%%%%%%%%%%%%%%%%%%%%%%%%%%%%%%%%%%%%%%%%%%%

\paragraph{Partial Least Squares (PLS).}
In modern statistical regression situations, one has to deal with problems where the dimension~$p$ of the covariates $X$ is large,
and where the size $n$ of the dataset is insufficient to provide reliable estimations. Using standard (parametric or nonparametric)
regression techniques in such situations may yield overfitting and therefore unstable estimations.
This curse of dimensionality~\citep{Geenens:2011} may be mitigated 
by identifying a low-dimensional subspace of the covariates $X$ that maintains a strong link between the projected covariates and the response variable $Y$. 
As an example, Partial Least Squares (PLS) regression~\citep{wold1975soft} aims at estimating linear combinations of $X$ coordinates having a high covariance with $Y$. Even though PLS has been initially developed within the chemometrics field~\citep{martens1992multivariate},
it has also received considerable attention in the statistical literature, see for instance~\citet{naik}.
Sliced Inverse Regression \citep[SIR,][]{Li1991} is an alternative method to estimate a so-called central dimension reduction subspace
based on an inverse regression model, {\it i.e.} when $X$ is written as a function of $Y$. Several extensions have been developed for PLS and SIR, see~\citet{cook2013envelopes,li2007partial} and~\citet{chiancone2017student,coudret2014new,portier} among others or~\citet{review} for a review. 
While the above-mentioned methods adopt the frequentist point of view, there also exist a number of works in the literature 
based on Bayesian approaches. In~\citet{reich2011sufficient}, the authors model the response variable $Y$ in terms of the predictors $X$ using a mixture model whose parameters are estimated with a Markov chain Monte Carlo (MCMC) procedure.
The converse point of view is adopted in~\citet{mao2010supervised}: $X$ is modelled as a function of $Y$ thanks to an inverse mixture model,
the estimation also requiring an MCMC method. A similar approach is proposed in~\citet{cai2021bayesian} using a Bayesian inverse regression through Gaussian processes and MCMC procedures.

\paragraph{Extreme Partial Least Squares (EPLS).}
The curse of dimensionality is exacerbated when modelling conditional extremes since tail events are rare by nature.
Nonparametric estimators of extreme conditional features~\citep{Bernoulli,Suite,  Annals} are thus impacted both by the scarcity of extremes and the high dimensional setting. 
Recently, some works have introduced dimension-reduction tools dedicated to conditional extremes. 
One can mention~\citet{sabourin,Gardes2018} who propose extreme analogues of the central dimension reduction subspace. In~\citet{wang2020extreme}, a semi-parametric approach is introduced for the estimation of extreme conditional quantiles based on a tail single-index model. The dimension reduction direction is estimated by fitting a misspecified linear quantile regression model.
 Extreme Partial Least Squares \citep[EPLS,][]{ExtremePLS2022} is a dimension reduction method relying on PLS 
principles for estimating the linear combinations of $X$ that best explain the extreme values of $Y$. 
See also \cite{Cambyse} for an adaptation of EPLS to functional covariates.

\paragraph{Shrinkage EPLS, contributions, and outline.}
In this work, we develop two shrinkage versions of the EPLS method for high-dimensional settings under the common acronym SEPaLS. 
The starting point consists of recognizing the EPLS estimator as a maximum likelihood estimator associated with a von Mises--Fisher likelihood (Section~\ref{sec-EPLS}).
The latter distribution, which naturally arises for modelling directional data distributed on the unit sphere~\citep{mardia2009directional}, is here adapted to hyperballs. 
Two prior distributions are introduced on the dimension reduction direction in
Section~\ref{sec-BEPLS}: a conjugate one based on the von Mises--Fisher distribution
 and a second one using the Laplace distribution (both defined on the unit sphere) to enforce sparsity.
Proposition~\ref{prop:conj} and Proposition~\ref{prop:sparse} show that the maximum a posteriori (MAP) estimator is available in closed form. Its computation does not require MCMC methods and can be interpreted as a shrinkage version of the
initial EPLS estimator. See Figure~\ref{fig:cadre} for a summary of the different PLS adaptations.
\begin{figure}
    \centering
\begin{tcolorbox}[colback=black!3!white,colframe=gray,boxrule=1pt,width=.83\linewidth, left=0cm]
    \begin{tabular}{rcl}
    \textbf{PLS}  & $\hat\beta$ & maximizes covariance between $\langle \beta,X\rangle$ and $Y$ \\
    \textbf{\textcolor{C1}{E}PLS}  & $\textcolor{C1}{\betaml(y)}$ & a PLS estimator for values of $Y$ larger than $y$ \\
    \multirow{2}{*}{\textbf{\textcolor{Fushia}{S}\textcolor{C1}{E}P\textcolor{Fushia}{a}LS}} 
    & $\textcolor{Fushia}{\betacmap(y)}$ & an EPLS estimator with conjugate prior \\
     & $\textcolor{Fushia}{\betasmap(y)}$ & an EPLS estimator with sparse prior
\end{tabular}    
\end{tcolorbox}
    \caption{Different Partial Least Squares approaches discussed here with their adaptations to the extreme and shrinkage frameworks.}
    \label{fig:cadre}
\end{figure}
Convergence results are also established
when the sample size tends to infinity, in Proposition~\ref{prop-recall},  Proposition~\ref{prop-asymp-conj}, and Proposition~\ref{prop-asymp-sparse}. 
The behavior of the two proposed estimators is illustrated on simulated data in Section~\ref{sec-simBEPLS}, while an application on French farm income data is described in Section~\ref{sec-appBEPLS} to assess the influence of various parameters on field-grown carrot production. The functions to compute Shrinkage Extreme Partial Least Squares  estimators are available in the R package \texttt{SEPaLS}\footnote{\label{note:SEPaLS}\texttt{https://github.com/hlorenzo/SEPaLS/}} \citep{SEPaLS-package}, while the R code replicating the figures can be found online\footnote{\label{note:SEPaLS_simus}\texttt{https://github.com/hlorenzo/SEPaLS\_simus/}}.
A discussion is provided in Section~\ref{sec-discBEPLS} and proofs are postponed to Appendix~\ref{sec-appendBEPLS}.

%%%%%%%%%%%%%%%%%%%%%%%%%%%%%%%%%%%%%%%%%%%%%%%%%%%%%%%%%%%%%%%%%%%%%%%%%%%%%%%%%%%%%%%%%%%%%%
\section{Extreme Partial Least Squares without shrinkage}
\label{sec-EPLS}
%%%%%%%%%%%%%%%%%%%%%%%%%%%%%%%%%%%%%%%%%%%%%%%%%%%%%%%%%%%%%%%%%%%%%%%%%%%%%%%%%%%%%%%%%%%%%%

Throughout, $\langle\cdot,\cdot\rangle$ is the Euclidean scalar product on ${\mathbb R}^p$, $\|\cdot\|_2$ is the corresponding quadratic norm and
$S^{p-1}=\{x\in \mathbb{R}^{p}, \|x\|_2=1 \}$ is the associated unit sphere.
Moreover, for any set $\{z_1,\dots,z_n\}$, $z_{1:n}$ denotes the vector  $(z_1^\top,\dots,z_n^\top)^\top$.
Plus, two sequences of random variables $(A_n)$ and $(B_n)$ (where $(B_n)$ is almost surely non-zero) are equivalent in probability if  $A_n/B_n\toP 1$ which is denoted by $A_n \simP B_n$. Also, we write $A_n=o_{\mathbb{P}}(B_n)$ if $A_n/B_n \toP 0$. 

We first recall in Subsection~\ref{par-model} the derivation of the EPLS estimator from a statistical regression model
and, in Subsection~\ref{par-ev}, the extreme-value assumptions necessary to establish its asymptotic properties.
Subsection~\ref{par-vmf} is dedicated to the presentation of the von Mises--Fisher distribution on the sphere and to its
adaptation to hyperballs. Based on these, we then reinterpret the EPLS direction as a maximum likelihood estimator
and derive its asymptotic properties in Subsection~\ref{par-mle}.

\subsection{EPLS model}
\label{par-model}

The following single-index inverse regression model is introduced in~\citet{ExtremePLS2022}:
\begin{enumerate}
    \item [\MM]\label{MM} $X= g(Y) \beta  +  \varepsilon$,
where $\beta \in S^{p-1}$ is the unknown direction which is the parameter of interest, $X$ and $\varepsilon$ are $p$-dimensional random vectors, $Y$ is a real random variable, and $g:{\mathbb R}\to {\mathbb R}$ is an unknown link function. 
\end{enumerate}
Model \ref{MM}~is referred to as an inverse regression model since the covariates $X$ are written as functions of the response variable $Y$, see~\cite{BGG2009,Cook2007} for similar inverse models in the SIR framework.
Under model \ref{MM}, if the distribution tail of $\varepsilon$ is negligible compared to the one of $g(Y)$, then $X \simeq g(Y)\beta$ for large values of $Y$, leading to the approximate single-index forward model $Y\simeq g^{-1}(\langle\beta, X\rangle)$.
Finally, let us stress that no independence assumption is made on $(X,Y,\varepsilon)$. 
Let  $\{(X_1,Y_1),\dots,(X_n,Y_n)\}$ be an $n$ sample with same distribution as $(X,Y)$. 

\begin{definition}[EPLS estimator of the unit direction $\beta$, \citealp{ExtremePLS2022}]\label{def:EPLS}
The EPLS estimator $\hat\beta(y_n)$ of the unit direction $\beta$ is obtained by maximizing with respect to $\beta\in S^{p-1}$ the empirical covariance between $\langle \beta,X\rangle$ and $Y$ conditionally on values of $Y$ larger than $y_n$:
\begin{equation}
\label{eq-opti1}
\hat\beta(y_n)= 
\argmax_{\|\beta\|_2=1}  \langle \beta, \hat v(y_n) \rangle =\frac{\hat v(y_n)}{\|\hat v(y_n)\|_2},
\end{equation}
where, for any threshold $y_n\in \mathbb{R}$, $\hat v(y_n)$ is defined by
\begin{equation}
 \label{eq-vhat}
 \hat v(y_n)= \sum_{i=1}^n X_i \Phi_{i}(y_n,Y_{1:n}),
\end{equation}
with, for all $ i\in\{1,\dots,n\}$,
$$
\Phi_{i}(y_n,Y_{1:n}) =  \frac{1}{n}\left(\hat{\bar F}(y_n) Y_i - \hat m_{Y}(y_n) \right)\mathbf{1}{\{Y_i\geq y_n\}},
$$
 the following first-order empirical moment
$$ 
\hat m_{Y}(y_n) = \frac{1}{n}\sum_{i=1}^n  Y_i \mathbf{1}{\{Y_i\geq y_n\}},
$$
and $\hat{\bar F}$ the empirical survival function of $Y$.
\end{definition}
%See~\citet{ExtremePLS2022} for details. 
The asymptotic properties of the EPLS estimator can be established under some assumptions on the distribution tails, described hereafter.

\subsection{Extreme-value framework}
\label{par-ev}

Three assumptions on the link function $g$ and the distribution tail of $Y$ and $\varepsilon$ are considered. They rely on the notion of regularly-varying functions. 
Recall that $\varphi$ is regularly-varying with index $\theta\in\mathbb{R}$ if and only if $\varphi$ is positive and
    $$
    \lim_{y\to\infty} \frac{\varphi(ty)}{\varphi(y)}=  t^{\theta},
    $$
    for all $t>0$. We refer to~\citet{Bing1989} for a detailed account of regular variations.
\begin{enumerate}
    \item [\Au]\label{Au} The density function $f$ of $Y$ is regularly-varying of index $-{1}/{{\gamma_Y}}-1$, with $0<{\gamma_Y}<1$. %${\gamma_Y}\in (0,1)$;
    \item [\Ad]\label{Ad} The link function $g$ is regularly-varying of index $c> 0$ and  $2{\gamma_Y}(c+1)<1$.
    \item [\At]\label{At} There exists $q>1/(c{\gamma_Y} )$ such that $\mathbb{E}(\|\varepsilon\|_2^q)<\infty$.
\end{enumerate} 
Assumption \ref{Au}~implies that the survival function $\bar{F}$ is regularly-varying with index $-1/{\gamma_Y}$, which in turn is equivalent to assuming that the distribution of $Y$ is in the Fr\'echet maximum domain of attraction with positive tail-index ${\gamma_Y}$,
see \citet[Theorem~1.5.8]{Bing1989} and  \citet[Theorem~1.2.1]{Haan2007}. This domain of attraction consists of heavy-tailed distributions, such as Pareto, Burr and Student distributions, see~\citet{beigoesegteu2004} for further examples.
The larger ${\gamma_Y}$ is, the heavier the tail.
The restriction to ${\gamma_Y}<1$ ensures that the first-order moment { $\mathbb{E}(Y\mathbf{1}{\{Y\geq y\}})$ exists for all $y\geq 0$}.
Assumption \ref{Ad}~ensures that the link function $g$  ultimately behaves like a power function.
Combined with \ref{Au}, it implies that  $g(Y)$ is heavy-tailed with tail-index $\gamma_{g(Y)}:=c{\gamma_Y}$.
Finally, \ref{At}~can be interpreted as an assumption on the tail of $\|\varepsilon\|_2$. It is satisfied, for instance, by distributions with exponential-like tails such as Gaussian, Gamma or Weibull distributions. 
More specifically, 
$\mathbb{E}(\|\varepsilon\|^q)<\infty$ implies that the tail-index associated with
$\|\varepsilon\|$ is such that $\gamma_{\|\varepsilon\|}<1/q$.
 Condition~\ref{At} thus imposes that $\gamma_{g(Y)}>\gamma_{\|\varepsilon\|}$, meaning that $g(Y)$ has an heavier right tail than $\|\varepsilon\|$.
Under model~\ref{MM}, the tail behaviors of $|\beta^t X|$ and $\|X\|$ are thus driven by $g(Y)$, {\it i.e.},
$\gamma_{\|X\|}={\gamma_{g(Y)}}$,
which is the desired property.
Finally, condition $2{\gamma_Y}(c+1)<1$
implies the existence of { var$(XY\mathbf{1}{\{Y\geq y\}})$ for all $y\geq 0$}.

\subsection{Two von Mises--Fisher distributions}
\label{par-vmf}

The von Mises--Fisher distribution $\vMFS(\mu, \kappa)$ on the unit sphere $S^{p-1}$, $p\geq2$, is defined by its probability density function~\citep{watson1956construction}:
$$
    f_{\vMFS}(x|\mu,\kappa) = c_p(\kappa) \exp\left(\kappa \langle\mu, x\rangle\right) \mathbf{1}\{\|x\|_2=1\},
$$
where $\mu\in S^{p-1}$ is a location parameter and $\kappa \ge 0$ is a concentration parameter. The normalizing constant is given by:
\begin{equation}
    \label{cpk}
    c_p(\kappa)=\frac{\kappa^{p/2 -1}}{(2\pi)^{p/2} I_{p/2-1}(\kappa)} \mbox{ if } \kappa>0 \mbox{ and }
    c_p(0)=\frac{\Gamma(p/2)}{(2\pi)^{p/2}} \mbox{ otherwise},
\end{equation}
where $I_{q}(\cdot)$ is the modified Bessel function of the first kind and order $q\geq0$
defined on $\mathbb{R}_+$ by
\begin{equation}
    \label{besselsum}
  \kappa \mapsto  I_q(\kappa) = \sum_{\ell=0}^\infty \frac{1}{\Gamma(q+\ell+1)\ell!}\left(\frac{\kappa}{2}\right)^{2\ell+q},
\end{equation}
see~\citet[Chapter~9]{abramowitz1965handbook}, with ${\Gamma(\cdot)}$  the Gamma function.
The von Mises--Fisher distribution on the unit sphere is widely used in the analysis of directional data and can be considered as a spherical analogue of the multivariate Gaussian distribution~\citep{mardia1975distribution}.
Let us also recall that, for all $\mu\in S^{p-1}$,  
$\vMFS(\mu, 0)$ is the uniform distribution on the unit sphere (and thus, $c_p(0)$ coincides with the inverse of the sphere surface)
and that $\mu$ is the mode of the $\vMFS(\mu, \kappa)$ distribution for all $\kappa>0$. We propose the following adaptation of this distribution on balls:
\begin{definition}[von Mises--Fisher distribution on the ball]
\label{def-ball}
The von Mises--Fisher distribution $\vMFB(\mu, r,\kappa)$ on the $p$-dimensional ball, $p\geq2$, of radius $r>0$ is defined by its probability density function:
$$
f_{\vMFB}(x|\mu,r,\kappa) = \frac{ 2\pi c_{p+2}(\kappa)}{r^p} \exp\left(\frac{\kappa \langle\mu, x\rangle}{r}\right) \mathbf{1}\{\|x\|_2\leq r\},
$$
where $\mu\in S^{p-1}$ is a location parameter and $\kappa \ge 0$ is a concentration parameter.
\end{definition}
\noindent We refer to Lemma~\ref{vMF/Bdist} in Appendix~\ref{sec-appendBEPLS} for a proof that $f_{\vMFB}(\cdot|\mu,r,\kappa) $ integrates to one.
The next paragraph shows that the $\vMFB$~distribution plays a central role in the interpretation
of the EPLS estimator as a maximum likelihood estimator.

\subsection{Maximum likelihood estimation}
\label{par-mle}

We first prove that the EPLS estimator, initially introduced by maximizing some empirical covariance, can also be interpreted as a maximum likelihood estimator.
It is thus denoted by $\betaml(y_n)$ in the sequel.
\begin{proposition}[EPLS estimator as a maximum likelihood estimator]
 \label{prop-MLE}
 The EPLS estimator of Definition~\ref{def:EPLS} is the maximum likelihood estimator of $\beta$, denoted by $\betaml(y_n)$, in the following model:
 \begin{enumerate}
  \item[(i)]\label{uni} $X_1,\dots,X_n$ are independent and, for all $i\in\{1,\dots,n\}$, $X_i$ given $(Y_{1:n},\varepsilon_i)$
  is $\vMFB(\beta,r_i,\kappa_i)$ distributed, with location parameter $\beta$, radius $r_i=|g(Y_i)| + \|\varepsilon_i\|_2$ and concentration parameter $\kappa_i=\theta_n r_i\Phi_{i}(y_n,Y_{1:n}) $, where $\theta_n>0$ is an arbitrary parameter.
  \item[(ii)]\label{deuxi} $(Y_{1:n},\varepsilon_{1:n})$ is distributed according to some arbitrary density $p(\cdot,\cdot)$ on $\mathbb{R}^n\times\mathbb{R}^{pn}$ that does not depend on $\beta$.
 \end{enumerate}
\end{proposition}
\noindent The next proposition provides a consistency result on the EPLS maximum likelihood estimator (Definition~\ref{def:EPLS} and Proposition~\ref{prop-MLE}).
\begin{proposition}[EPLS consistency]
 \label{prop-recall}
 Assume \ref{MM}, \ref{Au}, \ref{Ad}~and \ref{At}~hold. Let $y_n\to\infty$ such that $n\bar F(y_n)\to\infty$ and $n \bar F(y_n)^{1-2/q}/g^2(y_n) \to 0$
 as $n\to\infty$. Then,
$$
\sqrt{n\bar{F}(y_n)}  \left(\betaml(y_n) - \beta \right) \toP 0.
$$
\end{proposition}
\noindent We refer to~\citet{ExtremePLS2022} for a discussion of the assumptions on the $(y_n)$ sequence. Let us simply note that the associated rate of convergence is faster than $\sqrt{n\bar F(y_n)}$. Even though the exact rate is not available there, this result will reveal sufficient for deriving the exact rates of convergence associated with the shrunk estimators, see Proposition~\ref{prop-asymp-conj} and Proposition~\ref{prop-asymp-sparse} hereafter.

%%%%%%%%%%%%%%%%%%%%%%%%%%%%%%%%%%%%%%%%%%%%%%%%%%%%%%%%%%%%%%%%%%%%%%%%%%%%%%%%%%%%%%%%%%%%%%
\section{Shrinkage for Extreme Partial Least Squares}
\label{sec-BEPLS}
%%%%%%%%%%%%%%%%%%%%%%%%%%%%%%%%%%%%%%%%%%%%%%%%%%%%%%%%%%%%%%%%%%%%%%%%%%%%%%%%%%%%%%%%%%%%%%

The result of item (i) in Proposition~\ref{prop-MLE} opens the door to the construction of shrinkage estimators for $\beta$ based on the Bayesian paradigm, referred to as Shrinkage for Extreme Partial Least Squares (SEPaLS) estimators. 
A prior distribution $\pi(\cdot)$
is introduced on the direction parameter $\beta$ and the shrinkage effect
of the maximum a posteriori (MAP) estimator is investigated.
The posterior distribution is established in 
Subsection~\ref{par-post} and MAP estimators are derived for two particular cases of priors, a conjugate one based on the von Mises--Fisher distribution on the sphere in Subsection~\ref{par-conj}, and a sparse one based on the Laplace distribution in Subsection~\ref{par-sparse}.

\subsection{Posterior distribution}
\label{par-post}

Combining Bayes' rule with Proposition~\ref{prop-MLE} 
makes it possible to derive the posterior distribution of $\beta$. See Appendix~\ref{sec-appendBEPLS} for a detailed proof.

\begin{proposition}[SEPaLS posterior distribution]
 \label{prop-post}
Let $\theta_n>0$ and $\pi(\cdot)$ a prior distribution on the direction parameter $\beta\in S^{p-1}$.
 Then, under the model~\ref{uni}, \ref{deuxi} of Proposition~\ref{prop-MLE},
 the posterior distribution of $\beta$ is given by
 $$
p(\beta|X_{1:n},Y_{1:n},\varepsilon_{1:n}) \propto 
\pi(\beta) \exp \left(   K_n \langle \beta, \betaml(y_n)\rangle \right),
$$
where we set $K_n:=\theta_n \|\hat v(y_n)\|_2 $.
\end{proposition}
\noindent The mode of the above posterior distribution is referred to as the SEPaLS estimator in the sequel.
Its existence is ensured as soon as $\pi(\cdot)$ is continuous on $S^{p-1}$, since a continuous function on a compact domain attains its maximum value within that domain.
We focus on the computation of the SEPaLS estimator for two particular choices of $\pi(\cdot)$ described in the next two subsections.

\subsection{Conjugate $\vMFS$ prior}
\label{par-conj}

We first assume a $\vMFS$~prior distribution for the direction $\beta\in S^{p-1}$, with location parameter $\mu_0 \in S^{p-1}$ and concentration parameter $\kappa_0 \geq 0$. The unit vector $\mu_0$ can be interpreted as
a prior on $\beta$ while $\kappa_0$ is the confidence level on this prior.
A graphical representation in dimension $p=3$ of the density isocontours associated with this distribution is provided on the top of Figure~\ref{Fig-sphereConjugate} for $\mu_0=(1,0,0)^\top$ and  $\kappa_0\in\{0,1,10\}$. On the leftmost panel, the density is uniform on the unit sphere, and it becomes more peaked around $(1,0,0)^\top$ as $\kappa_0$ increases.
Proposition~\ref{prop-post}
entails that the posterior distribution is written for any $\beta\in S^{p-1}$ as:
$$
p(\beta|X_{1:n},Y_{1:n},\varepsilon_{1:n}) \propto  \exp\left(\langle\beta,    K_n \betaml(y_n) + \kappa_0 \mu_0  \rangle \right),
$$
which is still a $\vMFS$~distribution.
As expected, since the von Mises--Fisher distribution belongs to the exponential family, 
considering the associated conjugate prior for $\beta$ yields a posterior distribution of the same type~\citep{nunez2005bayesian,taghia2014bayesian}.
The following proposition is easily derived.
\begin{proposition}[MAP with conjugate prior]
\label{prop:conj}
Let $\theta_n>0$, $K_n:=\theta_n \|\hat v(y_n)\|_2 $ and set $\pi:= \vMFS(\mu_0,\kappa_0)$, with $\mu_0 \in S^{p-1}$ and $\kappa_0 \geq 0$, as prior distribution on $\beta$.
 Then, under the model~\ref{uni}, \ref{deuxi} of Proposition~\ref{prop-MLE},
 the posterior distribution of $\beta$ is given by
 $$
\beta| X_{1:n},Y_{1:n},\varepsilon_{1:n} \sim \vMFS(\mu_n,\kappa_n),
 $$
  with location parameter $\mu_n$ equal to the MAP estimator,
$$
 \mu_n=\betacmap(y_n)=  \frac{ K_n \betaml(y_n) + \kappa_0 \mu_0} {\| K_n\betaml(y_n) + \kappa_0 \mu_0 \|_2},
$$
and concentration parameter $\kappa_n= \| K_n \betaml(y_n) + \kappa_0 \mu_0\|_2$. 
\end{proposition}
\noindent In this conjugate framework, the computation of the MAP estimator is 
straightforward since the mode of the $\vMFS$~distribution coincides with the location parameter: $\betacmap(y_n)$ is
a linear combination of the prior direction $\mu_0$ with the EPLS estimator $\betaml(y_n)$. 
Letting $\kappa_0\to\infty$ yields $\betacmap(y_n)\to\mu_0$, the EPLS estimator is shrunk towards the prior direction.
In contrast, setting $\kappa_0=0$ amounts to assuming a uniform prior distribution for the direction $\beta$ and we thus recover the EPLS framework.  
This behavior is illustrated on the bottom panel of Figure~\ref{Fig-sphereConjugate}
 with $\betaml \propto (3/2,-1, 1/2)^\top$ and $K_n=1$.

We show in the next proposition that a similar situation arises when $K_n \simP c \sqrt{n \bar F(y_n)}\to\infty$
(where $c>0$) and the rate of convergence of $\betacmap(y_n)$ to $\beta$ is provided.

\begin{proposition}[MAP consistency under conjugate prior]
 \label{prop-asymp-conj}
Under the assumptions of Proposition~\ref{prop-recall}, let $c>0$ and
 $$
 \theta_n \simP c \sqrt{n \bar F(y_n)} / \|\hat v(y_n)\|_2,
 $$
 as $n\to\infty$,
then, 
 $$
 \sqrt{n\bar{F}(y_n)}   \left(\betacmap(y_n)-\beta\right) \toP (\kappa_0/c) \, P_{\beta}^\perp(\mu_0),
 $$
 where $P_{\beta}^\perp(\mu_0):=\mu_0-\langle\mu_0,\beta\rangle\beta$ denotes the projection of $\mu_0$ on the hyperplane orthogonal to $\beta$.
\end{proposition}
It appears that $\betacmap(y_n)$
converges to $\beta$ at the $\sqrt{n\bar{F}(y_n)}$ rate 
which is the classical convergence rate of most of extreme-value estimators since $n\bar{F}(y_n)$
is the effective number of tail observations involved in the estimator.
The MAP estimator can however reach a faster convergence rate  when $P_{\beta}^\perp(\mu_0)=0$
{\it i.e.} when $\mu_0=\beta$, meaning that the prior distribution is centred on the true (unknown) direction.

\begin{figure}
\begin{subfigure}{0.45\textwidth}
  \centering
  \includegraphics[width=\textwidth]{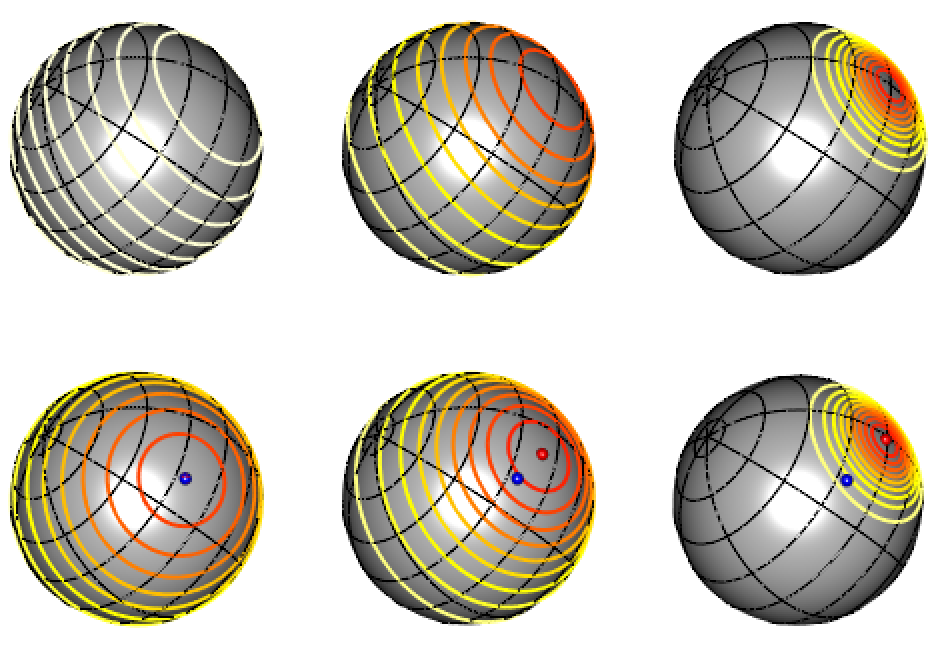}
  $\kappa_0 \approx 0$ \hspace{.12\textwidth} $\kappa_0 = 1$ \hspace{.12\textwidth}  $\kappa_0 = 10$
  \caption{Conjugate $\vMFS(\mu_0,\kappa_0)$ prior.}
  \label{Fig-sphereConjugate}
\end{subfigure}%
\hspace{.08\textwidth}
\begin{subfigure}{0.45\textwidth}
  \centering
  \includegraphics[width=\textwidth]{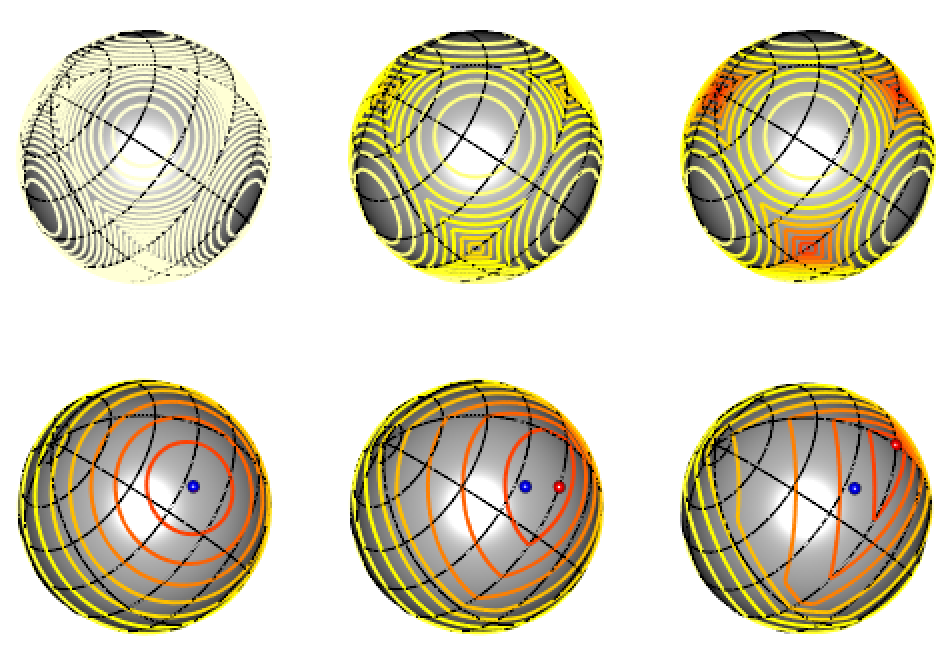}
  $\lambda \approx 0$ \hspace{.12\textwidth} $\lambda = 0.3$ \hspace{.12\textwidth}  $\lambda = 0.6$
  \caption{Sparse Laplace$(\lambda)$ prior.}
  \label{Fig-sphereSparse}
\end{subfigure}
\caption{Isocontour plots of (a) the von Mises--Fisher $\vMFS(\mu_0,\kappa_0)$ and (b) the Laplace$(\lambda)$ prior densities (top) and of the resulting posterior density (bottom) in dimension $p=3$. The estimators  $\betaml$ and $\hat{\beta}_{\text{map}}$ are depicted by blue and red points respectively.}
\label{fig:triplette_en_double}
\end{figure}

\subsection{Sparse Laplace prior}
\label{par-sparse}

The EPLS method can be adapted to take into account the information that only a few covariates in $X$ are useful to explain the extreme values of the response variable $Y$. 
 To this end,
consider a Laplace$(\lambda)$ distribution on the unit sphere: 
\begin{equation}
 \label{eq-laplace}
 \pi(\beta|\lambda)= \frac{1}{b_p(\lambda)}\exp(-\lambda \|\beta\|_1)\mathbf{1}\{\|\beta\|_2=1\}, \mbox{ with }
 b_p(\lambda) = \int_{\|x\|_2=1} \exp(-\lambda \|x\|_1)\dd x
\end{equation}
 as a prior for $\beta\in S^{p-1}$, where $\lambda \geq 0$ is a concentration parameter. We refer to~\citet{tibshirani1996regression}
 for the introduction of the Laplace prior in the regression context and to~\citet{chung2010sparse,vidaurre} for sparse versions of PLS in a non-extreme context.
 A graphical representation of the density isocontours of the Laplace distribution in dimension $p=3$ is provided on the top of Figure~\ref{Fig-sphereSparse} for $\lambda\in\{0,0.3,0.6\}$. On the leftmost panel, the density is nearly uniform on the unit sphere, and it becomes more peaked around the three vertices $(1,0,0)^\top$, $(0,1,0)^\top$ and $(0,0,1)^\top$ as $\lambda$ increases.
 
As a consequence of Proposition~\ref{prop-post}, the posterior distribution can be written as
\begin{equation}
\label{posteriorsparse}
p(\beta|X_{1:n},Y_{1:n},\varepsilon_{1:n}) \propto
 \exp\left(K_n\langle\beta, \betaml(y_n)\rangle -\lambda \|\beta\|_1  \right),
\end{equation}
for any $\beta\in S^{p-1}$. Although this posterior distribution does not correspond to a classical distribution on the unit sphere, the MAP can be computed in closed form:
\begin{proposition}[MAP with sparse prior]
\label{prop:sparse}
Let $\theta_n>0$, $K_n:=\theta_n \|\hat v(y_n)\|_2 $ and set $\pi(\cdot|\lambda)$ as the Laplace prior distribution~\eqref{eq-laplace} on $\beta$.
 Then, under the model \ref{uni}, \ref{deuxi} of Proposition~\ref{prop-MLE},
 the MAP estimator of $\beta$ is:
 $$
  \betasmap(y_n) = \tilde{\beta}(y_n)/{\|\tilde{\beta}(y_n)\|_2}, \mbox{ with } \tilde{\beta}_j(y_n)= S_{\lambda}( K_n\betamlj(y_n)), \; j\in\{1,\dots,p\},
 $$
and where $S_{\lambda}(\cdot)$ is the shrinkage operator defined as $
S_{\lambda}(x)=\sign(x) \left( |x| -\lambda \right)\mathbf{1}{\{|x|>\lambda\}}$, $x\in\mathbb{R}$.
\end{proposition}
\noindent The MAP is obtained by shrinking the coordinates of $\betaml(y_n)$  associated with the EPLS estimator towards zero. 
See 
%Figure~\ref{Fig-ShrinkOperator} for an illustration of the shrinkage operator and 
Theorem~3 of \cite{chung2010sparse}
for a similar result in a non-extreme framework.
The zero coordinates in $\betasmap(y_n)$ correspond to covariates in $X$ that have no impact on the extreme values of $Y$.
Note that when the concentration parameter is set to $\lambda=0$, we recover the EPLS method. 
The behavior of the $\betasmap$ estimator is illustrated on the bottom panel of Figure~\ref{Fig-sphereSparse}
with $\betaml \propto (3/2,-1, 1/2)^\top$ and $K_n=1$. When $\lambda$ is small, both estimates $\betaml$ and $\betasmap$ are superimposed. When $\lambda$ increases, $\betasmap$ gets closer and closer to the vertex $(1,0,0)^\top$.

Similarly to the conjugate case, 
when $K_n \simP c \sqrt{n \bar F(y_n)}\to\infty$
(where $c>0$), the rate of convergence of $\betasmap(y_n)$ to $\beta$ can be established.

\begin{proposition}[MAP consistency under sparse prior]
 \label{prop-asymp-sparse}
Under the assumptions of Proposition~\ref{prop-recall}, let $c>0$ and
 $$
 \theta_n \simP c \sqrt{n \bar F(y_n)} / \|\hat v(y_n)\|_2,
 $$
 as $n\to\infty$, then, for all $j\in\{1,\dots,p\}$ such that $\beta_j\neq 0$,
$$
  \sqrt{n\bar{F}(y_n)} 
\left(\betasmapj(y_n)-\beta_j\right) \toP  (\lambda/c) \left(\|\beta\|_1 \beta_j  - \sign(\beta_j) \right) .
 $$
Otherwise, if $\beta_j=0$, then $\betasmapj(y_n)=0$ with probability tending to 1.
\end{proposition}
\noindent It appears that the null coordinates of $\beta$ are recovered with large probability thanks to the Laplace prior.
Similarly to the conjugate case,  the MAP estimator
converges to $\beta$ at the usual $\sqrt{n\bar{F}(y_n)}$ rate. The convergence rate is higher when the non-zero coordinates of $\beta$ all coincide: $\beta_j=\sign(\beta_j)/\|\beta\|_1$ for all $j\in\{1,\dots,p\}$ such that $\beta_j\neq 0$.

%%%%%%%%%%%%%%%%%%%%%%%%%%%%%%%%%%%%%%%%%%%%%%%%%%%%%%%%%%%%%%%%%%%%%%%%%%%%%%%%%%%%%%%%%%%%%%
\section{Illustration on simulated data}
\label{sec-simBEPLS}
%%%%%%%%%%%%%%%%%%%%%%%%%%%%%%%%%%%%%%%%%%%%%%%%%%%%%%%%%%%%%%%%%%%%%%%%%%%%%%%%%%%%%%%%%%%%%%

%%%%========
\subsection{Experimental design}
\label{sub-design}
%%%%========

The behavior of the SEPaLS estimators $\betacmap$ and $\betasmap$ is illustrated on the regression model~\ref{MM}
with power link function: $t>0\mapsto g(t)=t^c$, $c\in\{1,1/2,1/4\}$.
The output variable $Y$ is distributed from a Pareto distribution with survival function $\bar{F}(y)=({y}/2)^{-1/{\gamma_Y}}$, $y\geq 2$ and with tail-index ${\gamma_Y}=1/5$.
 Each margin  $\varepsilon^{(j)}$, $j\in\{1,\dots,p\}$  of the error $\varepsilon$ is simulated as the absolute value of a $\mathcal{N}(0,\sigma^2)$ random variable and depending on $Y$ using the Clayton copula, an Archimedean copula~\cite[Section~4]{Nelsen2007}, defined for all $(u,v)\in[0,1]^2$ by
$$
C_\theta(u,v)
=
\left(
u^{-\theta}+v^{-\theta}-1
\right)^{-{1}/{\theta}},
$$
where $\theta\geq 0$ is a parameter tuning the dependence between the margins. 
Equivalently, the joint cumulative distribution function of $\varepsilon$ is given for all $x\in{\mathbb R}_+^p$ by
the one-factor model~\citep{krup}:
$$
F_\varepsilon(x)=\int_0^1 \prod_{j=1}^p \frac{\partial C_\theta}{\partial v}(2\Psi(x_j/\sigma)-1,v) \dd v,
$$
where $\Psi$ denotes the cumulative distribution function of the standard Gaussian distribution.
Note that $C_0(u,v)=uv$ represents the independence copula while, as $\theta\to\infty$, $C_\theta(u,v)\to \min(u,v)$ which represents the co-monotonicity copula.
The dependence between the margins is assessed using Kendall's tau $\tau(\theta)=\theta/(\theta+2)\in[0,1)$ and is thus limited to positive values.
We shall also consider the associated rotated copula defined by $\tilde C_\theta(u,v)=v-C_\theta(1-u,v)$ whose Kendall's tau is negative and given by $\tilde \tau(\theta)=-\tau(\theta)\in(-1,0]$, for all $\theta\geq 0$.
Here, $\theta \in\{1/2,8\}$ leads to four possible values of the Kendall's tau: $\{-0.8,-0.2,0.2,0.8\}$.

The standard deviation $\sigma$ is selected such that the signal-to-noise ratio, defined as $g(\bar F^{-1}(1/n))/\sigma$, is equal to 10. Note that $g(\bar F^{-1}(1/n))$ represents the approximate maximum value of $g$ on a $n$-sample from the distribution with associated survival function~$\bar F$. 
 
The sample size is fixed to $n=500$ and two dimensions are considered: $p\in\{30,300\}$. The true direction is 
$\beta=(1,1,0,\dots,0)^\top/\sqrt{2}$ for both dimensions.
  
The location parameter $\mu_0$ of the prior $\vMFS$~distribution (conjugate case) is set either to $\beta$, which corresponds to a perfect prior, or to $\tilde\beta:= (1,\dots,1,0,\dots,0)^\top/\sqrt{p/2}$, which is far from the true one, see Subsection~\ref{par-conj}. Four values of the concentration parameter are investigated: 
$\kappa_0\in \{0,10^{-4},3.10^{-3},10^{-2}\}$.
In the case of the Laplace prior (sparse case), we let
$\lambda\in \{0,10^{-4},5.10^{-4},10^{-3}\}$.
In both situations, we set $\theta_n:=1$ since this parameter
is irrelevant to the inference.

 \subsection{Performance assessment}
 
Let us define a similarity measure $\similarity$  between the theoretical vector $\beta$ and its MAP estimator  computed on $N = 1\ \!000$ replications as follows:
\begin{equation}
\similarity(y)= \frac{1}{N}\sum_{r=1}^{N}\big\langle\betarmap(y), \beta\big\rangle^2, 
 \end{equation}
 where $\betarmap$ denotes the MAP estimate on the $r^\mathrm{th}$ replication under either the conjugate or the sparse prior.
Clearly $\similarity \,\in[0,1]$ and the closer $\similarity$ is to 1, the larger the proximity is. In practice,  $\similarity (Y_{n-k+1,n})$ is computed as a function of the number of exceedances $k\in\{1,\dots,100\}$, 
where $Y_{n-k+1,n}$ denotes the $(n-k+1)^\mathrm{th}$ largest observation from the sample $\{Y_1,\ldots,Y_n\}$.

\newcommand{\legendconjugate}[1]{Finite sample behavior of the SEPaLS estimator computed with the conjugate prior on simulated data in dimension $d=30$ from a Pareto distribution (${\gamma_Y}=1/5,\;a=2$) and a (rotated) Clayton copula with Kendall's tau $\tau\in\{-0.8, -0.2, 0.2, 0.8\}$ (from top to bottom). The power of the link function $g(t)=t^c$ is fixed to $c=#1$. Vertically: $\similarity (Y_{n-k+1,n})$ between $\betacmap$ and $\beta$ for a prior direction $\mu_0=\beta$ (left) or $\mu_0=\tilde\beta$ (right) as a function of the number $k \in \{1,\dots, 100\}$ of exceedances (horizontally). The concentration parameter is {$\kappa_0\in \{0,10^{-4},3.10^{-3},10^{-2}\}$}, respectively in violet, blue, green and yellow. Coloured areas correspond to $90\%$ confidence intervals.%, lines to medians.
}

\begin{figure}[ht]
\centering
\begin{tabular}{ccc}%{m{.2cm} m{.47\textwidth} m{.47\textwidth}}
    & \multicolumn{2}{c}{Conjugate $\vMFS$ prior and link function $g(t)=t^c$ with $c=1$.}\\
    &  $\mu_0=\beta$ &   $\mu_0=\tilde\beta$\\
    \rotatebox[origin=l]{90}{\hspace{.9cm}$\tau=-0.8$}
    &\includegraphics[trim={1cm 1.5cm 0.5cm 2cm},clip,width=.43\textwidth]{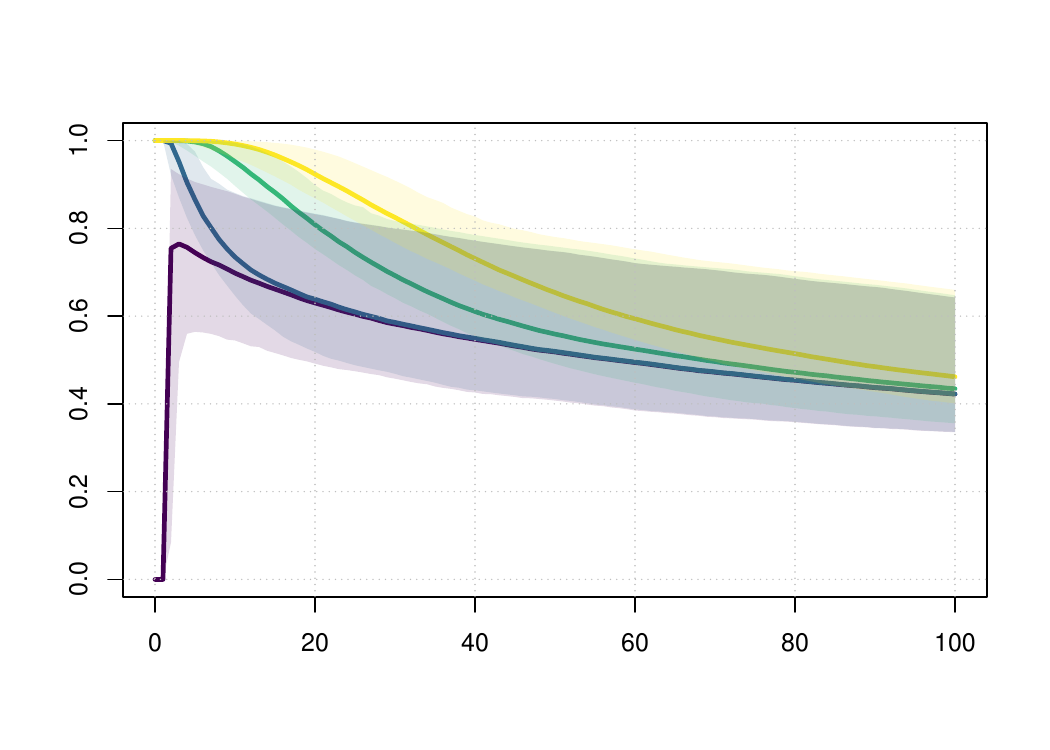}&
    \includegraphics[trim={1cm 1.5cm 0.5cm 2cm},clip,width=.43\textwidth]{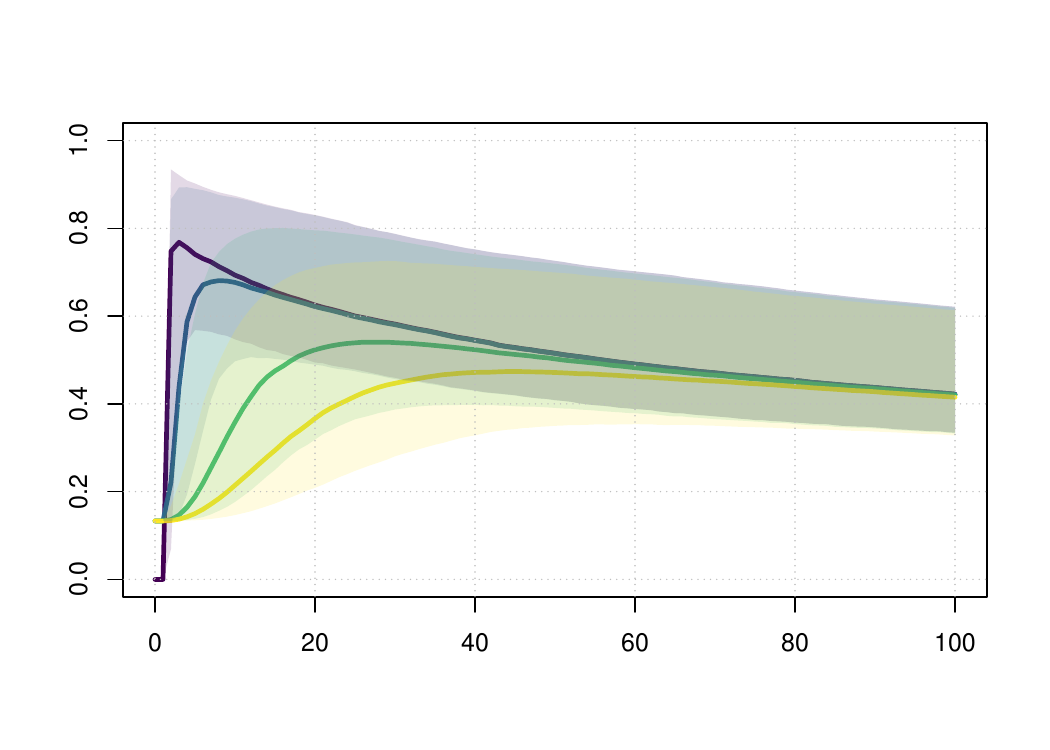}\\
    \rotatebox[origin=l]{90}{\hspace{.9cm}$\tau=-0.2$}
    &\includegraphics[trim={1cm 1.5cm 0.5cm 2cm},clip,width=.43\textwidth]{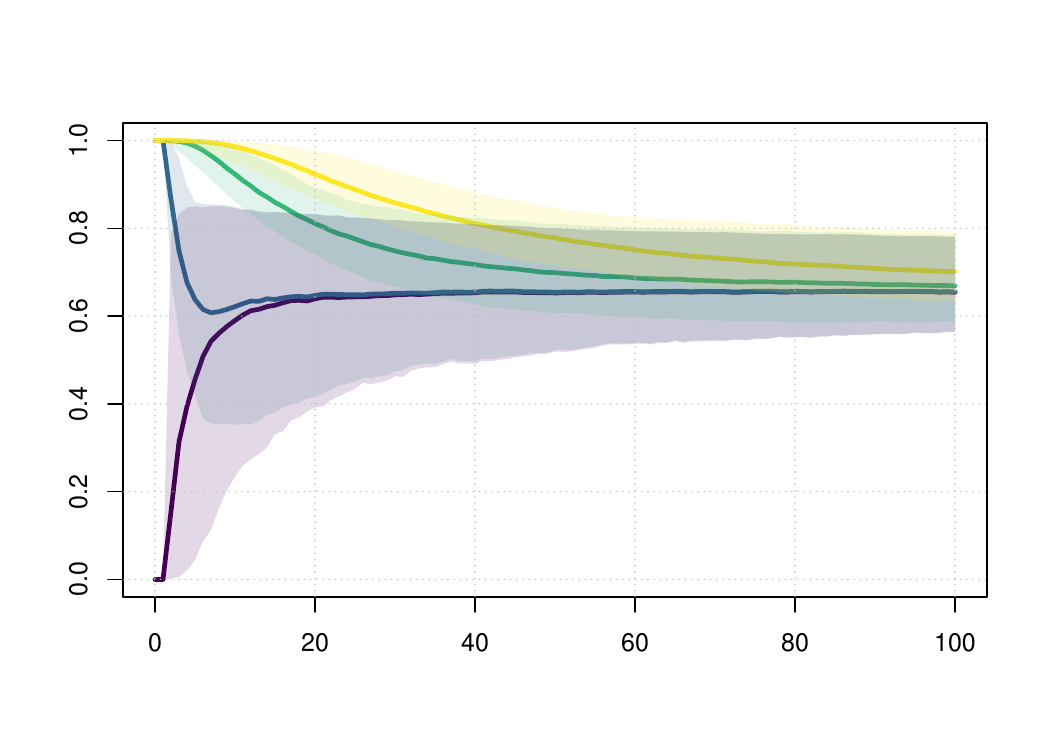}&
    \includegraphics[trim={1cm 1.5cm 0.5cm 2cm},clip,width=.43\textwidth]{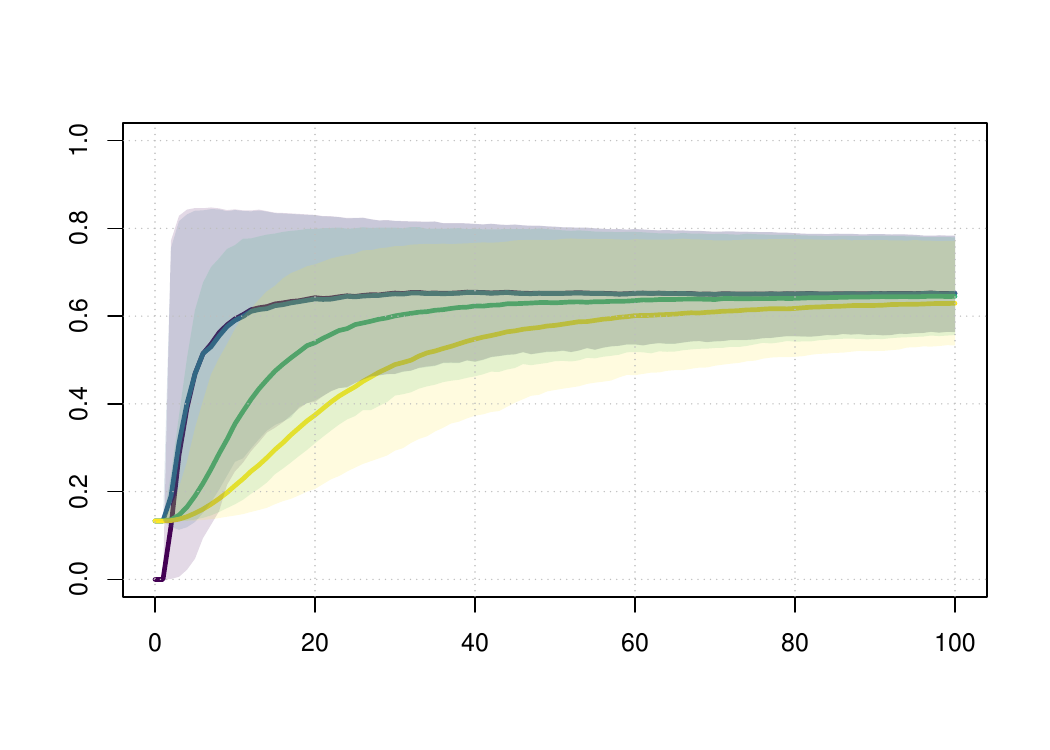}\\
    \rotatebox[origin=l]{90}{\hspace{.9cm}$\tau=0.2$}
    &\includegraphics[trim={1cm 1.5cm 0.5cm 2cm},clip,width=.43\textwidth]{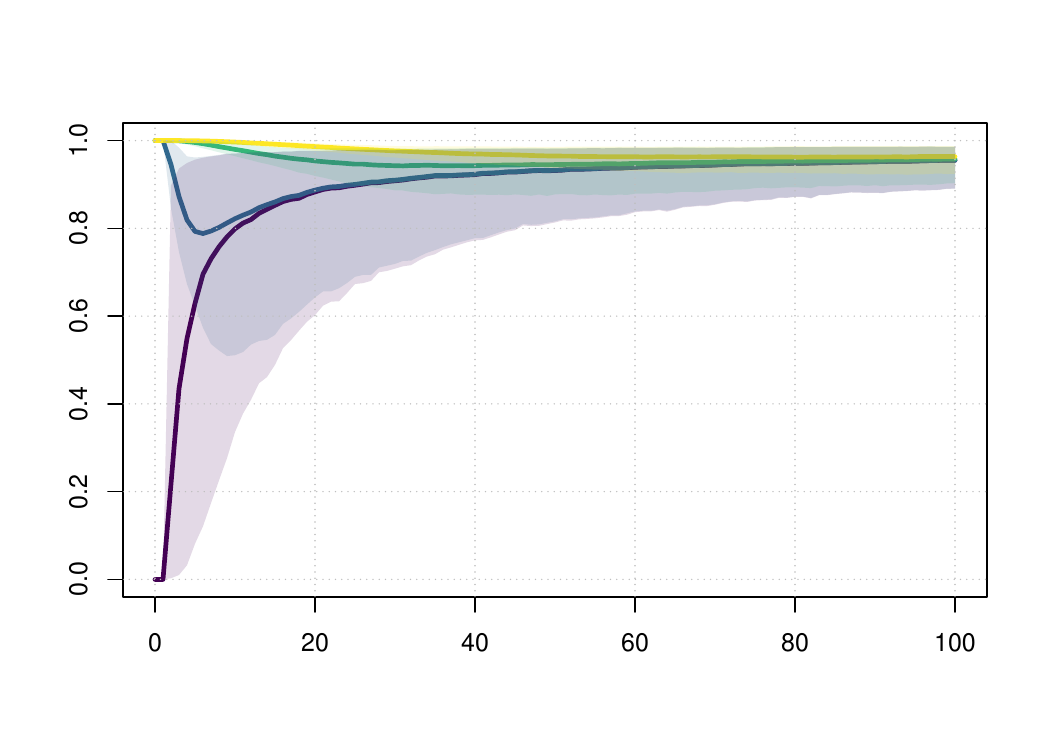}&
    \includegraphics[trim={1cm 1.5cm 0.5cm 2cm},clip,width=.43\textwidth]{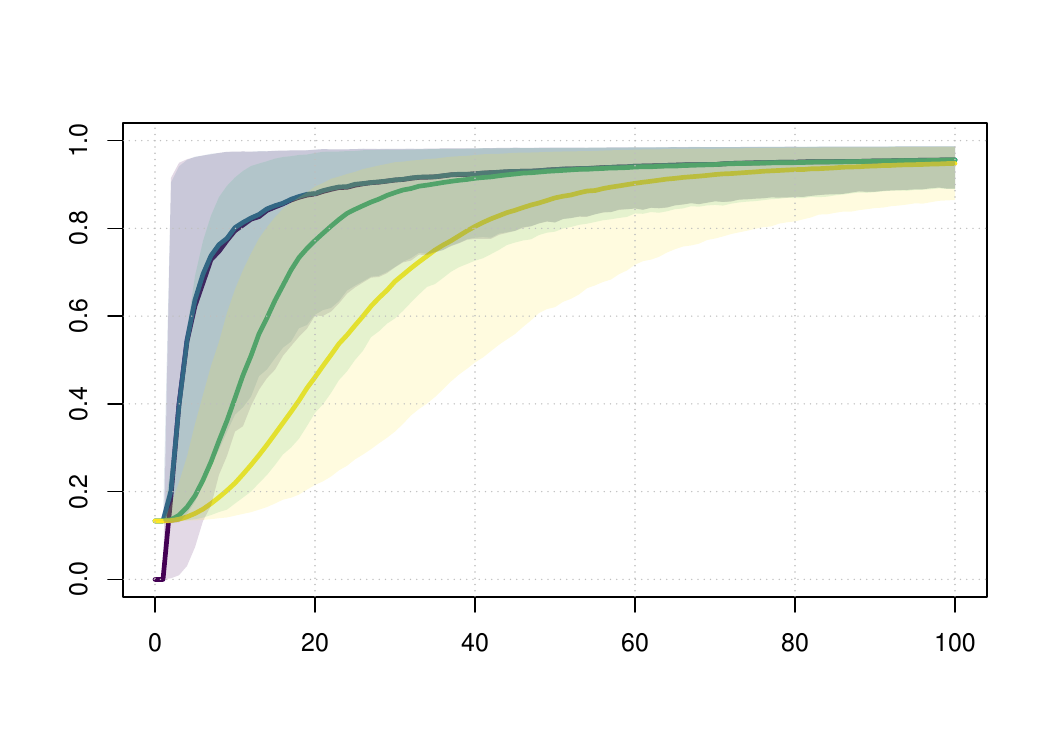}\\
    \rotatebox[origin=l]{90}{\hspace{.9cm}$\tau=0.8$}
    &\includegraphics[trim={1cm 1.5cm 0.5cm 2cm},clip,width=.43\textwidth]{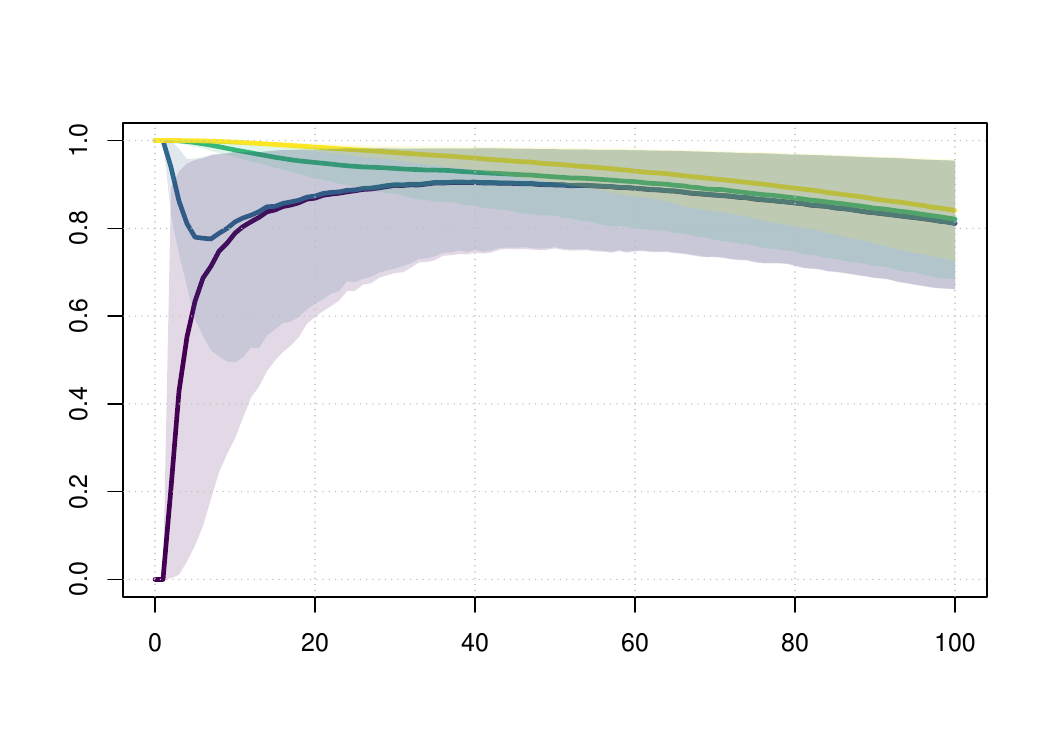}&
    \includegraphics[trim={1cm 1.5cm 0.5cm 2cm},clip,width=.43\textwidth]{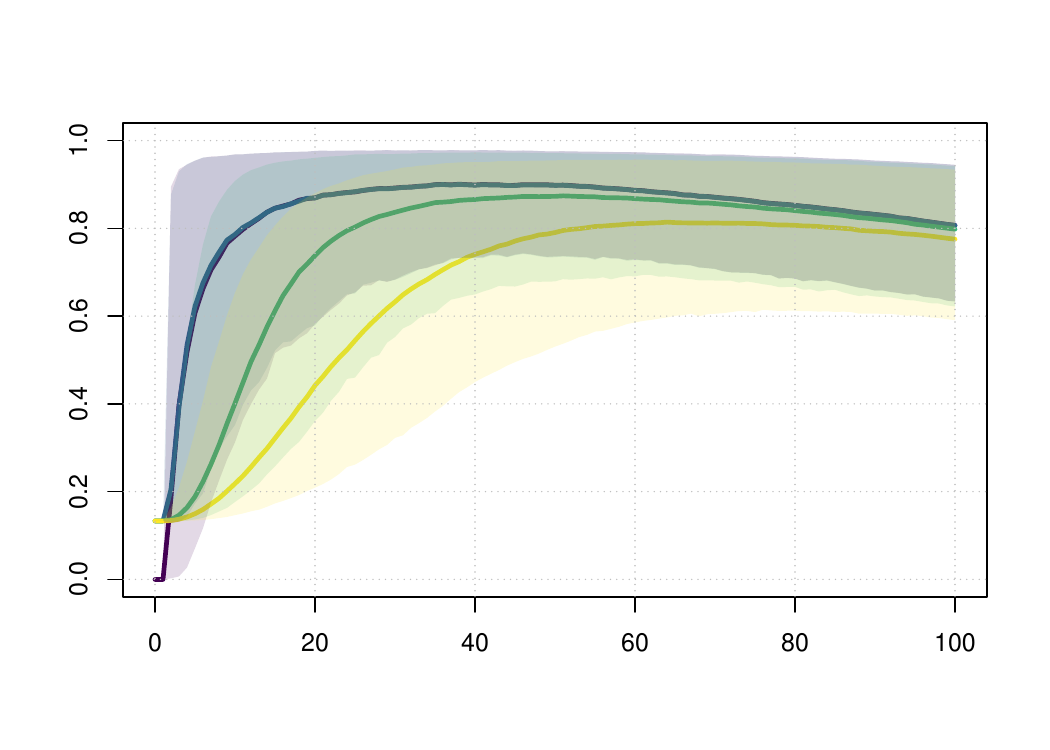}
\end{tabular}
\caption{\legendconjugate{1}}
\label{BEPLSConj1}
\end{figure}

\subsection{Results}

\textbf{Conjugate prior.}   The similarity measure $\similarity (Y_{n-k+1,n})$ between $\betacmap$ and $\beta$ is represented as a function
of $k\in\{1,\dots,100\}$ on Figure~\ref{BEPLSConj1} for the choice of parameter $c=1$. See Figure~\ref{BEPLSConj2} and Figure~\ref{BEPLSConj3} in Appendix~\ref{sec-extra-figures} for the cases $c\in\{1/2,1/4\}$. Each of these figures 
considers 32 configurations in dimension $d=30$:
$\kappa_0\in \{0,10^{-4},3.10^{-3},10^{-2}\}$, $\tau\in\{-0.8,-0.2,0.2,0.8\}$, and $\mu_0\in\{\beta,\tilde\beta\}$,
see Subsection~\ref{sub-design} for details.
Unsurprisingly, when $\mu_0=\beta$ {\it i.e.} when the prior direction points towards the true one, the shrinkage improves the results
of the original EPLS estimator (obtained when $\kappa_0=0$). Moreover, it reduces the sensitivity with respect to the number
of exceedances $k$, the dependence degree $\tau$, and the exponent $c$ of the link function. In all situations, one can obtain $\similarity \,\simeq 1$
with $\kappa_0=10^{-2}$.
In contrast, when $\mu_0=\tilde\beta$, the prior direction is ill-adapted since $\langle \tilde\beta,\beta\rangle^2=4/p\simeq 0.13$
and too large values of $\kappa_0$ deteriorate the EPLS estimator. As expected, the choice of $\mu_0$ is of primary importance in the 
conjugate prior.

%%%%%%%%%%%%%%%%%%%%%%%%%%%%%%%%%%%%%%%%%%%%%%%%%%%%%%%%%%%%%%%%%%%%%%%%%

\textbf{Sparse prior.}  Similarly, the similarity measure $\similarity (Y_{n-k+1,n})$ between $\betasmap$ and $\beta$ is represented as a function
of $k\in\{1,\dots,100\}$ on Figures~\ref{BEPLSSparse1}--\ref{BEPLSSparse3}  in Appendix~\ref{sec-extra-figures} for the cases $c\in\{1,1/2,1/4\}$. Each of these figures considers 32 configurations:
$\lambda\in \{0,10^{-4},5.10^{-4},10^{-3}\}$, $\tau\in\{-0.8,-0.2,0.2,0.8\}$, $c\in\{1,1/2,1/4\}$ and $d\in\{30,300\}$.
Here, the shrinkage always improves the results of the original EPLS estimator (obtained when $\lambda=0$) since the true direction
$\beta$ is rather sparse, it only has two non-zero coordinates. Enforcing sparsity allows to obtain $\similarity \,\simeq 0.8$ (resp. $\similarity \,\simeq 0.6$)
in dimension $p=30$ (resp. $d=300$) with exponents $c\geq 1/2$.  The case of small exponents $(c=1/4)$ appears to be more complicated, 
the maximum value of $R$ depending on the dimension $p$ and on the dependence degree $\tau$.

%%%%%%%%%%%%%%%%%%%%%%%%%%%%%%%%%%%%%%%%%%%%%%%%%%%%%%%%%%%%%%%%%%%%%%%%%

%%%%%%%%%%%%%%%%%%%%%%%%%%%%%%%%%%%%%%%%%%%%%%%%%%%%%%%%%%%%%%%%%%%%%%%%%%%%%%%%%%%%%%%%%%%%%%

\section{Application to real data}
\label{sec-appBEPLS}
%%%%%%%%%%%%%%%%%%%%%%%%%%%%%%%%%%%%%%%%%%%%%%%%%%%%%%%%%%%%%%%%%%%%%%%%%%%%%%%%%%%%%%%%%%%%%%

The SEPaLS method is illustrated on data extracted from the Farm Accountancy Data Network (FADN)\footnote{Available in French at:\\ \texttt{https://agreste.agriculture.gouv.fr/agreste-web/servicon/I.2/listeTypeServicon/}.}.
This dataset targets French farms described by numerous qualitative and quantitative variables over the period 2000--2015.
Here, we focus on the $n=598$ farms producing field-grown carrots.
The response variable $Y$ is the production of carrots (in quintals) and the covariate $X$ is made of $p=259$ continuous variables
including meteorological and economic measurements.
Our goal is to investigate, among the 259 collected factors, which ones may influence
the upper tail of $Y$, {\it i.e.} are linked to large productions of carrots. 
A similar study could be achieved on the small productions of carrots by focusing on the upper tail of $1/Y$.

Three visual checks are first carried out in Figure~\ref{fig:Histo_hill_real} to verify whether the heavy-tail hypothesis on $Y$ is realistic.
The histogram of the $\{Y_1,\dots,Y_n\}$ on the top left panel is skewed to the right and has a heavy right tail.
Besides, the Hill estimator~\citep{Hill}
$$
 \hat {\gamma}_Y(k) = \frac{1}{k} \sum_{i=1}^{k}  \log(Y_{n-i+1,n}/Y_{n-k,n})
$$
 of the tail-index ${\gamma_Y}$ is drawn on the top right panel as a function of $k\in\{1,\dots,500\}$.
 The resulting graph is stable on the range $k\in\{160,\dots,280\}$ and points towards ${\gamma_Y}\simeq 0.72$.
 Finally, selecting $k=199$ (this choice is discussed below), the associated quantile-quantile plot of the log-excesses $\log(Y_{n-i+1,n}/Y_{n-k,n})$ against the quantiles  $\log(k/i)$
 of the unit exponential distribution, $i\in\{1,\dots,k\}$, exhibits a linear trend (bottom panel) which is further empirical evidence that the heavy-tail assumption is appropriate, see \citet[pp.109--110]{beigoesegteu2004}. 

\begin{figure}
    \centering
    \begin{tabular}{cccc}
        \rotatebox[origin=l]{90}{\hspace{2cm}$\text{Frequency}$} &
        \includegraphics[width=.4\textwidth]{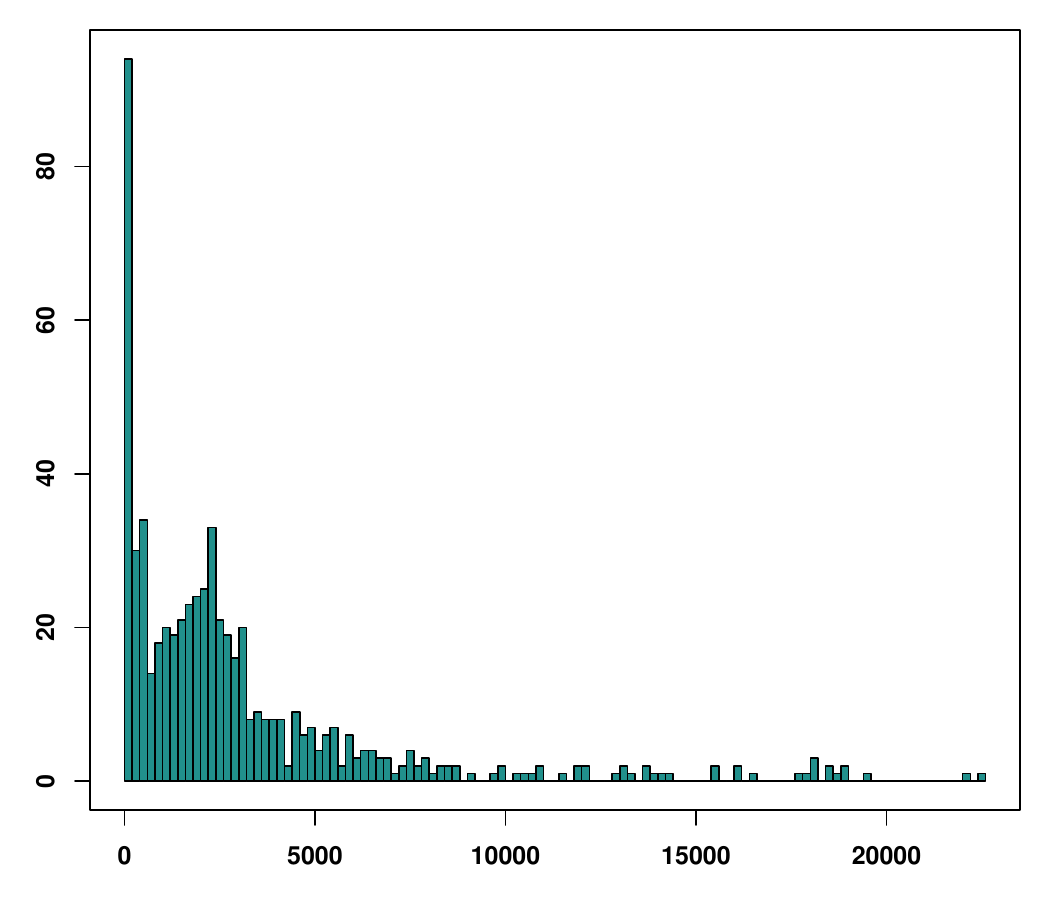} &
        \rotatebox[origin=l]{90}{\hspace{2.2cm}$\hat{\gamma}_Y(k)$} 
        &
        \includegraphics[width=.4\textwidth]{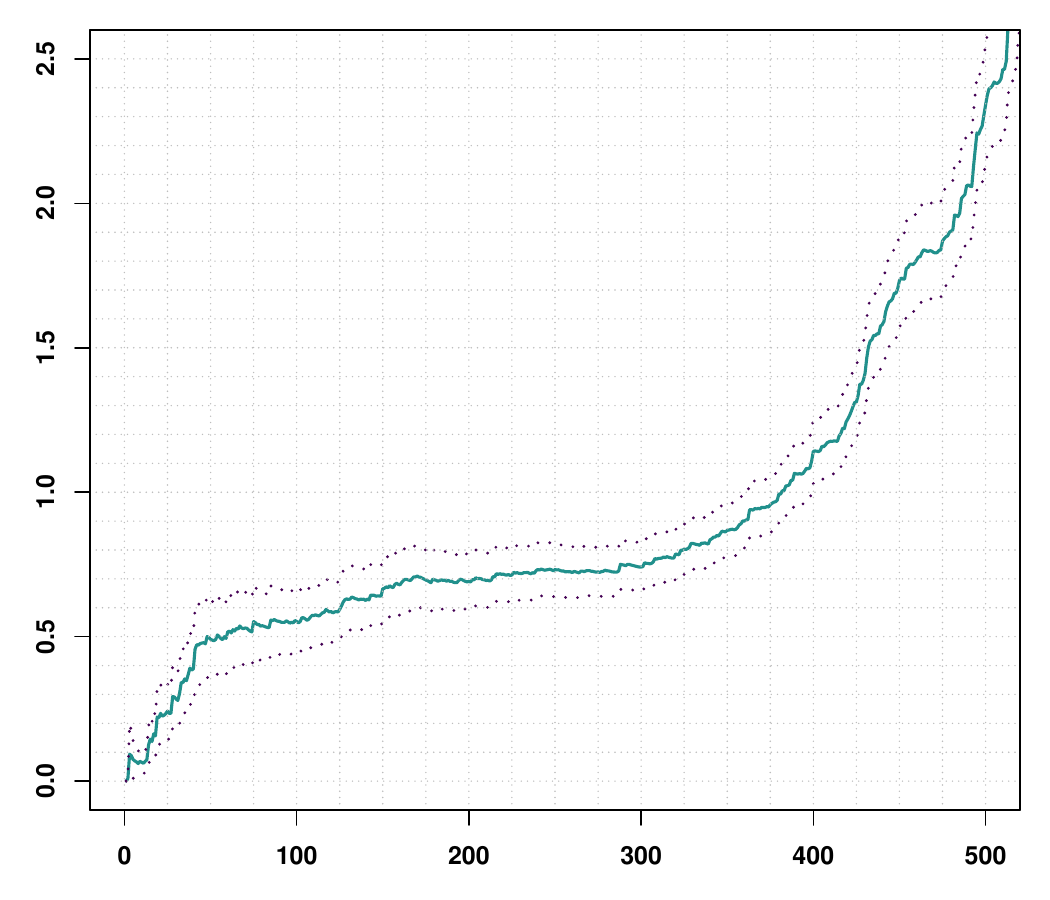}\\
        \rotatebox[origin=l]{90}{\hspace{1cm}$\log(Y_{n-i+1,n}/Y_{n-199,n})$} &
        \includegraphics[trim={1cm 1cm 0cm 0cm},clip,width=.4\textwidth]{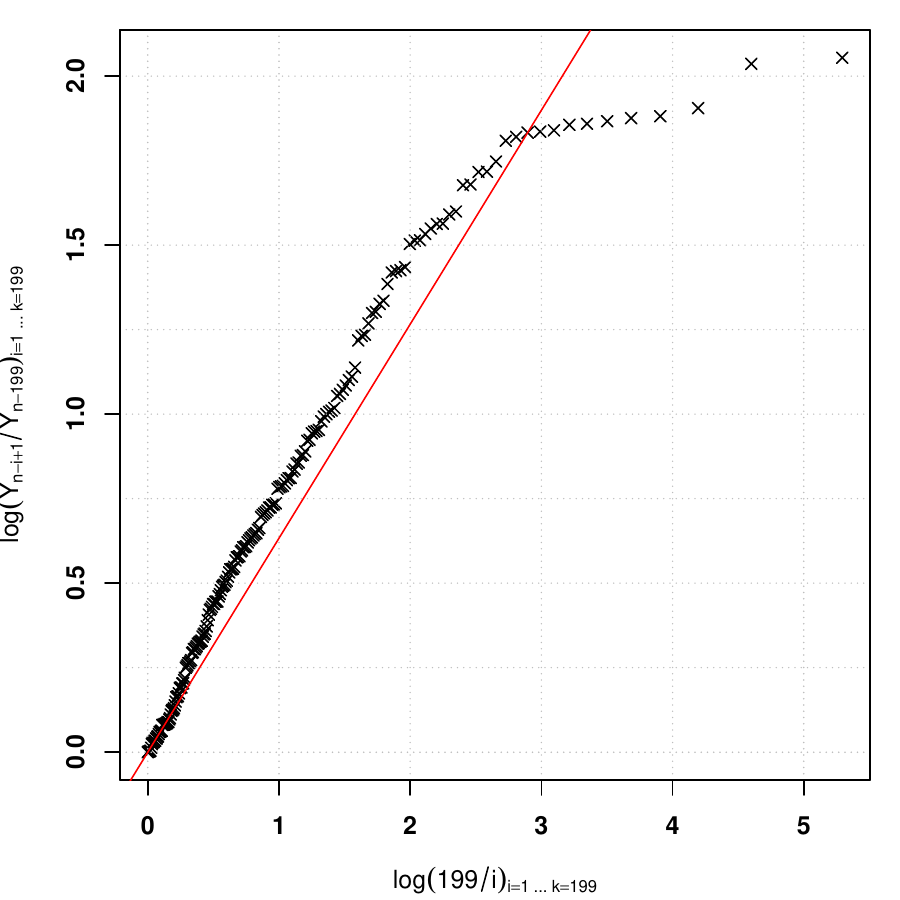} &
        \rotatebox[origin=l]{90}{\hspace{3cm}$Y_i$} &
        \includegraphics[trim={1cm 1cm 0cm 0cm},clip,width=.4\textwidth]{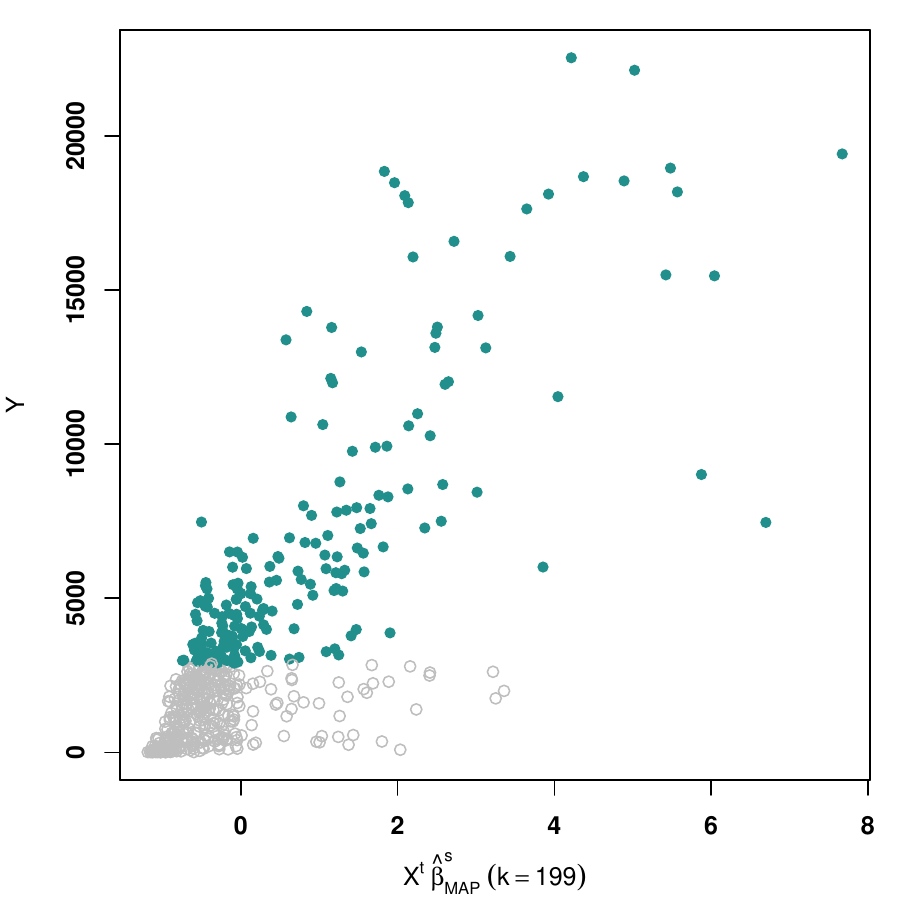}\\
        &
        $\log(199/i)$
        &
        &
        $\langle X_i, \betasmap(Y_{n-199+1,n})\rangle$
    \end{tabular}
    \caption{Real data example. Top left: Histogram of $\{Y_1,\dots,Y_n\}$. Top right: Hill plot $k\in\{1,\dots,500\}\mapsto \hat{\gamma}_Y(k)$ and associated confidence intervals (dotted lines).
    Bottom left: Quantile-quantile plot (horizontally: $\log(k/i)$, vertically: $\log(Y_{n-i+1,n}/Y_{n-k,n})$, 
    for $i \in \{1,\dots , k\}$) drawn with $k=199$, 
    the regression line is superimposed in red.
    Bottom right: Scatter-plot $(\langle X_i, \betasmap(Y_{n-k+1,n})\rangle,Y_i)$, $i\in\{1,\dots,n\}$ with $k=199$ depicted in green. Points below the threshold ($Y_i\leq Y_{n-k+1,n}$) are depicted in gray.}
    \label{fig:Histo_hill_real} % bottom right formarly {fig:regr_sparse_coeff2}
\end{figure}

In the following, we focus on the sparse estimator $\betasmap$ since the use of $\betacmap$ would require an initial guess for $\beta_0$ which is not obvious in this application context. 
The next two conditional tail correlation measures are introduced to interpret the results obtained with $\betasmap$:
\begin{align}
\label{eqcorY}
\rho( \langle X,\betasmap(y)\rangle ,Y|Y\geq y) &= \frac{{\mbox{cov}}(\langle X,\betasmap(y)\rangle ,Y|Y\geq y)}{\sigma(\langle X, \betasmap(y)\rangle |Y\geq y)\sigma(Y|Y\geq y)},\,\,\text{(see Figure \ref{Fig-corY}),}\\
\label{eqcorX}
\rho( \langle X,\betasmap(y)\rangle ,X^{(j)}|Y\geq y) &= \frac{{\mbox{cov}}(\langle X,\betasmap(y)\rangle,X^{(j)}|Y\geq y)}{\sigma(\langle X,\betasmap(y)\rangle |Y\geq y)\sigma(X^{(j)}|Y\geq y)},\,\,\text{(see Figure \ref{Fig-corX}),}
\end{align}
with $j\in\{1,\dots,p\}$.
The role of the tail correlation measure~\eqref{eqcorY} is to assess the correlation in the tail between the response variable $Y$ and the summary $\langle X,\betasmap(y)\rangle$ of the predictors built by the SEPaLS method. 
% %%%%
% \subsection{Sparse prior}
% %%%%
It is computed at the threshold $y=Y_{n-k+1,n}$ and plotted on Figure~\ref{Fig-corY} as a function of the number of exceedances $k$ for several levels of shrinkage $\lambda$.
Note that, when $k$ is small, the correlation vanishes for a wide range of $\lambda$ values since, in this case, the prior weight is too large compared to the likelihood one. 
The global maximum is located at $k=199$ which
corresponds to a stable region of the Hill estimator according to Figure~\ref{fig:Histo_hill_real}.
The maximum correlation ($\rho\simeq 0.79$) is reached at $\lambda=353$. 

\begin{figure}[ht]
    \centering
    \begin{subfigure}{0.95\textwidth}
  \centering
  \begin{tabular}{cc}%{m{.2cm} m{.9\textwidth}}
    \rotatebox[origin=l]{90}{\hspace{2cm}$\scriptstyle{\rho( \langle X, \betasmap(y)\rangle ,Y|Y\geq y)}$}&
    \includegraphics[width=.9\textwidth]{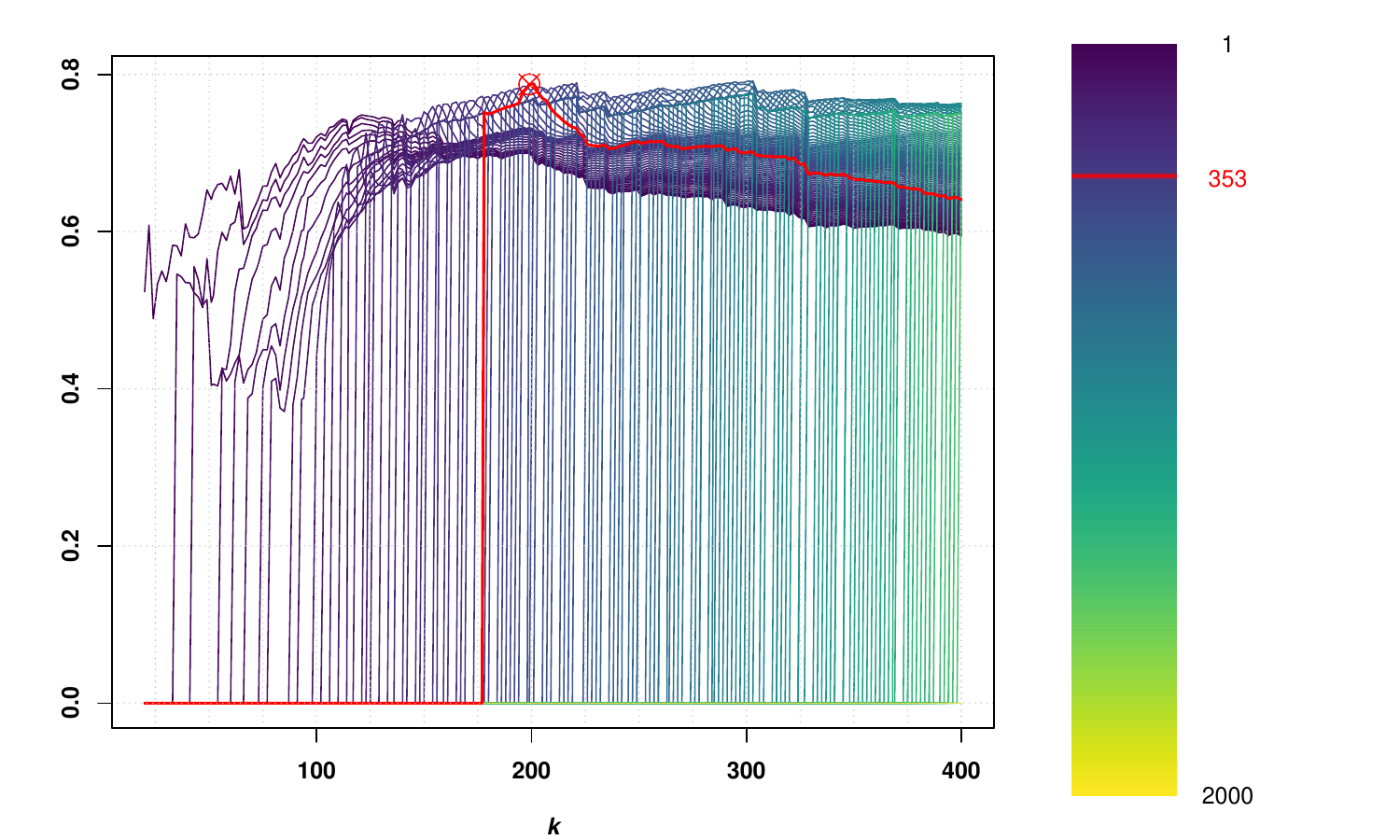}
    \end{tabular}
  \caption{Correlation $\rho( \langle X, \betasmap(y)\rangle ,Y|Y\geq y)$.}
  \label{Fig-corY}
    \end{subfigure}\\
    \begin{tabular}{cc}%{m{.2cm} m{.9\textwidth}}
    \rotatebox[origin=l]{90}{\hspace{2.2cm}$\scriptstyle{\rho( \langle X, \betasmap(y)\rangle ,X^{(j)}|Y\geq y)}$}&
    \begin{subfigure}{0.48\textwidth}
  \centering
    \includegraphics[height=6cm]{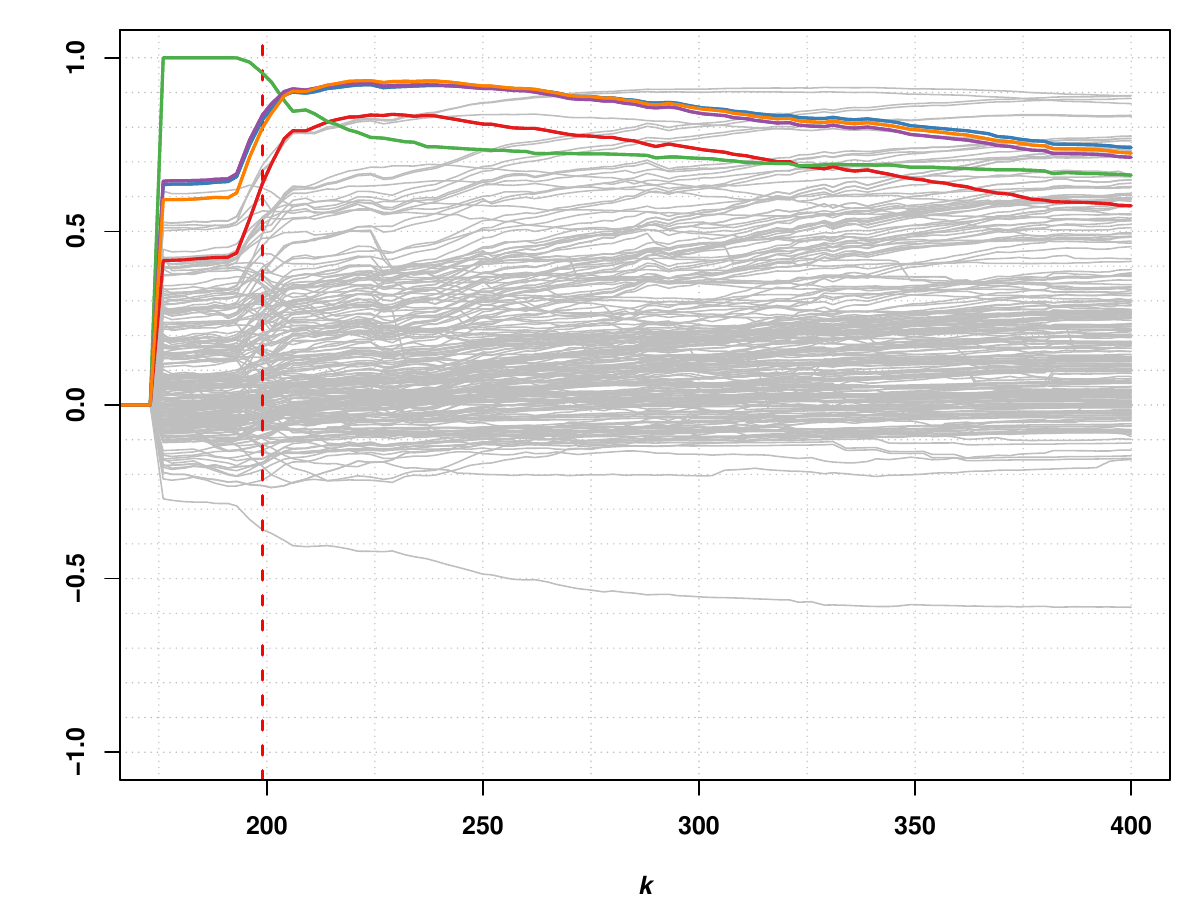}
  \caption{Correlation $\rho( \langle X, \betasmap(y)\rangle ,X^{(j)}|Y\geq y)$.}
  \label{Fig-corX}
    \end{subfigure}\hspace{1cm}
    \begin{subfigure}{0.37\textwidth}
  \centering
    \includegraphics[height=6cm]{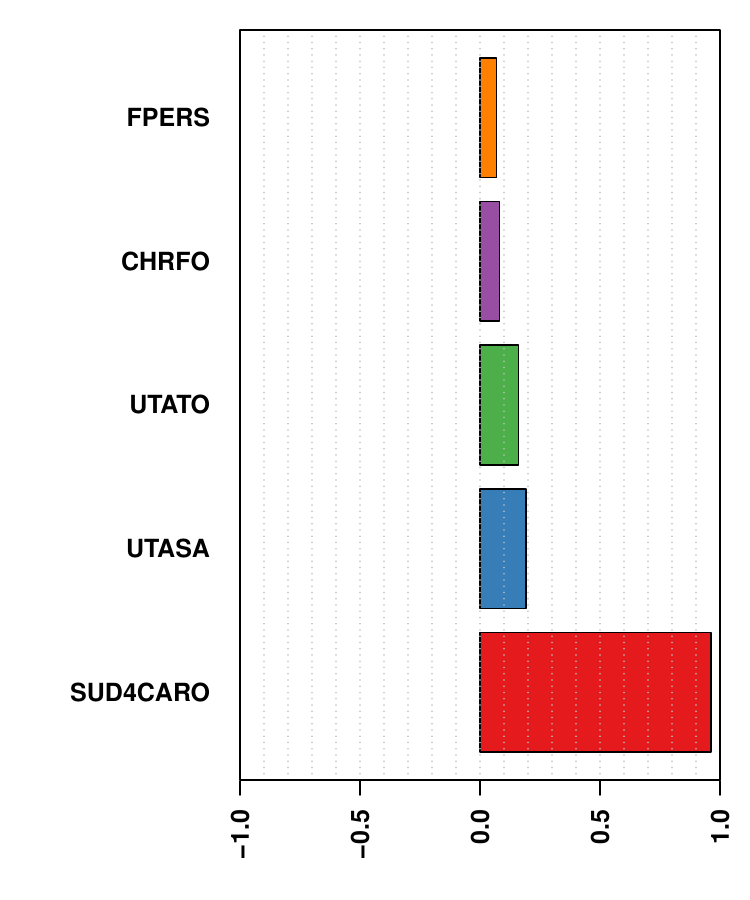}
  \caption{Non-zero coordinates of $\betasmap$.}
  \label{Fig-selected-var}
    \end{subfigure}
    \end{tabular}
    \caption{Real data example. (a) Correlation $\rho( \langle X, \betasmap(y)\rangle ,Y|Y\geq y)$ defined in Equation~\eqref{eqcorY} computed at $y=Y_{n-k+1,n}$ as a function of $k\in\{20,\dots,400\}$ for 200 evenly distributed values of $\lambda$ in $\{1,\dots,2000\}$. The selected pair $(k,\lambda)=(199,353)$ is depicted in red.
    (b) Correlation $\rho( \langle X, \betasmap(y)\rangle ,X^{(j)}|Y\geq y)$ defined in Equation~\eqref{eqcorX} computed at $y=Y_{n-k+1,n}$ as a function of $k\in\{175,\dots,400\}$ for $\lambda=353$ and $j\in\{1,\dots,259\}$. 
    (c) Non-zero coordinates of $\betasmap(Y_{n-k+1,n})$ for the selected pair $(k,\lambda)=(199,353)$.
     The colour code is the same for both left and right panels.}
    \label{fig:data_real_sparse} % bottom formerly {fig:regr_sparse_coeff}
\end{figure}

The role of the tail correlation measure~\eqref{eqcorX} is to assess the correlation in the tail between the summary $\langle X,\betasmap(y)\rangle$ of the predictors built by the SEPaLS method and the initial ones $X^{(j)}$, $j\in\{1,\dots,p\}$. It is computed at the threshold $y=Y_{n-k+1,n}$ and plotted on
 Figure~\ref{Fig-corX} as a function of the number of exceedances $k$ for $\lambda=353$. 
All correlation curves feature nice stability with respect to $k$, especially in the neighbourhood of $k=199$.

In the sequel, we thus select $k=199$ and $\lambda=353$. With these choices, only 5 coordinates of $\betasmap$ out of 259 are estimated to non-zero values, see Figure~\ref{Fig-selected-var} for an illustration and Table~\ref{tab:realDEscrib} for a description of the selected variables. Meteorological variables are discarded since large productions of carrots do not seem to depend on weather conditions. 
Remarking on Figure~\ref{fig:Histo_hill_real} that the summary variable $\langle X,\betasmap(y)\rangle$ is positively correlated with the high values of $Y$,
one can conclude that, unsurprisingly, large productions are associated with large cultivated areas (\texttt{SUD4CARO}), 
large amounts of work both in terms of time (\texttt{UTASA}, \texttt{UTATO}) and remuneration charges (\texttt{FPERS}), and large investments in supplies (\texttt{CHRFO}).

\begin{table}
    \def~{\hphantom{0}}
    
    \begin{tabular}{cllc}
    \\
    \toprule
    \textbf{Selected variables} & \textbf{Description} & \textbf{Units} & $\betasmapj$ \\
    \midrule
    \texttt{SUD4CARO} & Area cultivated with field-grown carrots&  hectares & 0.978\\
    \texttt{UTASA} & Salaried work & UTA$^{(\star)}$ &0.158\\
    \texttt{UTATO} & Salaried and not salaried work &  UTA$^{(\star)}$ & 0.124\\
    \texttt{CHRFO} & Actual cost of stored supplies& euros & 0.038\\
    \texttt{FPERS} & Remuneration charges & euros & 0.026\\
    \bottomrule
    \end{tabular}
    \caption{Real data example. Description of the 5 selected variables (out of  259) associated with 598 farms producing field-grown carrots in France from 2000 to 2015. 
    The last column displays the corresponding non-zero coordinates of $\betasmap$.  \\
    $^{(\star)}$ UTA: amount of work associated with one full-time working person during one year.}
    \label{tab:realDEscrib}
\end{table}

\section{Discussion}
\label{sec-discBEPLS}

We proposed a Bayesian interpretation of the EPLS model to introduce prior information on
the direction of dimension reduction for extreme values. Two examples of shrinkage priors are provided:
a conjugate von Mises--Fisher prior allowing to consider an initial guess on the direction, and a Laplace prior enforcing sparsity on the estimated direction.
Finite sample experiments demonstrate that the proposed method is effective in high dimension ($d=300$ on simulated data and $d\simeq 260$ on real data) with moderate sample sizes ($n=500$ on simulated data and $n\simeq 600$ on real data). 
Here, we limited ourselves to the estimation of one single direction, but the SEPaLS method can easily be adapted to the estimation
of multiple directions using the iterative procedure described in \citet[Section~4]{ExtremePLS2022}.
We also focused on prior distributions yielding explicit shrinkage estimators.
It would be of interest to investigate the use of other priors:
either uninformative priors such as Jeffreys' one~\citep{jeffreys1946invariant} or other shrinkage priors~\citep{van2019shrinkage}
can be considered. The computation of the posterior mode estimate would rely on an MCMC procedure.

\section*{Acknowledgements}

This work is partially supported by the French National Research Agency (ANR) in the framework of the Investissements d'Avenir Program (ANR-15-IDEX-02). J.~Arbel acknowledges the support of ANR-21-JSTM-0001 grant. 
S.~Girard acknowledges the support of the Chair ``Stress Test, Risk Management and Financial Steering'', led by the French Ecole Polytechnique and its Foundation and sponsored by BNP Paribas.

\clearpage

\bibliographystyle{apalike}
\bibliography{bib}

\clearpage

\appendix

\section{Appendix: Proofs}
\label{sec-appendBEPLS}

%%%%%%%%%%%%%%%%%%%%%%

This first lemma establishes that $f_{\vMFB}(\cdot|\mu,r,\kappa) $ is a proper density function integrating to one.

\begin{lemma}
\label{vMF/Bdist}
Let $p\geq2$.
For all $\mu\in S^{p-1}$, $r>0$ and $\kappa\geq0$,
$$
\int_{\|x\|_2\leq r} \frac{1}{r^p} \exp\left(\frac{\kappa \langle\mu, x\rangle}{r}\right) \dd x = \frac{1}{ 2\pi c_{p+2}(\kappa)},
$$ 
where $c_{p+2}(\kappa)$ is defined in~\eqref{cpk}.
\end{lemma}

\begin{proof}[Proof of Lemma~\ref{vMF/Bdist}]
The change of variable $x\mapsto y=x/r$ leads to
$$
\int_{\|x\|_2\leq r} \frac{1}{r^p} \exp\left(\frac{\kappa \langle\mu, x\rangle}{r}\right) \dd x =  \int_{\|y\|_2\leq 1}  \exp({\kappa \langle\mu, y\rangle}) \dd y,
$$ 
and switching to polar coordinates yields
\begin{align*}
\int_{\|y\|_2\leq 1}  \exp({\kappa \langle\mu, y\rangle}) \dd y&= \int_{0}^{1} \rho^{p-1}\int_{S^{p-1}} \exp({\rho\kappa \langle\mu, u\rangle}) \dd u \dd\rho, \\
&= \int_{0}^{1} \frac{\rho^{p-1}}{c_p(\rho\kappa)} \dd\rho\\
&= \frac{(2\pi)^{p/2}}{\kappa^{{p/2 -1}}} \int_{0}^{1} \rho^{p/2} I_{p/2 -1}(\rho\kappa) \dd\rho\\
 &= \frac{(2\pi)^{p/2}}{\kappa^{{p}}}  \int_{0}^{\kappa} t^{p/2} I_{p/2 -1}(t) \dd t.
\end{align*}
From the definition of the modified Bessel function~(\ref{besselsum}) as a power series with infinite radius of convergence, one has:
\begin{align*}
\int_{0}^{\kappa} t^{p/2} I_{p/2-1}(t) \dd t 
&=  \sum_{\ell=0}^\infty \left(\frac{1}{2^{2\ell+p/2-1}\Gamma(p/2+\ell)\ell!}   \int_{0}^{\kappa} t^{2\ell+p-1} \dd t \right)\\
&= \sum_{l=0}^\infty \frac{\kappa^{2\ell+p}}{2^{2\ell+p/2-1}\Gamma(p/2+\ell)\ell! \,(2\ell+p)} . 
\end{align*}
Taking account of $(p/2+\ell) \Gamma(p/2+\ell) = \Gamma(p/2+\ell+1)$, it follows
\begin{align*}
\int_{0}^{\kappa} t^{p/2} I_{p/2-1}(t) \dd t &=  \kappa^{p/2} \sum_{\ell=0}^\infty \frac{1}{\Gamma(p/2+\ell+1)\ell!}  \left(\frac{\kappa}{2}\right)^{2\ell+p/2}
= \kappa^{p/2} I_{p/2}(\kappa),
\end{align*}
leading to
\begin{align*}
\int_{\|x\|_2\leq r} \frac{1}{r^p} \exp\left(\frac{\kappa \langle\mu, x\rangle}{r}\right) \dd x &= \frac{(2\pi)^{p/2}}{\kappa^{p/2}} I_{p/2}(\kappa) = \frac{1}{ 2\pi c_{p+2}(\kappa)},
\end{align*}
which concludes the proof. 
\end{proof}

%%%%%%%%%%%%%%%%%%%%%%%%%%%%%%%%%%%%%%%%%%%%%%%%%%

\begin{proof}[Proof of Proposition~\ref{prop-MLE}]
For any $\theta_n>0$, in view of~\eqref{eq-vhat},
the optimization problem~(\ref{eq-opti1}) can be rewritten as: 
\begin{equation}
\label{optim2}
\hat\beta(y_n) = \argmax_{\|\beta\|_2=1} \exp \left(\theta_n \langle\beta, \hat v(y_n)\rangle \right) 
=  \argmax_{\|\beta\|_2=1} \prod_{i=1}^n \exp \left( \theta_n\langle\beta, X_i\rangle \Phi_{i}(y_n,Y_{1:n})  \right).
\end{equation}
Under model \ref{MM}, the triangle inequality yields 
$
\|X_i\|_2 \leq |g(Y_i)| +  \|\varepsilon_i\|_2,
$
and thus, conditionally on $(Y_i,\varepsilon_i)$, $X_i$ belongs to the ball centred at 0 with radius $r_i:=|g(Y_i)| +  \|\varepsilon_i\|_2$.
The optimization problem~(\ref{optim2}) can be rewritten in terms of densities associated with the $\vMFB$~distribution as
\begin{align*}
 \hat\beta(y_n) &= \argmax_{\|\beta\|_2=1} \prod_{i=1}^n f_{\vMFB}\left(X_i|\beta,  r_i=|g(Y_i)| + \|\varepsilon_i\|_2,\kappa_i=\theta_nr_i \Phi_{i}(y_n,Y_{1:n}) \right).
\end{align*}
It appears that $\hat\beta$ can be interpreted as the estimator maximizing the likelihood conditionally on $(Y_{1:n},\varepsilon_{1:n})$.
Since the density $p(\cdot,\cdot)$ of $(Y_{1:n},\varepsilon_{1:n})$ does not depend on $\beta$, one also has
$$ 
\hat\beta(y_n) = \argmax_{\|\beta\|_2=1} \left(\prod_{i=1}^n f_{\vMFB}\left(X_i|\beta, r_i=|g(Y_i)| + \|\varepsilon_i\|_2, \kappa_i=\theta_n r_i\Phi_{i}(y_n,Y_{1:n})\right)\right) p(Y_{1:n},\varepsilon_{1:n}),
$$
and thus $\hat\beta(y_n)$ can also be viewed as the unconditional maximum likelihood estimator of $\beta$.
\end{proof}

%%%%%%%%%%%%%%%%%%%%%%%%%%%%%%%%%%%%%%%%%%%%%%%%%

\noindent The next lemma will reveal useful in the proof of Proposition~\ref{prop-recall} below.
\begin{lemma}
 \label{lem-renormalised-sequence}
Let $(\sigma_n)$ and $(c_n)$ be positive real sequences with $\sigma_n\to 0$ as $n\to\infty$. Let $A$ be a random vector in $\mathbb{R}^p$, $b\in S^{p-1}$ a non-random vector,  and $(B_n)$ a sequence of random vectors in $\mathbb{R}^p$ such that
$$
\sigma_n^{-1}\left(\frac{B_n}{c_n}-b\right) \tod A.
$$
Then,
$$
\sigma_n^{-1}\left(\frac{B_n}{\Vert B_n \Vert_2}-b\right) \toP P_b^\perp(A),
$$
where $P_b^\perp(A):= A -  \langle b,A\rangle b$ denotes the projection of $A$ on the hyperplane orthogonal to $b$.
\end{lemma}
\begin{proof}[Proof of Lemma~\ref{lem-renormalised-sequence}]
Let $\epsilon_n:=\sigma_n^{-1}\left(\frac{B_n}{c_n}-b\right) - A$. From the assumption of convergence in distribution, we have that $\epsilon_n$ converges in distribution to a Dirac mass at 0. Clearly,
\begin{align*}
    \Vert B_n \Vert_2^2 = c_n^2\Vert b+ \sigma_n(A+\epsilon_n)\Vert_2^2
     = c_n^2\left(1+2\sigma_n\langle b,A+\epsilon_n\rangle +\bigOP(\sigma_n^2)\right),
\end{align*}
and inverting the latter equality yields 
\begin{align*}
    c_n = \Vert B_n \Vert_2\left(1-\sigma_n\langle b,A+\epsilon_n\rangle +\bigOP(\sigma_n^2)\right).
\end{align*}
Replacing in the expression of $B_n=c_n(b+\sigma_n(A+\epsilon_n))$, we obtain
\begin{align*}
    B_n &= \Vert B_n \Vert_2\left(1-\sigma_n\langle b,A+\epsilon_n\rangle +\bigOP(\sigma_n^2))(b+ \sigma_n(A+\epsilon_n)\right)\\
    &= \Vert B_n \Vert_2\left(b+\sigma_n(A+\epsilon_n-b\langle b,A+\epsilon_n\rangle) +\bigOP(\sigma_n^2)\right),
\end{align*}
and therefore
\begin{align*}    \sigma_n^{-1}\left(\frac{B_n}{\Vert B_n \Vert_2}-b\right) = A+\epsilon_n-b\langle b,A+\epsilon_n\rangle +\bigOP(\sigma_n^2)
    \toP A - b \langle b,A\rangle = P_b^\perp(A),
\end{align*}
which is the desired result.
\end{proof}

\begin{proof}[Proof of Proposition~\ref{prop-recall}] 
From~\citet[Theorem~1]{ExtremePLS2022}, one has
$$
\sqrt{n\bar{F}(y_n)}  \left(\frac{\hat v(y_n)}{\|{ v(y_n)}\|_2 } - \beta \right) \tod \xi \beta,
$$
with $\xi$ a centered Gaussian random variable and where
$$
 v(y_n):=\bar F(y_n) \mathbb{E}(XY  \mathbf{1}_{\{Y\geq y_n\}}) - \mathbb{E}(X\mathbf{1}_{\{Y\geq y_n\}})\mathbb{E}(Y\mathbf{1}_{\{Y\geq y_n\}}).
$$%Let us also recall that, from~\cite[Equation~(28)]{ExtremePLS2022}, one has $\|v(y_n)\|_2\sim \lambda_2(\gamma,c) \Lambda(y_n)$ as $n\to\infty$, where $\lambda_2(\gamma,c)>0$ and $\Lambda(y_n)= y_ng(y_n) \bar F(y_n)^2$.
The result follows from Lemma~\ref{lem-renormalised-sequence} applied with
$\sigma_n=1/\sqrt{n\bar F(y_n)}$, $B_n=\hat v(y_n)$, $c_n=\|{v(y_n)}\|_2$, $b=\beta$, $A=\xi\beta$ and therefore $P_b^\perp(A)=0$.
\end{proof}

%%%%%%%%%%%%%%%%%%%%%%%%%%%%%%%%%%%%%%%%%%%%%%%%%

\begin{proof}[Proof of Proposition~\ref{prop-post}]
In view of Bayes' rule,
the posterior distribution of $\beta$ is given by
$$  
p(\beta|X_{1:n},Y_{1:n},\varepsilon_{1:n})
    \propto \pi(\beta) p(Y_{1:n},\varepsilon_{1:n})\prod_{i=1}^n f_{\vMFB}\left(X_i|\beta,r_i=|g(Y_i)| + \|\varepsilon_i\|_2, \kappa_i= \theta_n r_i \Phi_{i}(y_n,Y_{1:n}) \right).
$$
Since $p(Y_{1:n},\varepsilon_{1:n})$ does not depend on $\beta$, the posterior distribution can be simplified as
\begin{align*}
 p(\beta|X_{1:n},Y_{1:n},\varepsilon_{1:n}) &\propto \pi(\beta)\prod_{i=1}^n f_{\vMFB}\left(X_i|\beta,r_i=|g(Y_i)| + \|\varepsilon_i\|_2, \kappa_i= \theta_n r_i \Phi_{i}(y_n,Y_{1:n}) \right)\\
  &\propto \pi(\beta)\prod_{i=1}^n \exp \left( \theta_n\langle\beta, X_i\rangle \Phi_{i}(y_n,Y_{1:n})  \right)\\
&= \pi(\beta) \exp \left( \theta_n  \|\hat v(y_n)\|_2 \langle\beta,  \betaml(y_n)\rangle \right),
\end{align*}
and the result is proved.
\end{proof}

\begin{proof}[Proof of Proposition~\ref{prop-asymp-conj}]
Let $\sigma_n=1/\sqrt{n \bar F(y_n)}$. Combining Proposition~\ref{prop:conj} and Proposition~\ref{prop-recall}, it follows
$$
\betacmap(y_n) = \frac{\beta  + \sigma_n\varepsilon_n+ (\kappa_0/K_n) \mu_0}{\|\beta  + \sigma_n\varepsilon_n+ (\kappa_0/K_n) \mu_0\|_2},
$$
where $\varepsilon_n:=\sigma_n^{-1}(\betaml(y_n)-\beta)\toP 0$. 
Taking account of $\sigma_n\to 0$ and $1/K_n \simP \sigma_n/c \to 0$ as $n\to\infty$, a first order Taylor expansion yields:
$$
\|\beta  + \sigma_n\varepsilon_n+ (\kappa_0/K_n) \mu_0\|_2^2 = 1 + 2 (\kappa_0/K_n)  \langle\mu_0,\beta\rangle +o_{\mathbb{P}}(\sigma_n) + o_{\mathbb{P}}(1/K_n),
$$
{since $\|\beta\|_2=1$,} and therefore 
$$
1/\|\beta  + \sigma_n\varepsilon_n+ (\kappa_0/K_n) \mu_0\|_2 = 1 - (\kappa_0/K_n)  \langle\mu_0,\beta\rangle +o_{\mathbb{P}}(\sigma_n) + o_{\mathbb{P}}(1/K_n).
$$
Replacing, we get
$$
\betacmap{(y_n)} = \beta + (\kappa_0/K_n) (\mu_0 -  \langle \mu_0,\beta\rangle\beta) +o_{\mathbb{P}}(\sigma_n) + o_{\mathbb{P}}(1/K_n),
$$
or equivalently,
$$
\sigma_n^{-1}(\betacmap{(y_n)} -\beta) = \kappa_0/(\sigma_n K_n)  (\mu_0 -  \langle \mu_0,\beta\rangle\beta) +o_{\mathbb{P}}(1) + o_{\mathbb{P}}(1/(\sigma_nK_n)),
$$
and the result is proved under the assumption that $\sigma_n K_n \toP c>0$ as $n\to\infty$.
\end{proof}

%%%%%%%%%%%%%%%%%%%%%%%%%%%%%%%%%%%%%%%%%

\begin{proof}[Proof of Proposition~\ref{prop:sparse}]
In view of~\eqref{posteriorsparse}, the MAP estimator is given by:
\begin{align*}
    \betasmap(y_n) &= \argmin_{\|\beta\|_2^2=1} \lambda \|\beta\|_1 -  K_n\langle\beta, \betaml(y_n)\rangle\\
     &= \argmin_{\|\beta\|_2^2=1} \sum_{j=1}^{p} \left(\lambda |\beta_j| -  K_n\beta_j \betamlj(y_n)\right)\\
     &=  \argmin_{\|\beta\|_2^2=1} \sum_{j=1}^{p} |\beta_j| \left(\lambda  -  K_n\sign(\beta_j)\betamlj(y_n)\right).
\end{align*}
Introducing $b_j=|\beta_j|$ and $s_j=\sign(\beta_j)$ so that $\beta_j=s_j b_j$, the above optimization problem can be rewritten as
$$
 \betasmap(y_n) = \argmin_{b,s} \sum_{j=1}^p b_j(\lambda -  K_n s_j \betamlj(y_n) ) \; \mbox{ s.t. } \|b\|_2^2=1,\; b_j\geq 0, \; |s_j|=1, \;j\in\{1,\dots,p\}.
$$
Clearly, the solution w.r.t. $s$ is given by
 $s_j=\sign(\betamlj(y_n))$ for all $j\in\{1,\dots,p\}$ and therefore
$$
 \betasmap(y) = \argmin_{b\in{\mathbb R}^p}  C(b), \; \mbox{ s.t. } \|b\|_2^2=1,\; b_j\geq 0, \;j\in\{1,\dots,p\}
$$
where
$$C(b) = \sum_{j=1}^p b_j(\lambda -  K_n |\betamlj(y_n)| ).$$

Let us introduce the two sets of indices 
\begin{align*}
J_+ = \left\{
 j\in\{1,\dots,p\} \, ; \,  \lambda - K_n |\hat{\beta}_{\text{ml},j}(y)|\geq 0
\right\} \text{ and }
J_- = \left\{
 j\in\{1,\dots,p\}  \, ; \,  \lambda - K_n |\hat{\beta}_{\text{ml},j}(y)|<0
\right\},
\end{align*}
such that $C(b)=C_+(b)-C_-(b)$ where
\begin{align*}
C_+(b)=\sum_{j\in J_+} b_j\big(\lambda-K_n |\hat{\beta}_{\text{ml},j}(y)|\big) 
\quad \text{and} \quad
C_-(b)=\sum_{j\in J_-} b_j\big|\lambda-K_n |\hat{\beta}_{\text{ml},j}(y)|\big|.
\end{align*}
The minimum of the non-negative term $C_+(b)$ is reached for $b_j=0$, $\forall j\in J_+$.
The negative term $C_-(b)$ corresponding to negative values of $\lambda-K_n |\hat{\beta}_{\text{ml},j}(y)|$ remains and the problem 
can be rewritten as
\begin{align*}
\hat{\beta}_{\text{map}}^{\text{s}}(y)  =
\argmin_{b\in{\mathbb R}^p} \sum_{j\in J_-} b_j\left(\lambda - K_n |\hat{\beta}_{\text{ml},j}(y)|\right) 
\quad \text{s.t.}\quad
\|b\|_2^2=1
\quad \text{and}\quad\left\{
\begin{array}{l}
b_j\geq 0, \quad j\in\{1,\dots,p\},\\
 b_j=0, \quad j\in J_+.
\end{array}\right.
\end{align*}
One can recognise a problem of minimization of projection on the vector of negative terms $\big(
\lambda - K_n |\hat{\beta}_{\text{ml},j}(y)|\big)_{j\in J_-}$
which is solved for positive terms $(b_j)_{j\in J_-}$ defined by
\begin{align*}
\forall j\in J_-,\ b_j = \big(K_n |\hat{\beta}_{\text{ml},j}(y)|-\lambda\big)/\sqrt{\delta}\quad
\text{where}\quad  \delta = \sum_{j\in J_-} \big(K_n |\hat{\beta}_{\text{ml},j}(y)|-\lambda\big)^2.
\end{align*}
One can notice that $\delta=\|S_{\lambda}(K_n |\hat{\beta}_{\text{ml}}(y)|)\|_2^2$, and therefore
\begin{align*}
\hat{\beta}_{\text{map}}^{\text{s}}(y) &= S_{\lambda}(K_n |\hat{\beta}_{\text{ml}}(y)|)/\|S_{\lambda}(K_n |\hat{\beta}_{\text{ml}}(y)|)\|_2.
\end{align*}
The result is thus proved.

%%%

 \end{proof}
 
 %%%%%%%%%%%%%%%%%%%%%%%%%%%%%%%%%%%%%%%%%%%%%%%%%%%%%%%%

 \begin{proof}[Proof of Proposition~\ref{prop-asymp-sparse}]
Let us recall the notation introduced in the proof of Proposition~\ref{prop-asymp-conj}: $\sigma_n=1/\sqrt{n \bar F(y_n)}$. Combining Proposition~\ref{prop:sparse} and Proposition~\ref{prop-recall}, it follows
that $\betasmap(y_n) = \tilde \beta(y_n)/ \|\tilde \beta(y_n)\|_2$ with,
for all $j\in\{1,\dots,p\}$:
$$
\tilde \beta_j(y_n)= S_\lambda\left( K_n (\beta_j + \sigma_n \varepsilon_{j,n}) \right),
$$
where $\varepsilon_n\toP 0$. Two cases arise:
\begin{itemize}
 \item If $\beta_j=0$ then, clearly, $\tilde\beta_j(y_n)=0$ with probability tending to one, since $K_n\sigma_n\toP c$ and $\varepsilon_n \toP 0$ as $n\to\infty$.
 \item If $\beta_j\neq 0$, then $K_n\toP \infty$ and $K_n\sigma_n\toP c$ entail $|K_n (\beta_j + \sigma_n \varepsilon_{j,n})|\toP \infty$
 as $n\to\infty$
 and, therefore, with probability tending to one, 
 \begin{equation}
 \label{eq-betajneq}
 \tilde\beta_j(y_n) =\sign(\beta_j)\left( K_n (|\beta_j| \pm \sigma_n \varepsilon_{j,n}) - \lambda\right) = 
 \beta_j K_n \left( 1  - \frac{ \lambda }{|\beta_j| K_n}(1+o_{\mathbb{P}}(1)) \right).
  \end{equation}
\end{itemize}
As a consequence, one has, with probability tending to one, 
\begin{align*}
 \|\tilde\beta(y_n)\|_2^2 &= K_n^2 \sum_{\beta_j\neq 0} \beta_j^2 \left( 1  - \frac{ \lambda }{|\beta_j| K_n}(1+o_{\mathbb{P}}(1)) \right)^2 \\
 &= K_n^2 \left\{1 + \sum_{\beta_j\neq 0} \beta_j^2 \left(  \frac{ \lambda^2 }{\beta_j^2 K_n^2}(1+o_{\mathbb{P}}(1)) - \frac{ 2 \lambda }{|\beta_j| K_n}(1+o_{\mathbb{P}}(1)) \right) \right\},
\end{align*}
since $\|\beta\|_2=1$. It follows that
$$
 \|\tilde\beta(y_n)\|_2^2  = K_n^2 \left( 1 - \frac{2 \lambda \|\beta\|_1}{K_n} (1+o_{\mathbb{P}}(1)) \right),
$$
with probability tending to one, leading to
$$
\frac{1}{\|\tilde\beta(y_n)\|_2} =  \frac{1}{K_n} \left( 1 +  \frac{ \lambda \|\beta\|_1}{K_n} (1+o_{\mathbb{P}}(1)) \right).
$$
Combining with~(\ref{eq-betajneq}), one has, for all $j\in\{1,\dots,p\}$ such that $\beta_j\neq 0$,
$$
\frac{\tilde\beta_j(y_n)}{\|\tilde\beta(y_n)\|_2} = \beta_j \left( 1 + \frac{\lambda}{K_n} \left(\|\beta\|_1 - \frac{1}{|\beta_j|} \right)(1+o_{\mathbb{P}}(1)) \right),
$$
or equivalently,
$$
\sigma_n^{-1}\left(\frac{\tilde\beta_j(y_n)}{\|\tilde\beta(y_n)\|_2} - \beta_j\right) =  \frac{\lambda}{K_n\sigma_n} \left(\|\beta\|_1 - \frac{1}{|\beta_j|} \right)\beta_j\; (1+o_{\mathbb{P}}(1)) ,
$$
and $K_n\sigma_n\toP c$ proves the result.
\end{proof}

% \newpage 

\section{Appendix: Additional figures}\label{sec-extra-figures}

We provide below additional figures corresponding to the illustration on simulated data presented in Section~\ref{sec-simBEPLS}. They correspond to the use of the conjugate prior with parameter $c\in\{1/2,1/4\}$ (while the case $c=1$ can be found in the main text), and the sparse prior with parameter $c\in\{1,1/2,1/4\}$.

\begin{figure}[htb]
\centering
\begin{tabular}{ccc}
    & \multicolumn{2}{c}{Conjugate $\vMFS$ prior and link function $g(t)=t^c$ with $c=1/2$.}\\
    &  $\mu_0=\beta$ &   $\mu_0=\tilde\beta$\\
    \rotatebox[origin=l]{90}{\hspace{.9cm}$\tau=-0.8$} &
    \includegraphics[trim={1cm 1.5cm 0.5cm 2cm},clip,width=.43\textwidth]{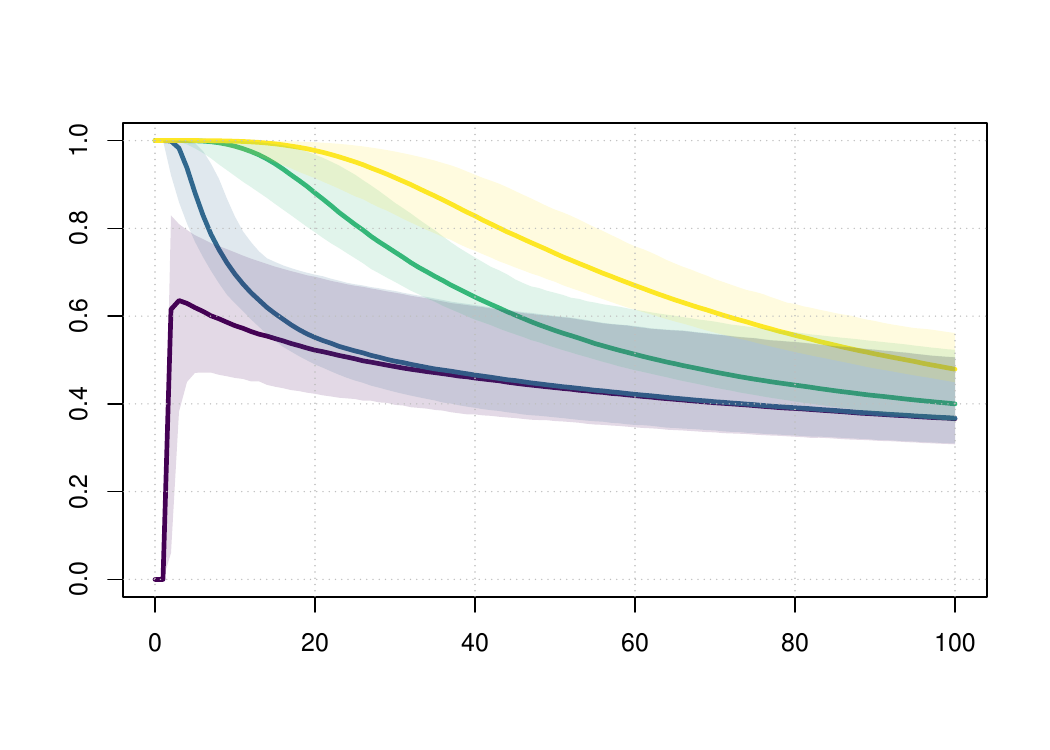} & 
    \includegraphics[trim={1cm 1.5cm 0.5cm 2cm},clip,width=.43\textwidth]{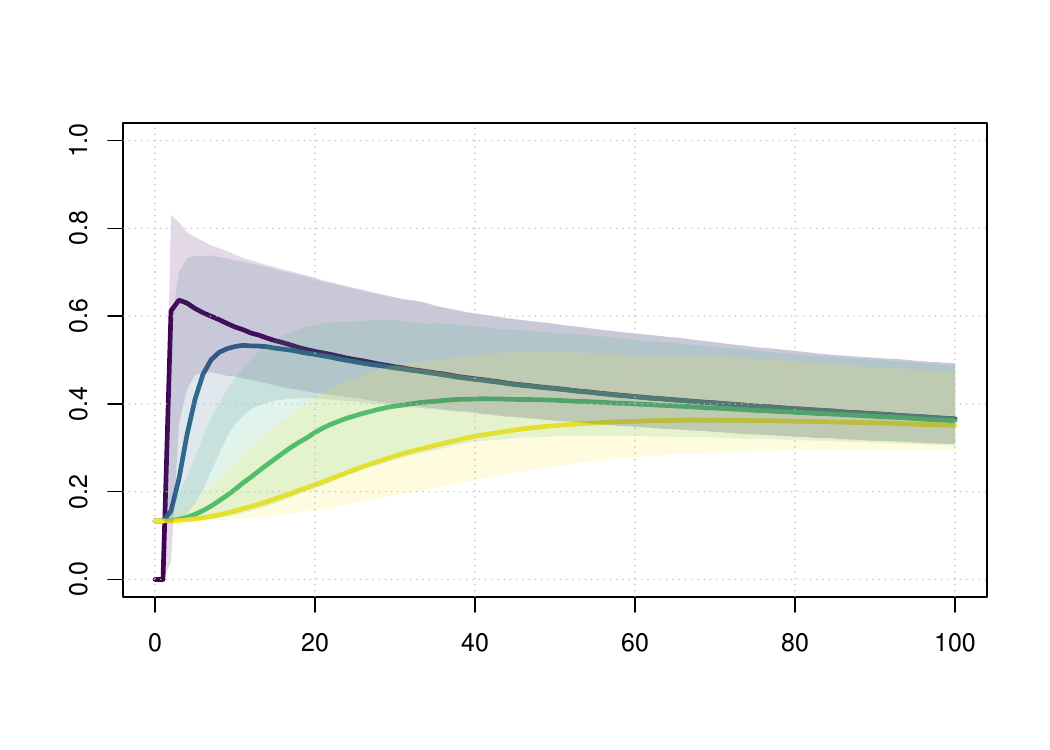}
 \\ \rotatebox[origin=l]{90}{\hspace{.9cm}$\tau=-0.2$} &
    \includegraphics[trim={1cm 1.5cm 0.5cm 2cm},clip,width=.43\textwidth]{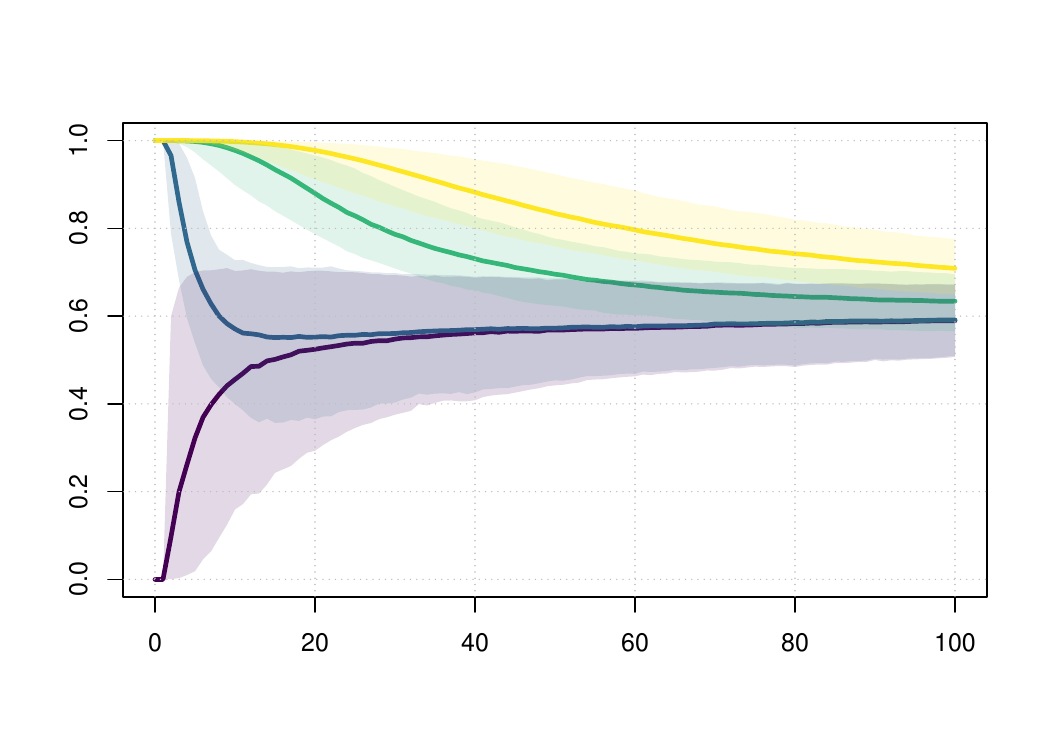}
    &
    \includegraphics[trim={1cm 1.5cm 0.5cm 2cm},clip,width=.43\textwidth]{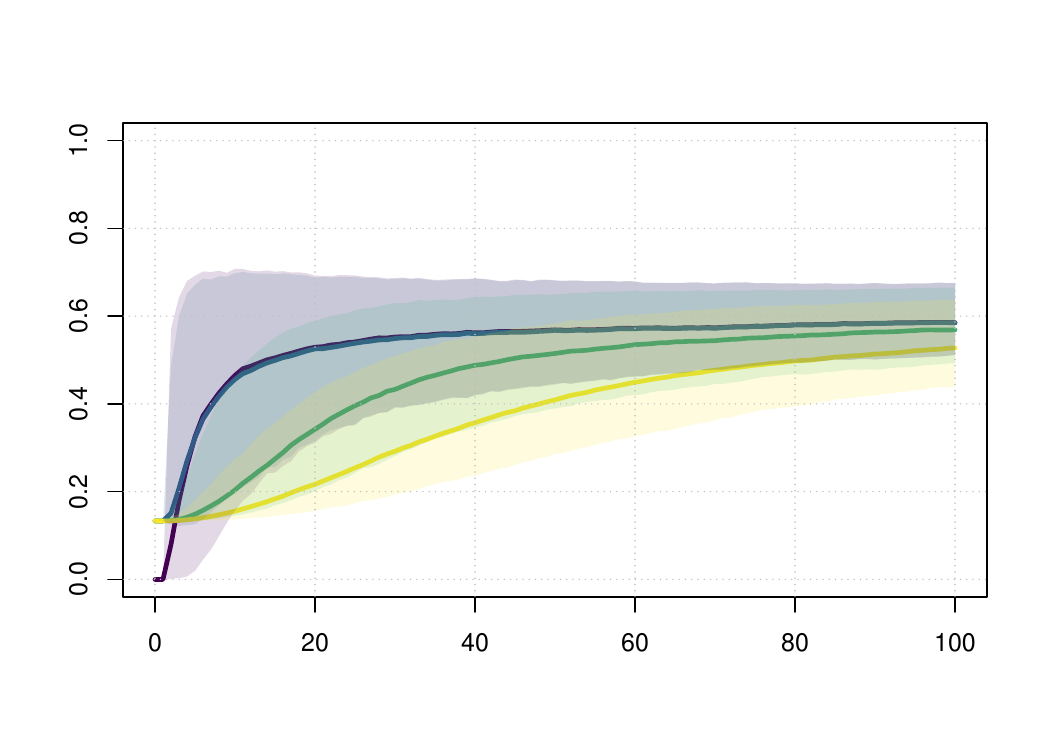}
 \\ \rotatebox[origin=l]{90}{\hspace{.9cm}$\tau=0.2$} &
    \includegraphics[trim={1cm 1.5cm 0.5cm 2cm},clip,width=.43\textwidth]{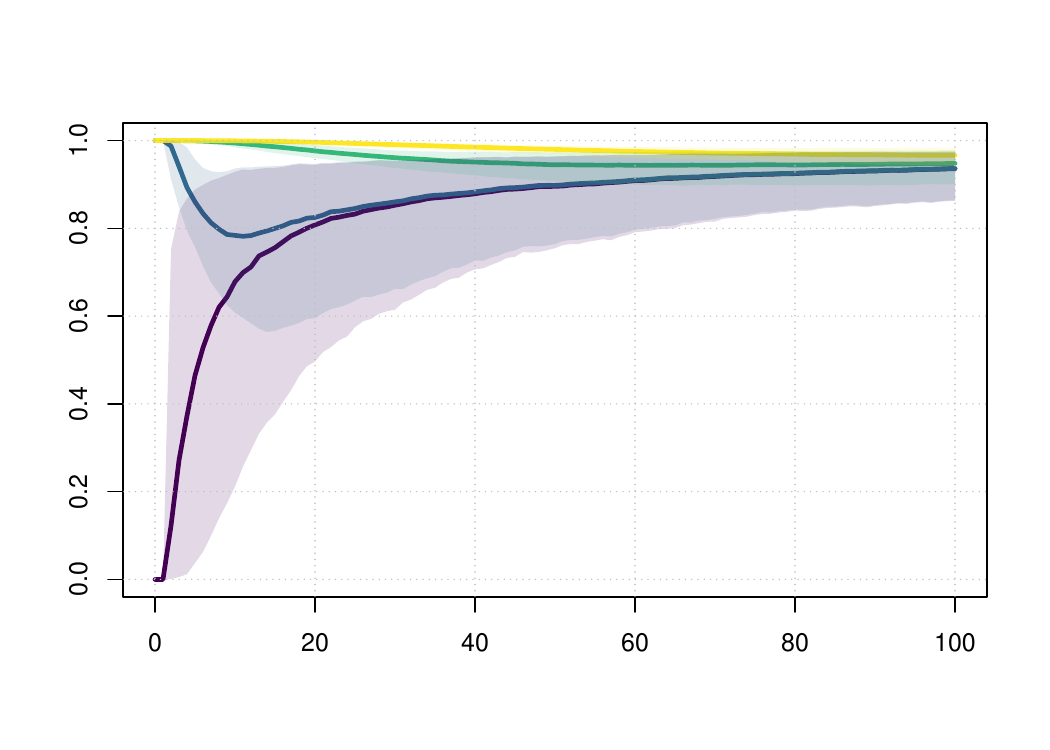}
    &
    \includegraphics[trim={1cm 1.5cm 0.5cm 2cm},clip,width=.43\textwidth]{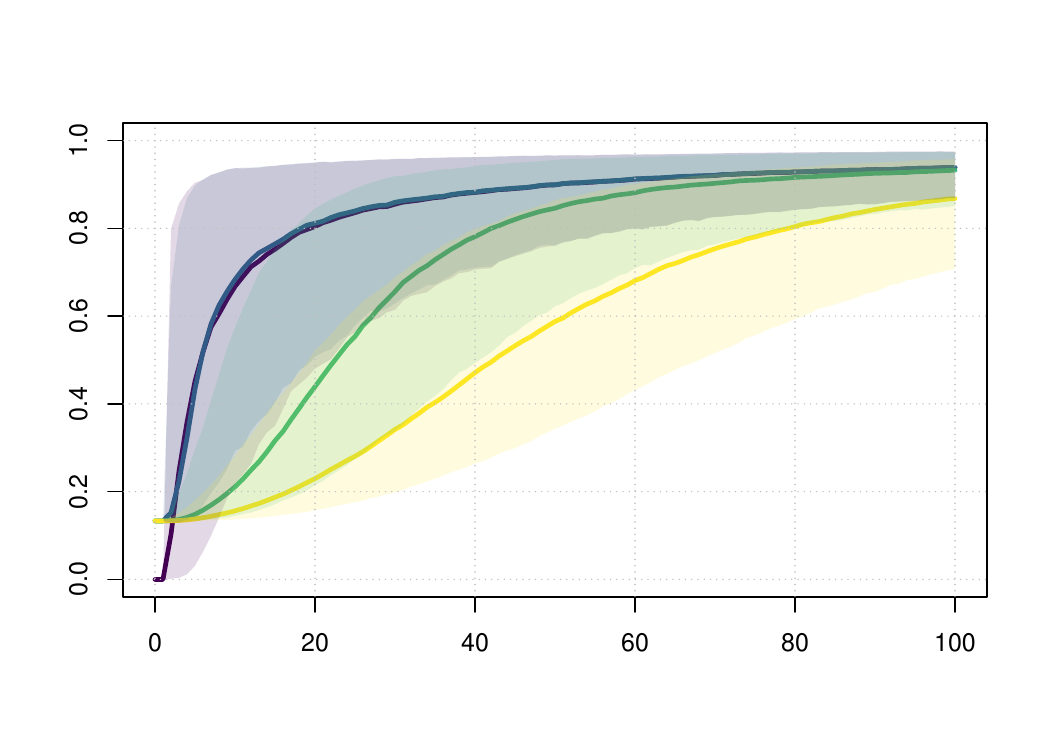}
 \\ \rotatebox[origin=l]{90}{\hspace{.9cm}$\tau=0.8$} &
    \includegraphics[trim={1cm 1.5cm 0.5cm 2cm},clip,width=.43\textwidth]{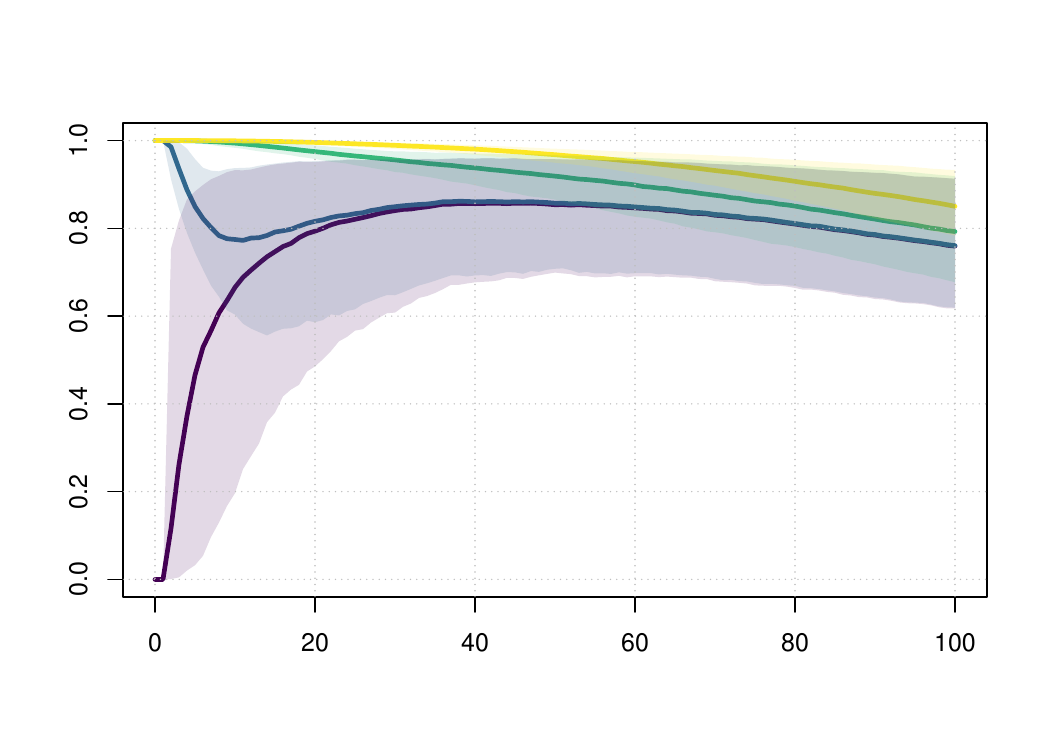}
    &
    \includegraphics[trim={1cm 1.5cm 0.5cm 2cm},clip,width=.43\textwidth]{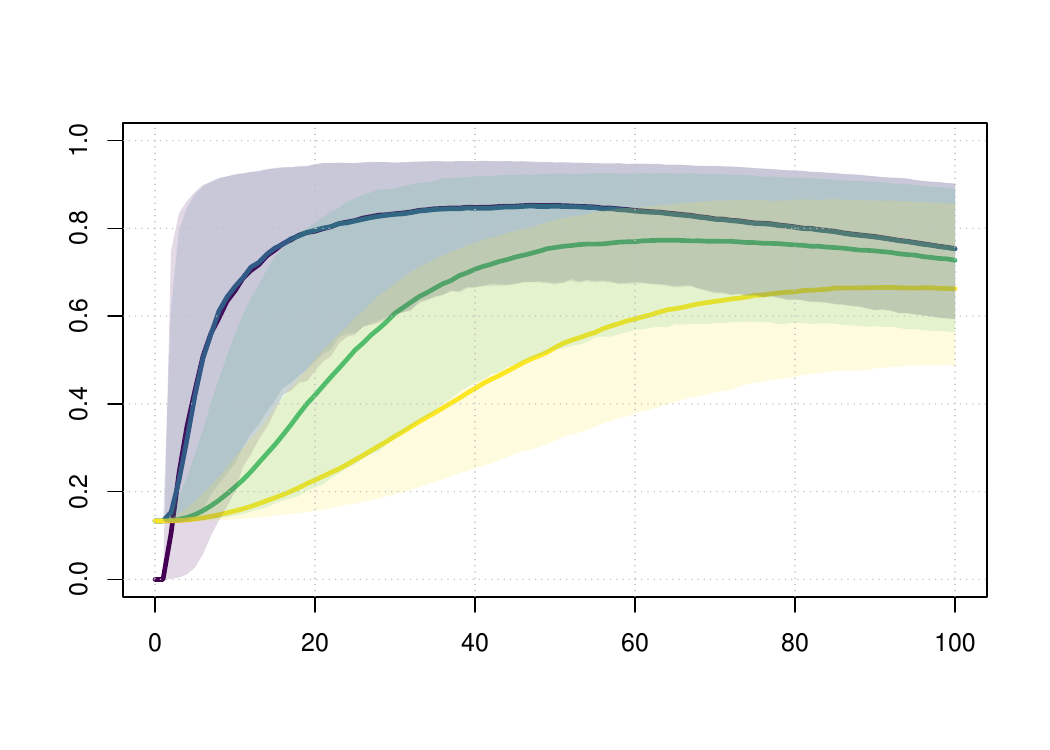}
\end{tabular}
\caption{\legendconjugate{1/2}}
\label{BEPLSConj2}
\end{figure}

\begin{figure}[htb]
\centering
\begin{tabular}{ccc}
    & \multicolumn{2}{c}{Conjugate $\vMFS$ prior and link function $g(t)=t^c$ with $c=1/4$.}\\
    &  $\mu_0=\beta$ &   $\mu_0=\tilde\beta$\\
    \rotatebox[origin=l]{90}{\hspace{.9cm}$\tau=-0.8$} &
    \includegraphics[trim={1cm 1.5cm 0.5cm 2cm},clip,width=.43\textwidth]{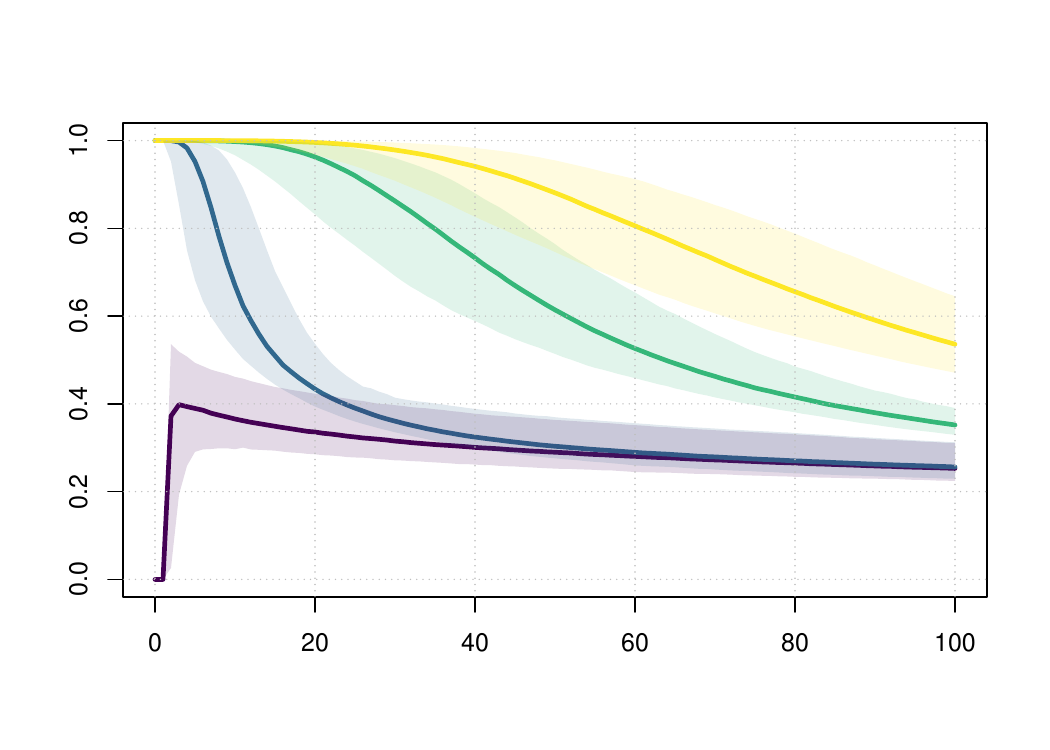}
    &
    \includegraphics[trim={1cm 1.5cm 0.5cm 2cm},clip,width=.43\textwidth]{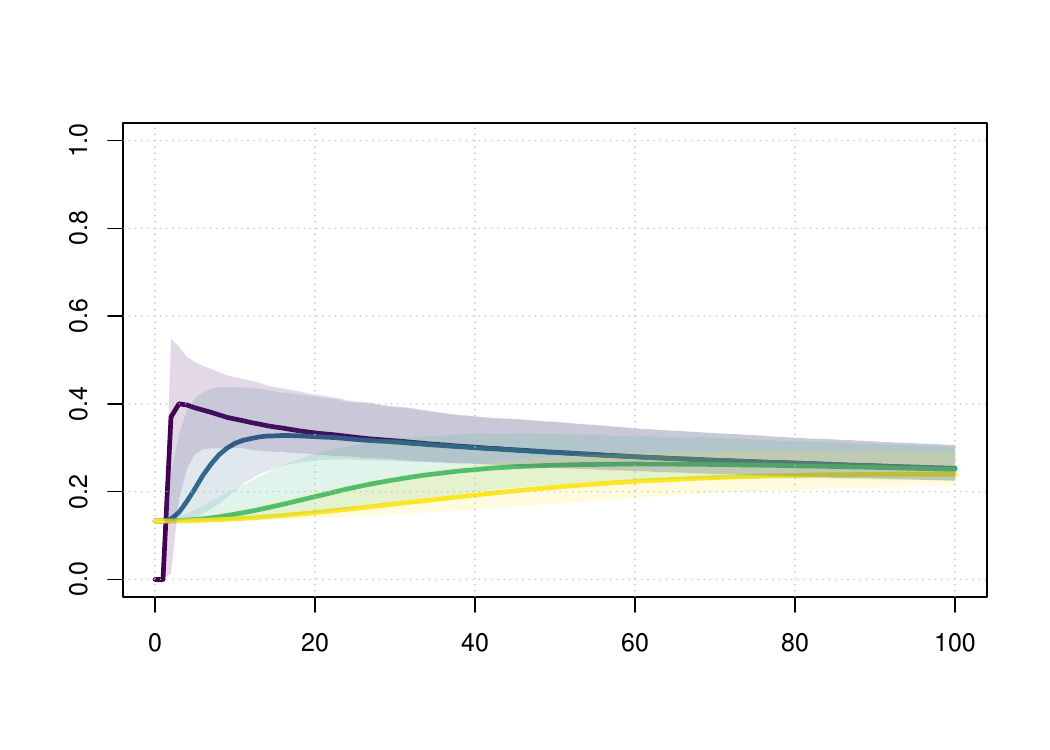}
 \\ \rotatebox[origin=l]{90}{\hspace{.9cm}$\tau=-0.2$} &
    \includegraphics[trim={1cm 1.5cm 0.5cm 2cm},clip,width=.43\textwidth]{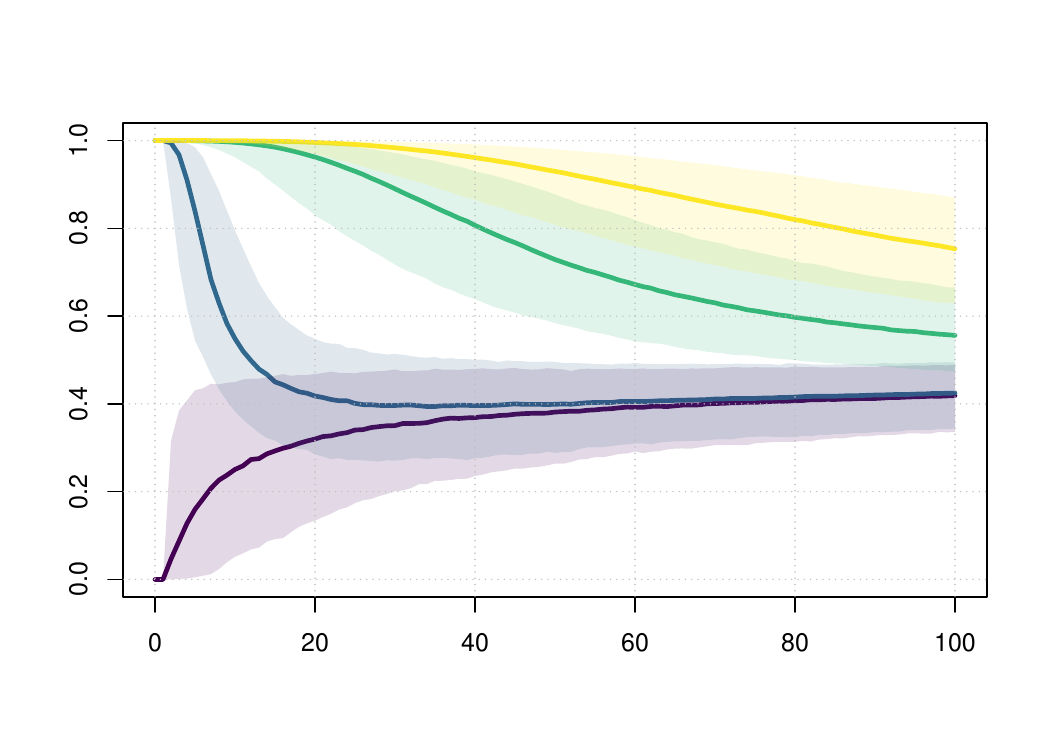}
    &
    \includegraphics[trim={1cm 1.5cm 0.5cm 2cm},clip,width=.43\textwidth]{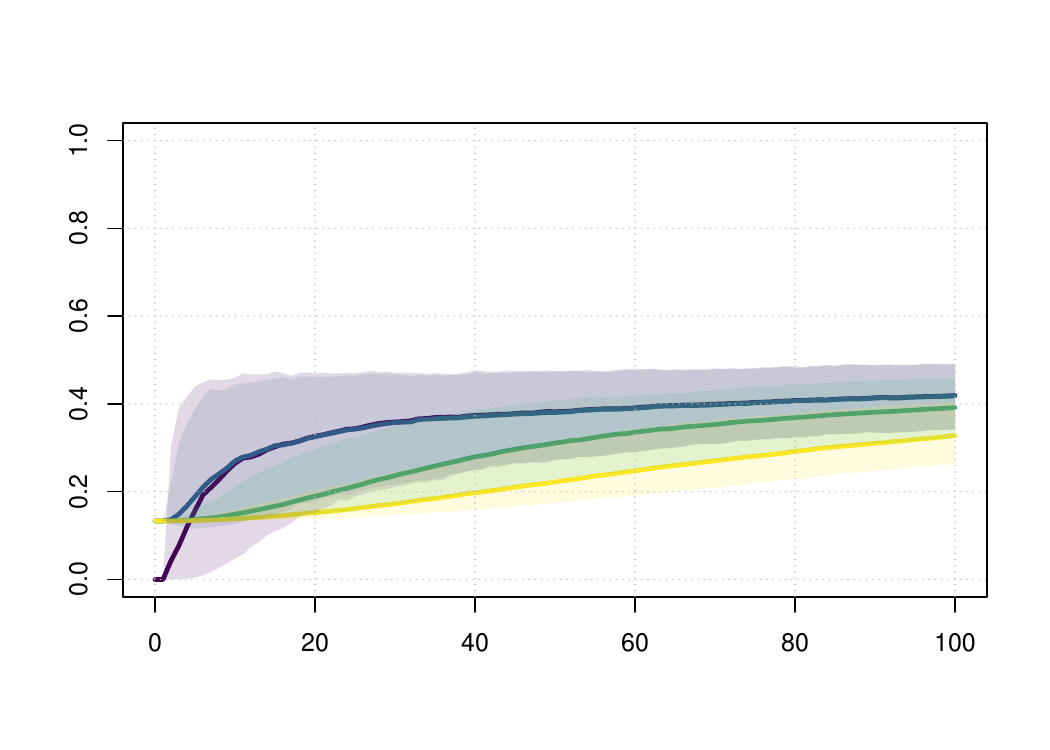}
 \\ \rotatebox[origin=l]{90}{\hspace{.9cm}$\tau=0.2$} &
    \includegraphics[trim={1cm 1.5cm 0.5cm 2cm},clip,width=.43\textwidth]{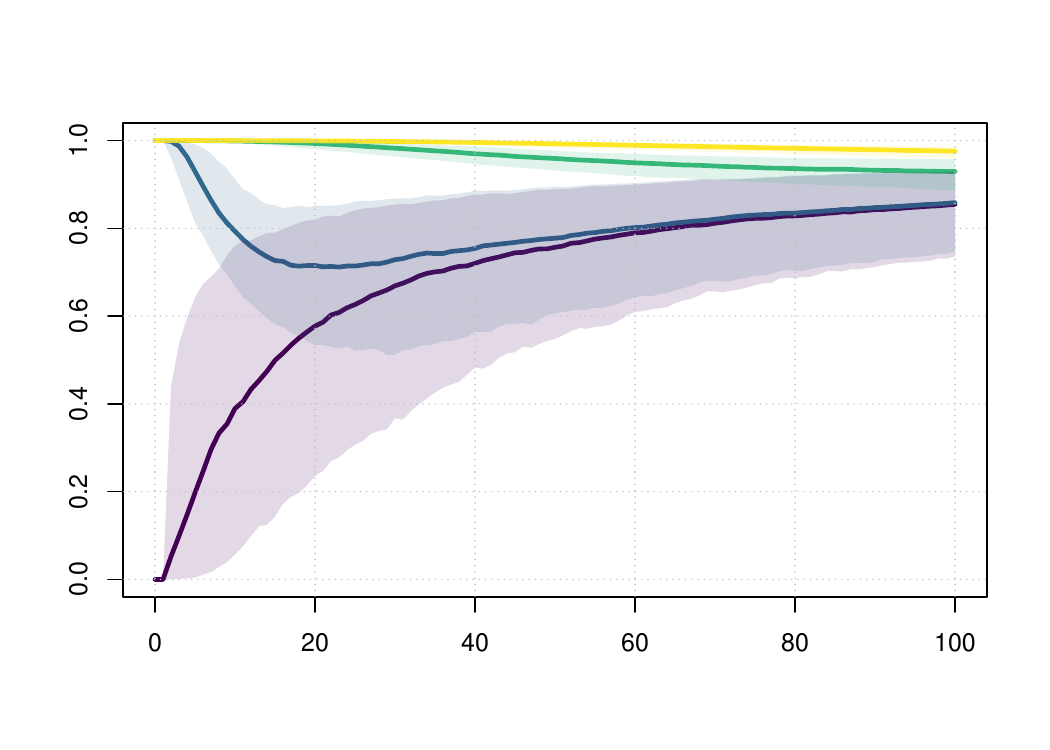}
    &
    \includegraphics[trim={1cm 1.5cm 0.5cm 2cm},clip,width=.43\textwidth]{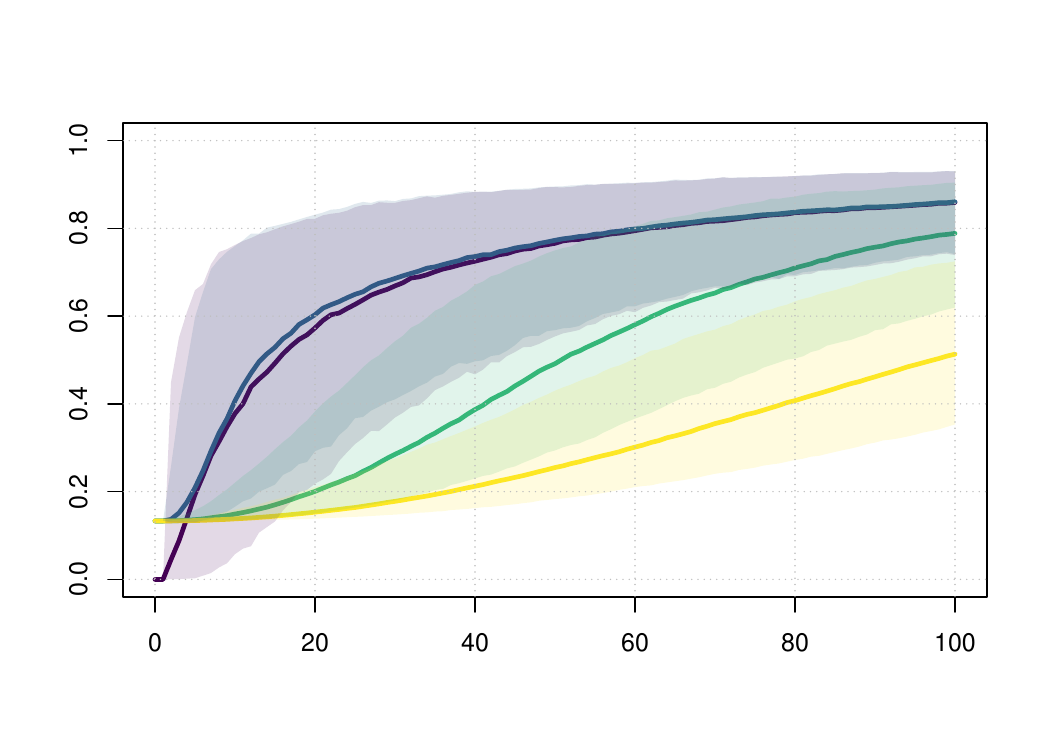}
 \\ \rotatebox[origin=l]{90}{\hspace{.9cm}$\tau=0.8$} &
    \includegraphics[trim={1cm 1.5cm 0.5cm 2cm},clip,width=.43\textwidth]{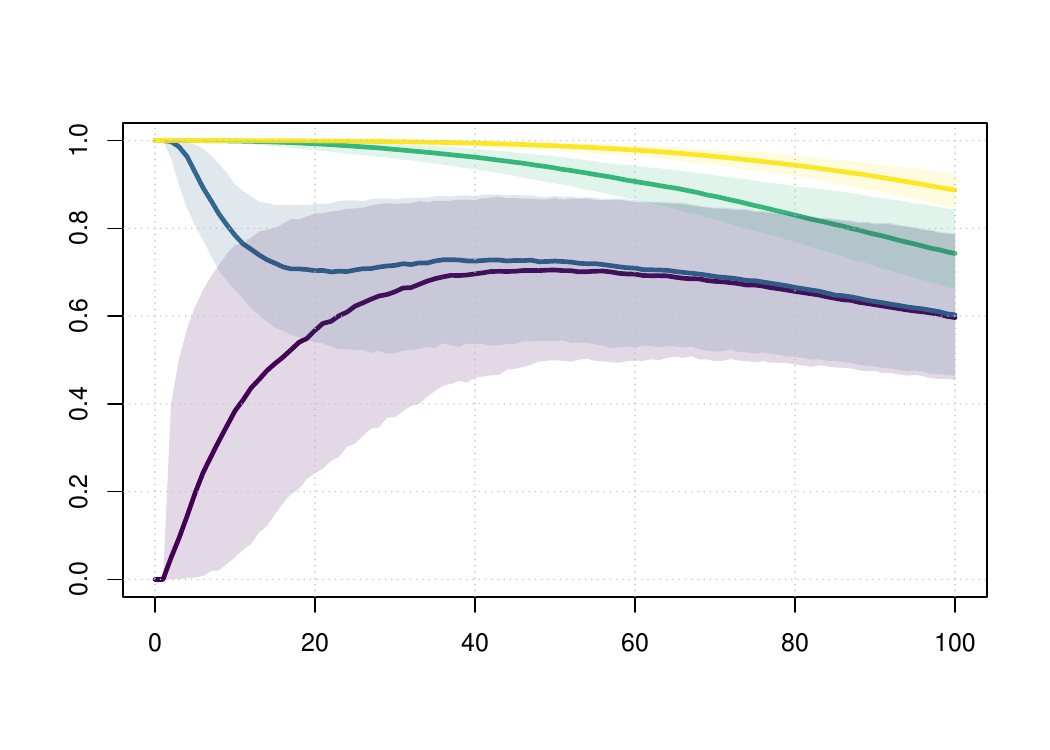}
    &
    \includegraphics[trim={1cm 1.5cm 0.5cm 2cm},clip,width=.43\textwidth]{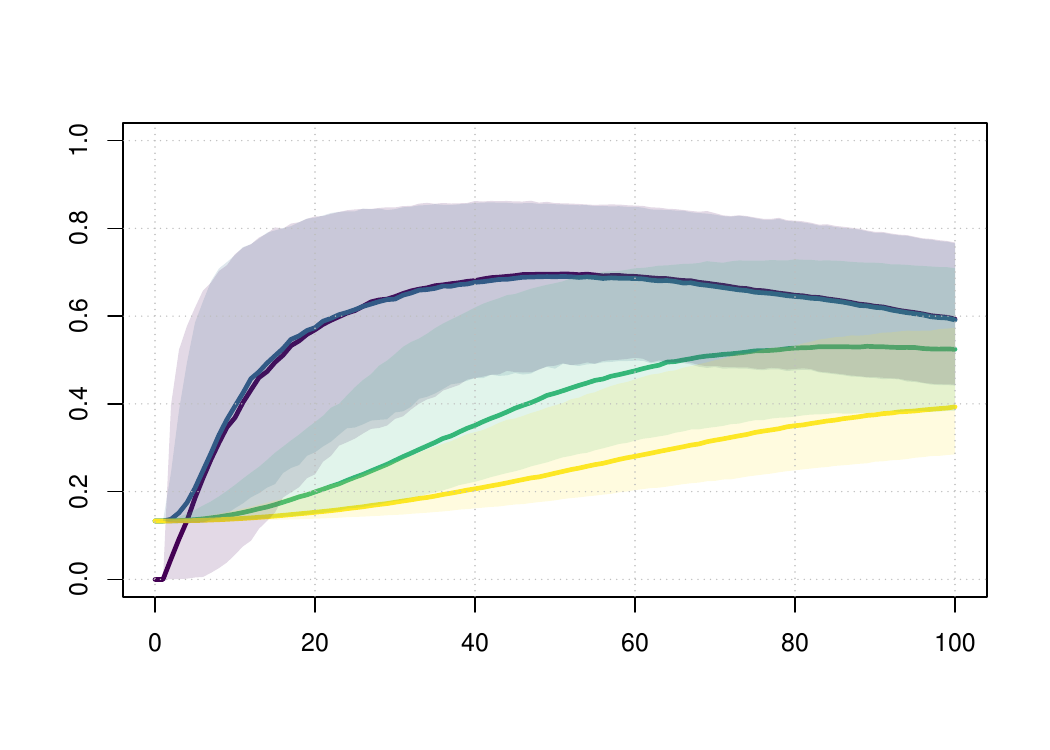}
    \end{tabular}
\caption{\legendconjugate{1/4}}
\label{BEPLSConj3}
\end{figure}

\newcommand{\legendsparse}[1]{Finite sample behavior of the SEPaLS estimator computed with the sparse prior on simulated data in dimension $d=30$ (left) and $d=300$ (right) from a Pareto distribution (${\gamma_Y}=1/5,\;a=2$) and a (rotated) Clayton copula with Kendall's tau $\tau\in\{-0.8, -0.2, 0.2, 0.8\}$ (from top to bottom). The power of the link function $g(t)=t^c$ is fixed to $c=#1$. Vertically: $\similarity (Y_{n-k+1,n})$ between $\betasmap$ and $\beta$ for as a function of the number $k \in \{1,\dots, 100\}$ of exceedances (horizontally). The concentration parameter is $\lambda\in \{0,10^{-4},5.10^{-4},10^{-3}\}$, respectively in violet, blue, green and yellow. Coloured areas correspond to $90\%$ confidence intervals.%, lines to medians.
}

\begin{figure}[htb]
\centering
\begin{tabular}{ccc}
    & \multicolumn{2}{c}{Sparse Laplace prior and link function $g(t)=t^c$ with $c=1$.}\\
    &  $d=30$ &   $d=300$\\
    \rotatebox[origin=l]{90}{\hspace{.9cm}$\tau=-0.8$} &
    \includegraphics[trim={1cm 1.5cm 0.5cm 2cm},clip,width=.43\textwidth]{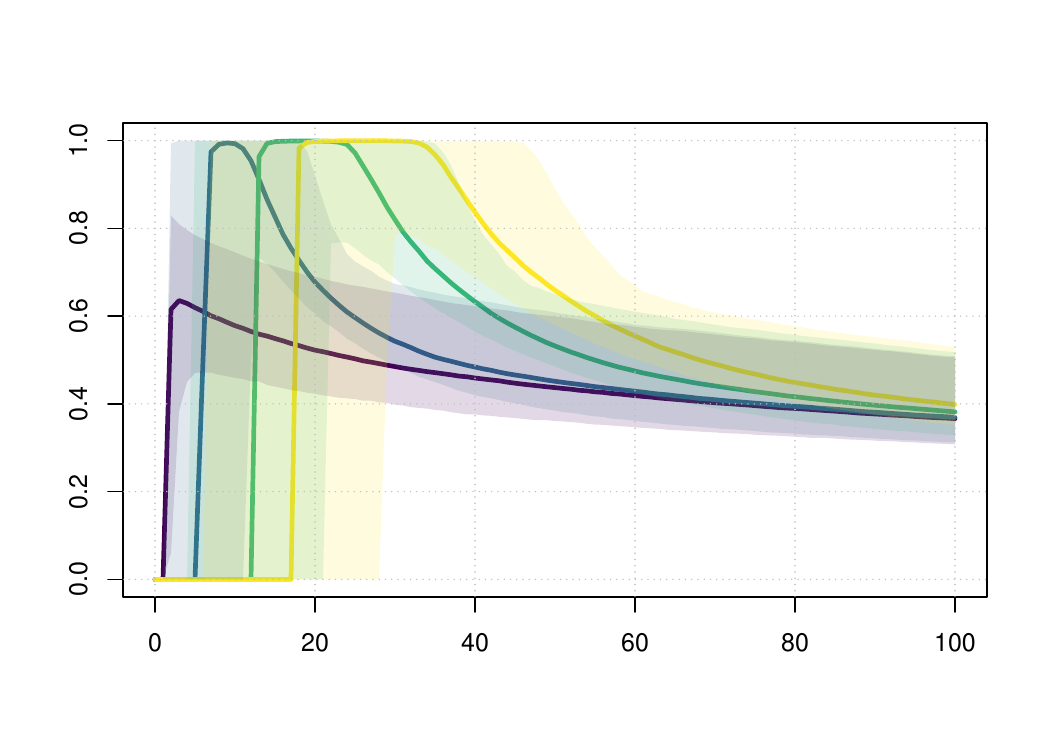}
    &
    \includegraphics[trim={1cm 1.5cm 0.5cm 2cm},clip,width=.43\textwidth]{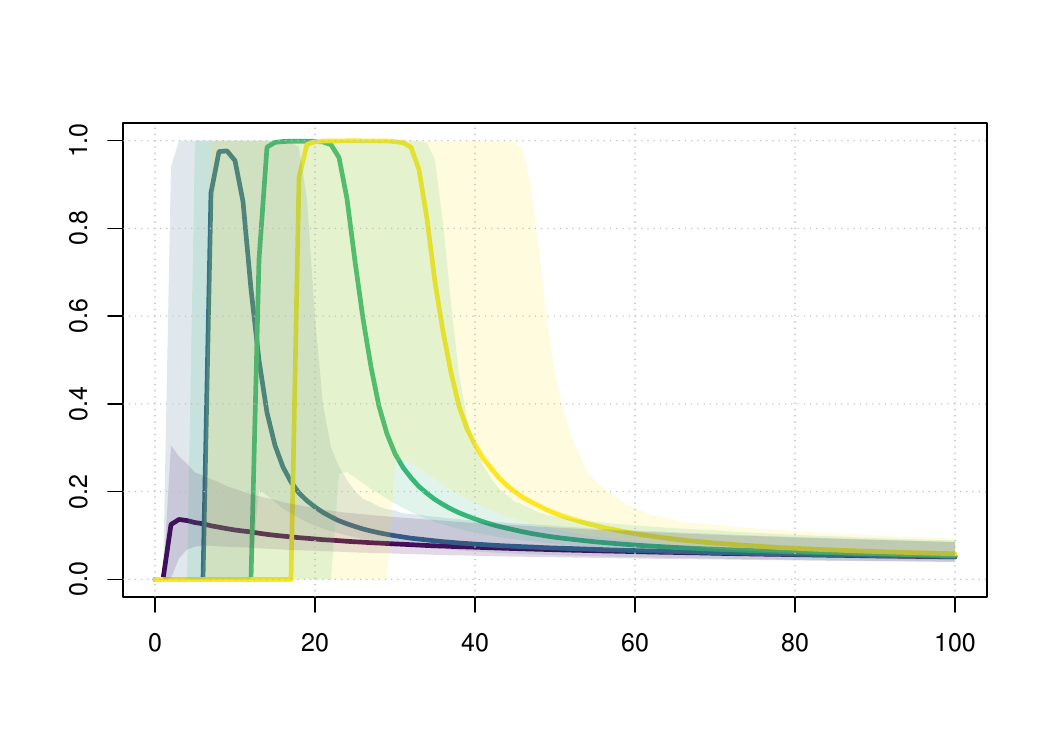}
 \\ \rotatebox[origin=l]{90}{\hspace{.9cm}$\tau=-0.2$} &
    \includegraphics[trim={1cm 1.5cm 0.5cm 2cm},clip,width=.43\textwidth]{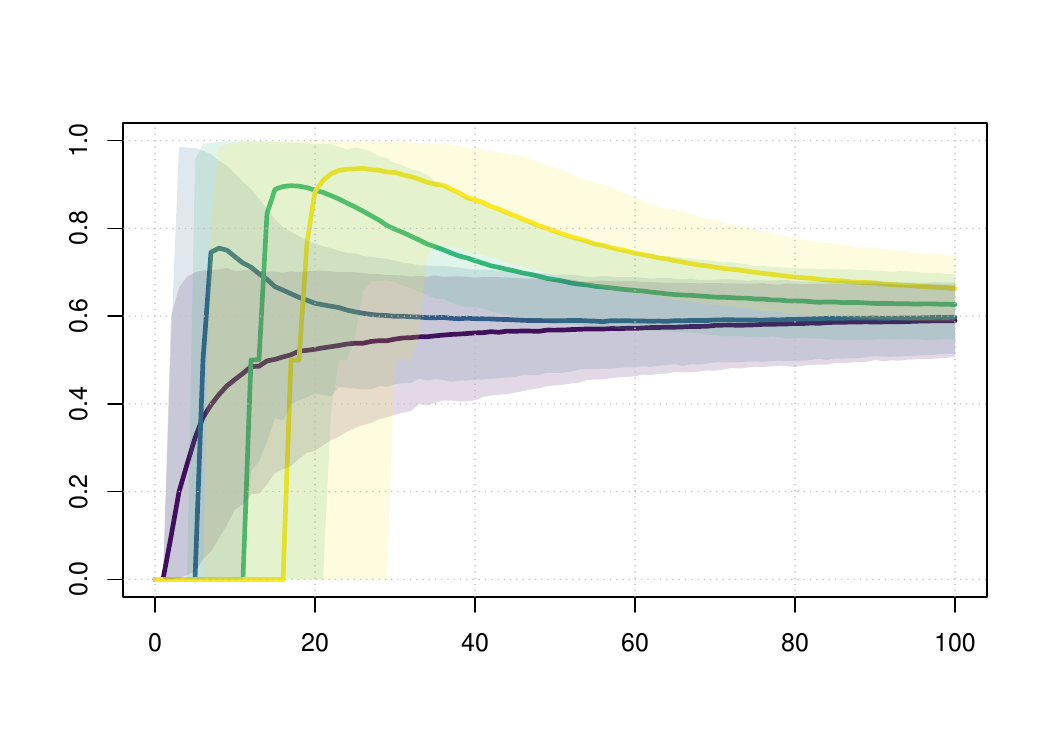}
    &
    \includegraphics[trim={1cm 1.5cm 0.5cm 2cm},clip,width=.43\textwidth]{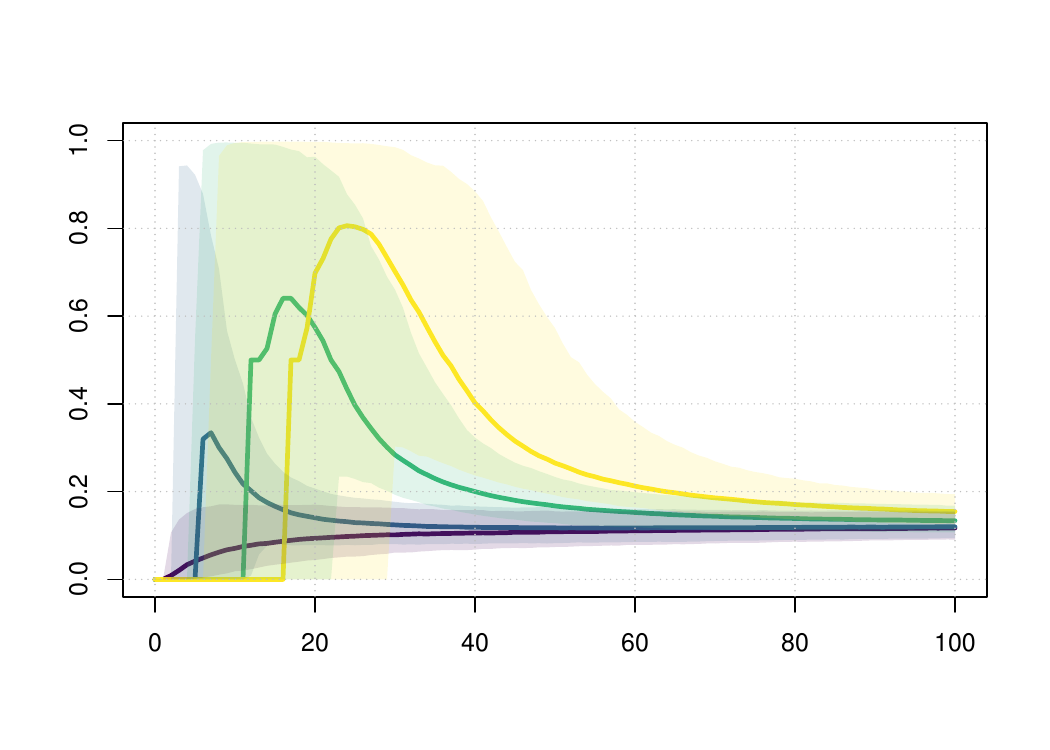}
 \\ \rotatebox[origin=l]{90}{\hspace{.9cm}$\tau=0.2$} &
    \includegraphics[trim={1cm 1.5cm 0.5cm 2cm},clip,width=.43\textwidth]{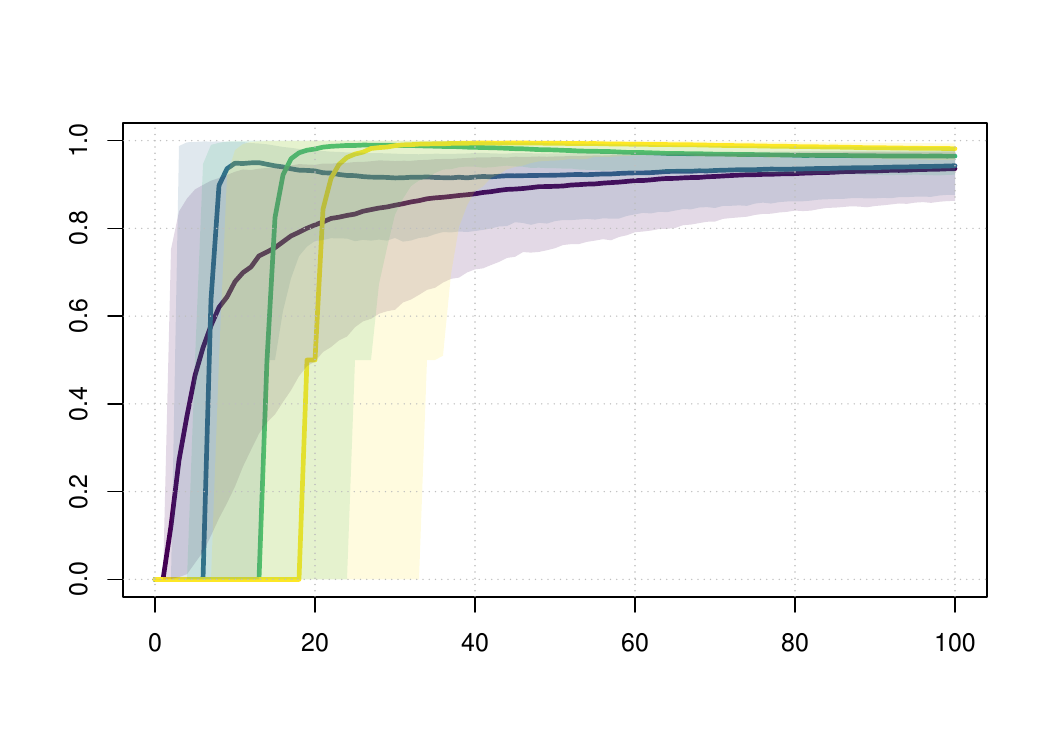}
    &
    \includegraphics[trim={1cm 1.5cm 0.5cm 2cm},clip,width=.43\textwidth]{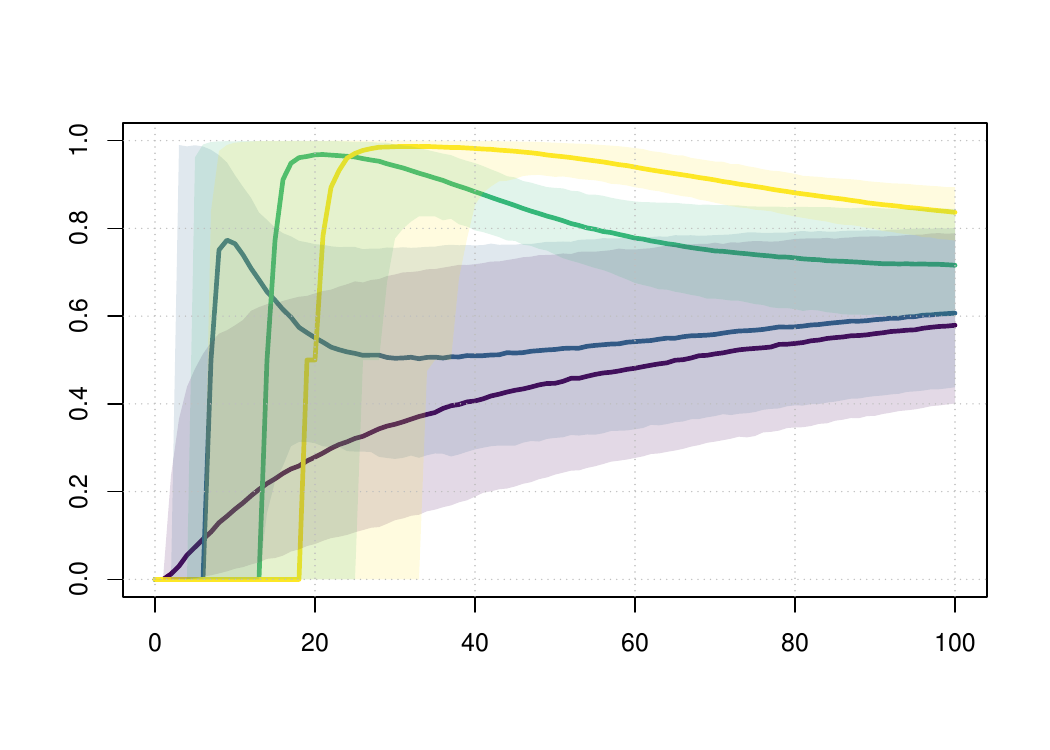}
 \\ \rotatebox[origin=l]{90}{\hspace{.9cm}$\tau=0.8$} &
    \includegraphics[trim={1cm 1.5cm 0.5cm 2cm},clip,width=.43\textwidth]{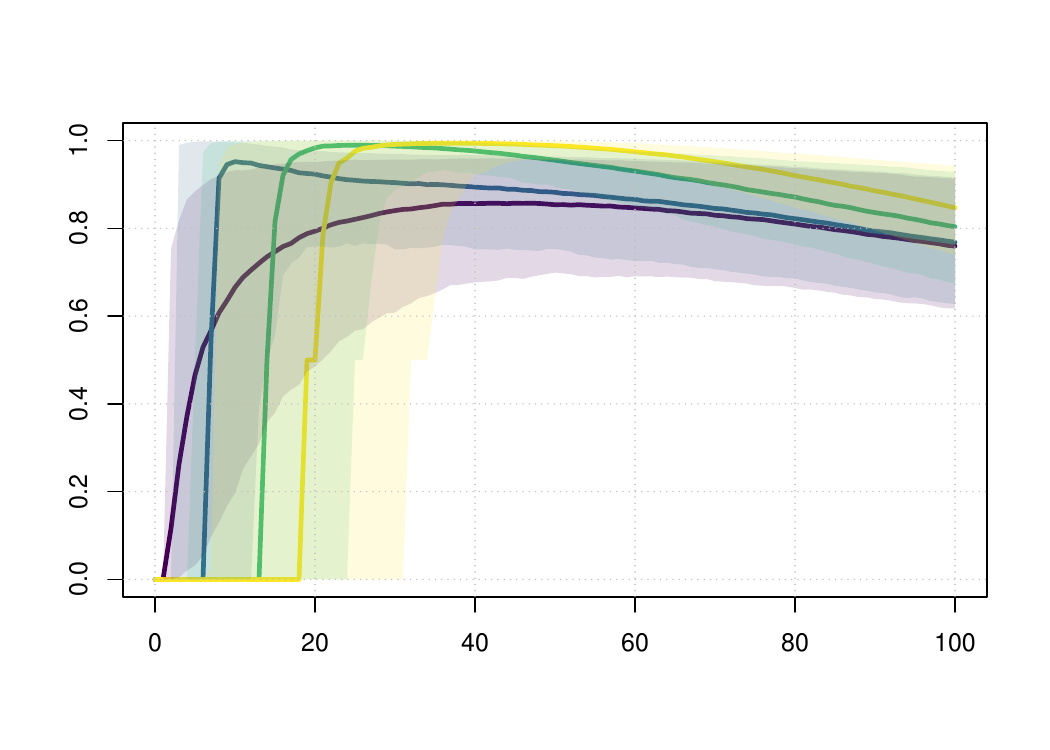}
    &
    \includegraphics[trim={1cm 1.5cm 0.5cm 2cm},clip,width=.43\textwidth]{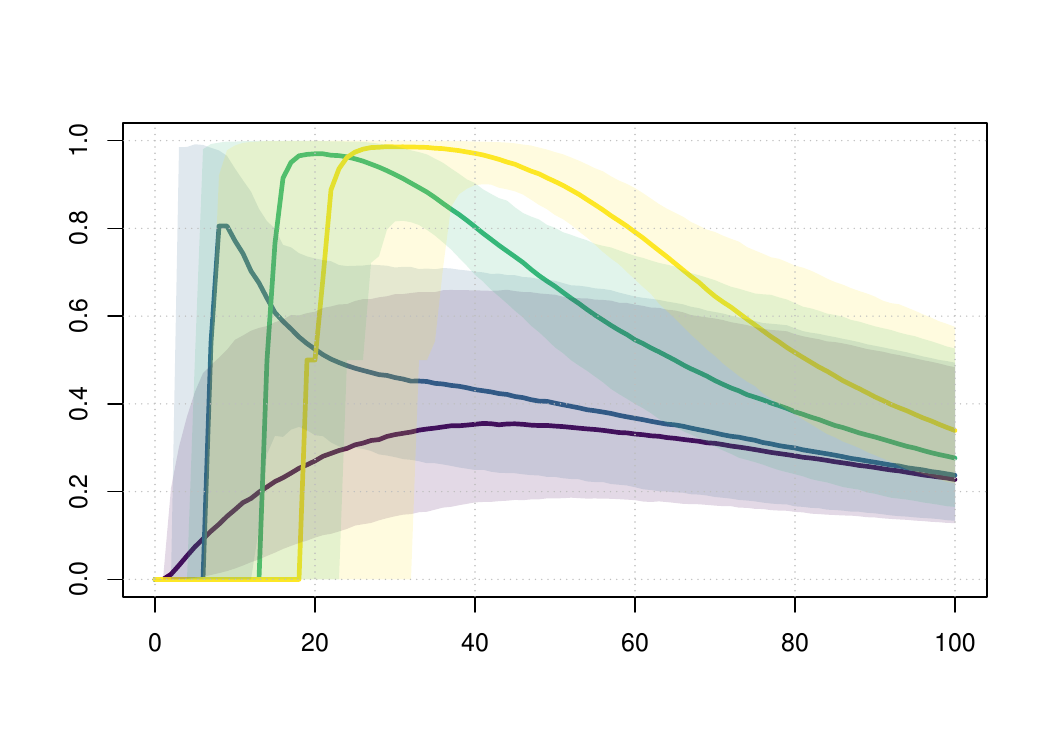}
    \end{tabular}
\caption{\legendsparse{1}}
\label{BEPLSSparse1}
\end{figure}

\begin{figure}[htb]
\centering
\begin{tabular}{ccc}
    & \multicolumn{2}{c}{Sparse Laplace prior and link function $g(t)=t^c$ with $c=1/2$.}\\
    &  $d=30$ &   $d=300$\\
    \rotatebox[origin=l]{90}{\hspace{.9cm}$\tau=-0.8$} &
    \includegraphics[trim={1cm 1.5cm 0.5cm 2cm},clip,width=.43\textwidth]{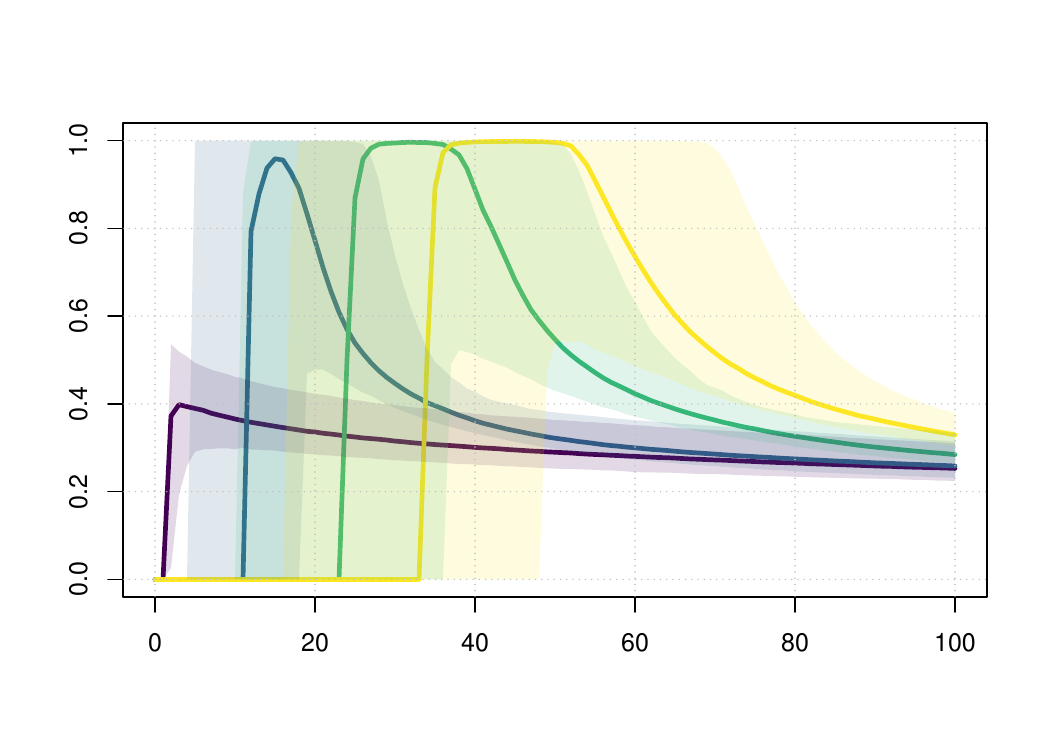}
    &
    \includegraphics[trim={1cm 1.5cm 0.5cm 2cm},clip,width=.43\textwidth]{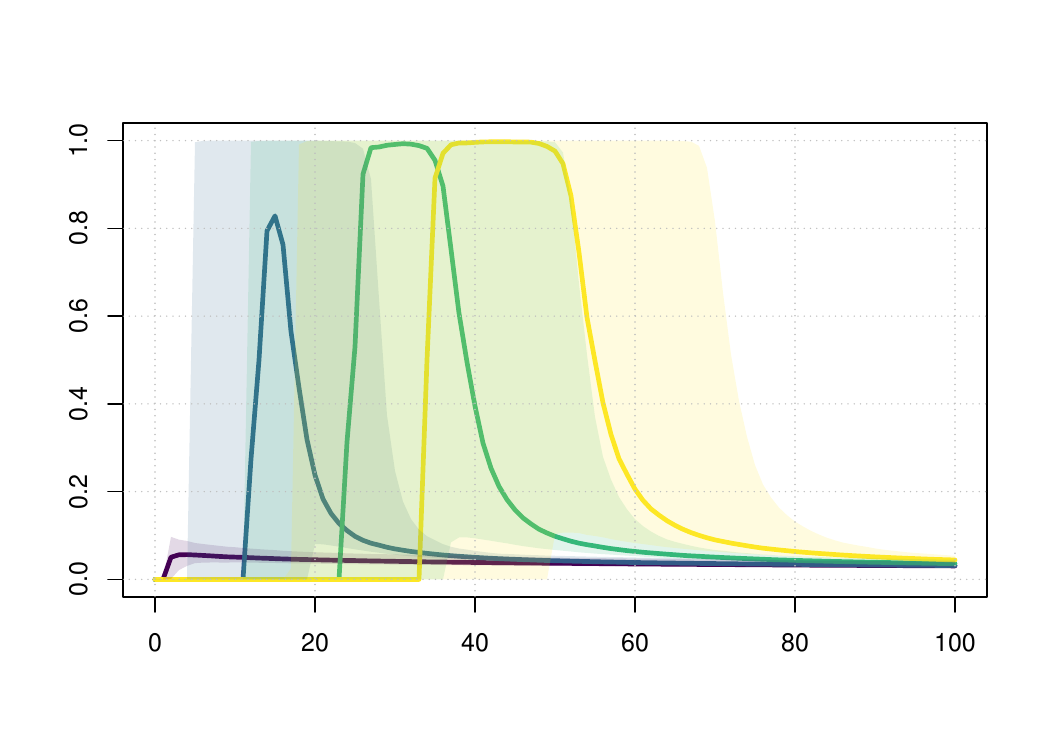}
 \\ \rotatebox[origin=l]{90}{\hspace{.9cm}$\tau=-0.2$} &
    \includegraphics[trim={1cm 1.5cm 0.5cm 2cm},clip,width=.43\textwidth]{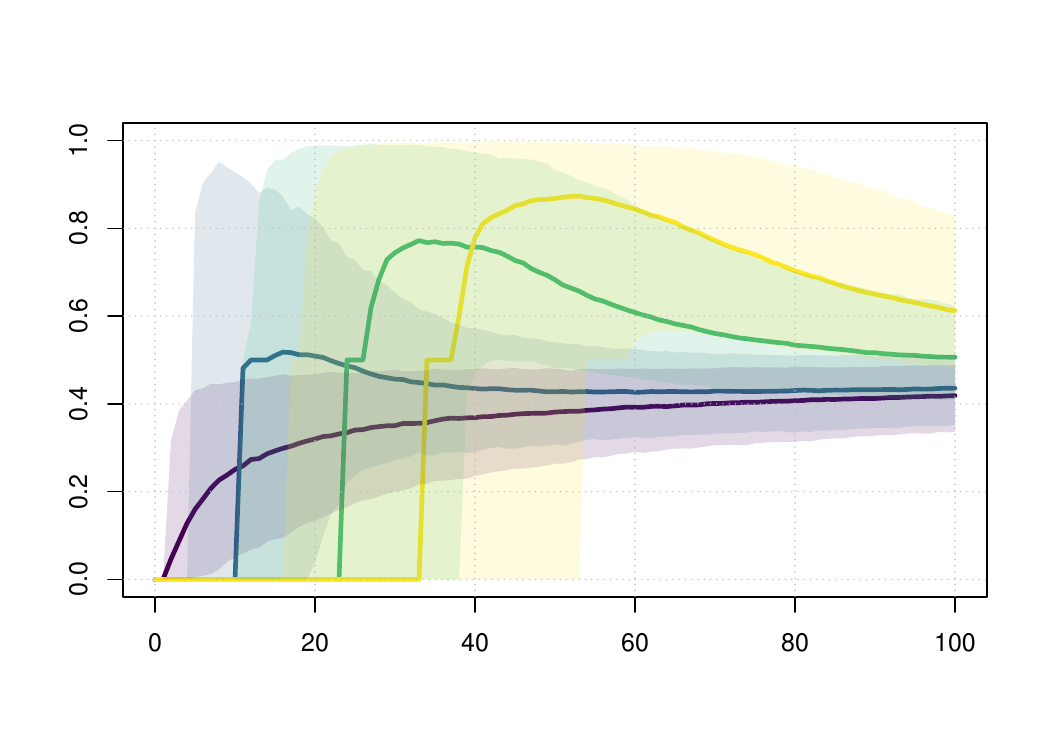}
    &
    \includegraphics[trim={1cm 1.5cm 0.5cm 2cm},clip,width=.43\textwidth]{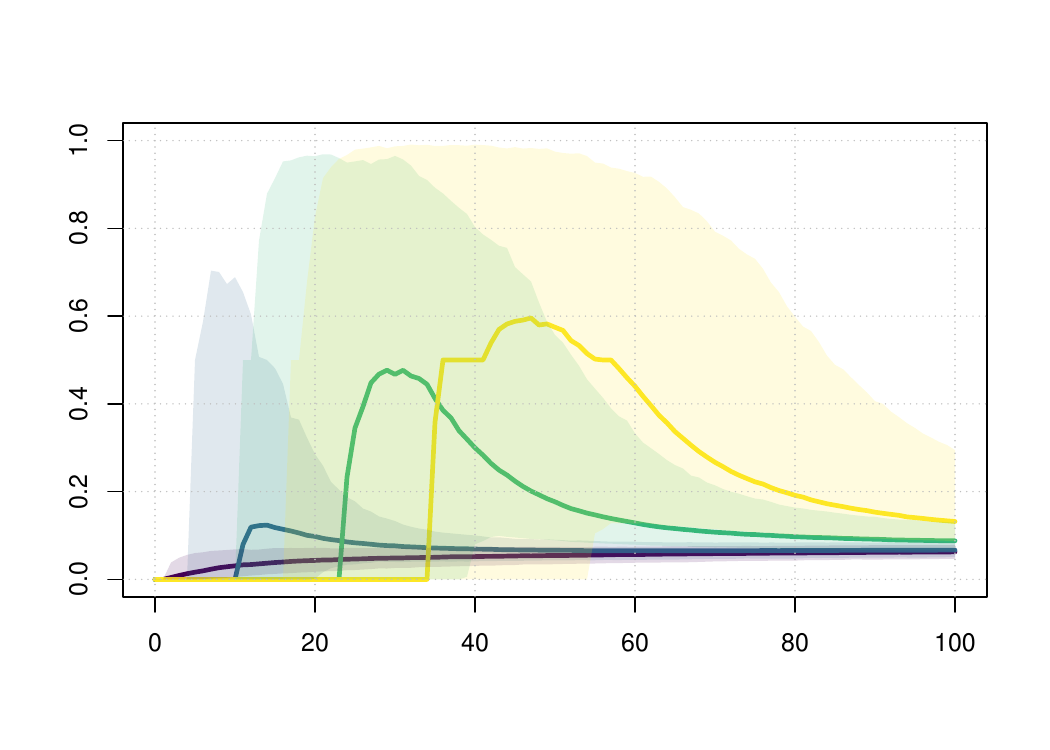}
 \\ \rotatebox[origin=l]{90}{\hspace{.9cm}$\tau=0.2$} &
    \includegraphics[trim={1cm 1.5cm 0.5cm 2cm},clip,width=.43\textwidth]{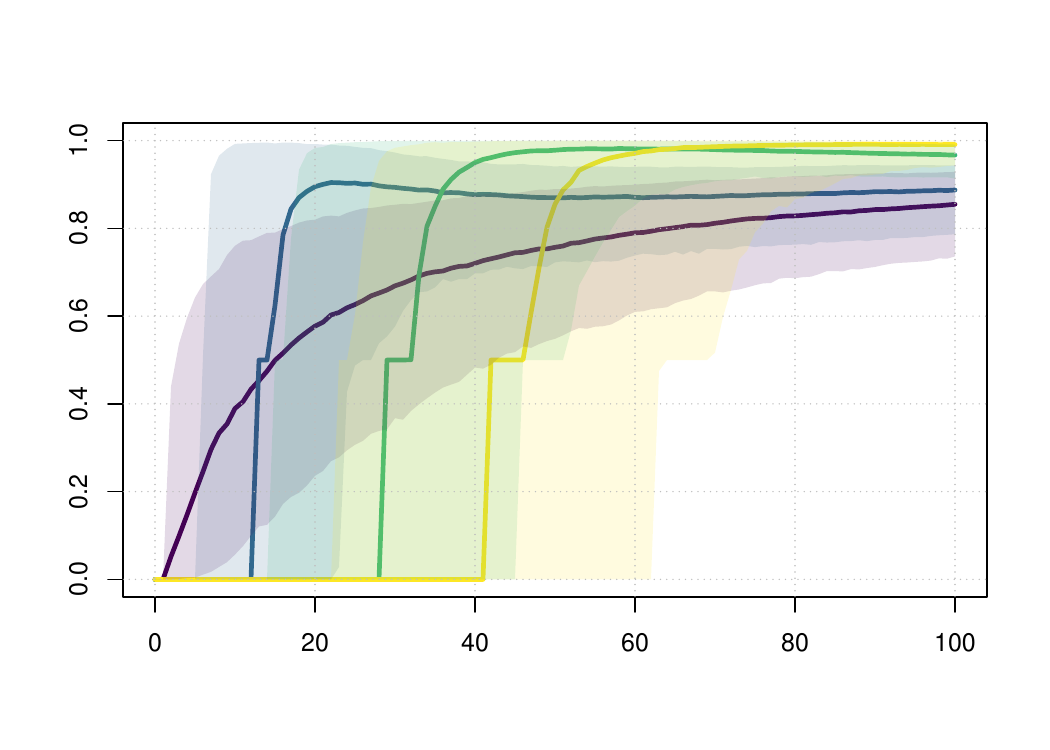}
    &
    \includegraphics[trim={1cm 1.5cm 0.5cm 2cm},clip,width=.43\textwidth]{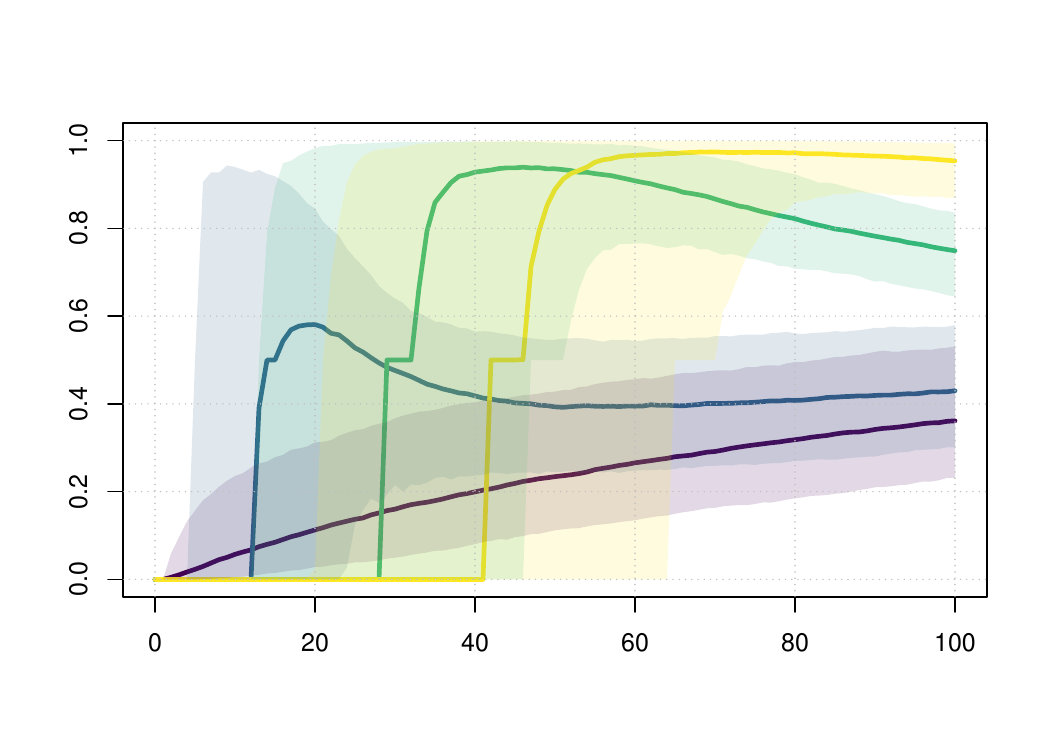}
 \\ \rotatebox[origin=l]{90}{\hspace{.9cm}$\tau=0.8$} &
    \includegraphics[trim={1cm 1.5cm 0.5cm 2cm},clip,width=.43\textwidth]{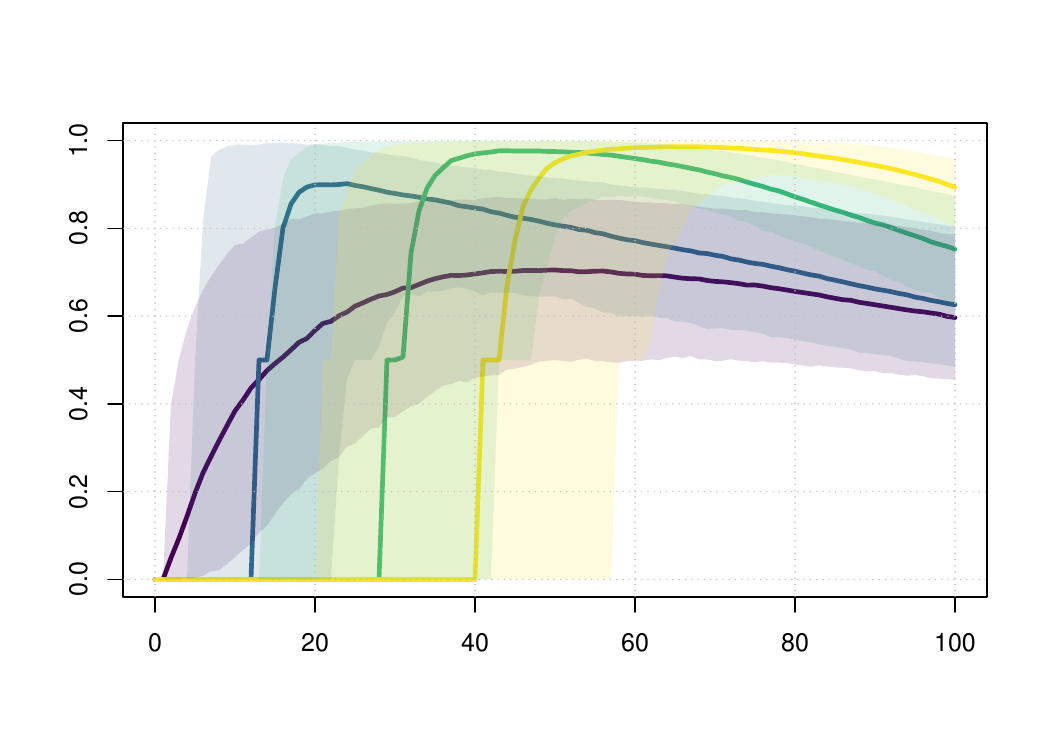}
    &
    \includegraphics[trim={1cm 1.5cm 0.5cm 2cm},clip,width=.43\textwidth]{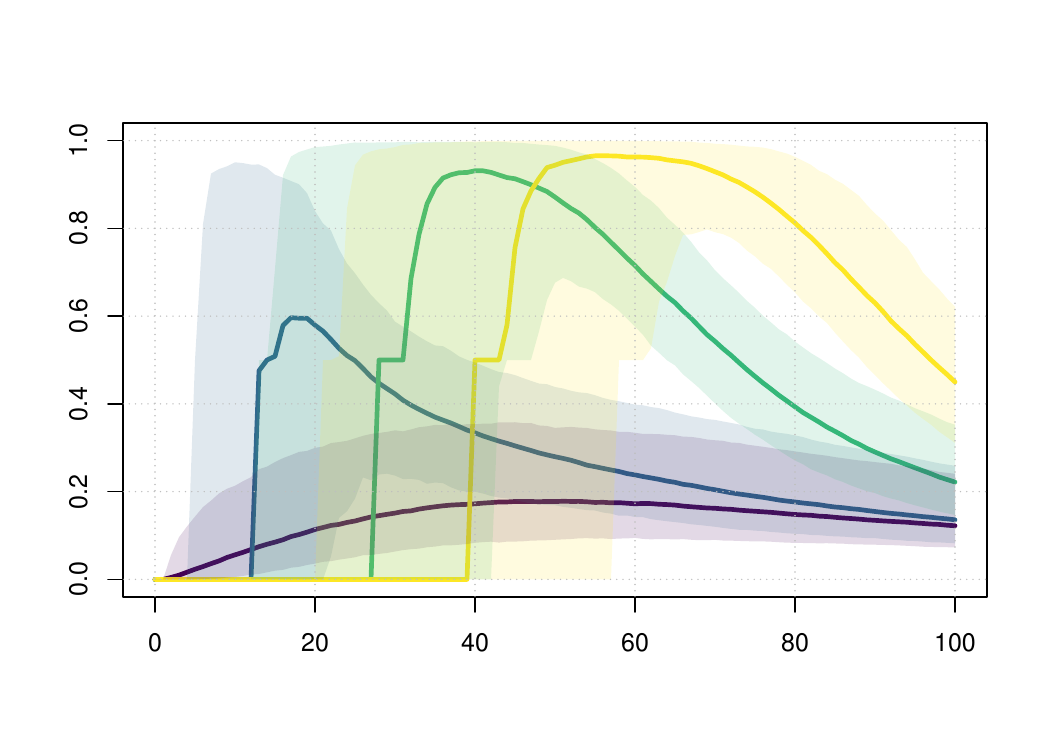}
    \end{tabular}
\caption{\legendsparse{1/2}}
\label{BEPLSSparse2}
\end{figure}

\begin{figure}[htb]
\centering
\begin{tabular}{ccc}
    & \multicolumn{2}{c}{Sparse Laplace prior and link function $g(t)=t^c$ with $c=1/4$.}\\
    &  $d=30$ &   $d=300$\\
    \rotatebox[origin=l]{90}{\hspace{.9cm}$\tau=-0.8$} &
    \includegraphics[trim={1cm 1.5cm 0.5cm 2cm},clip,width=.43\textwidth]{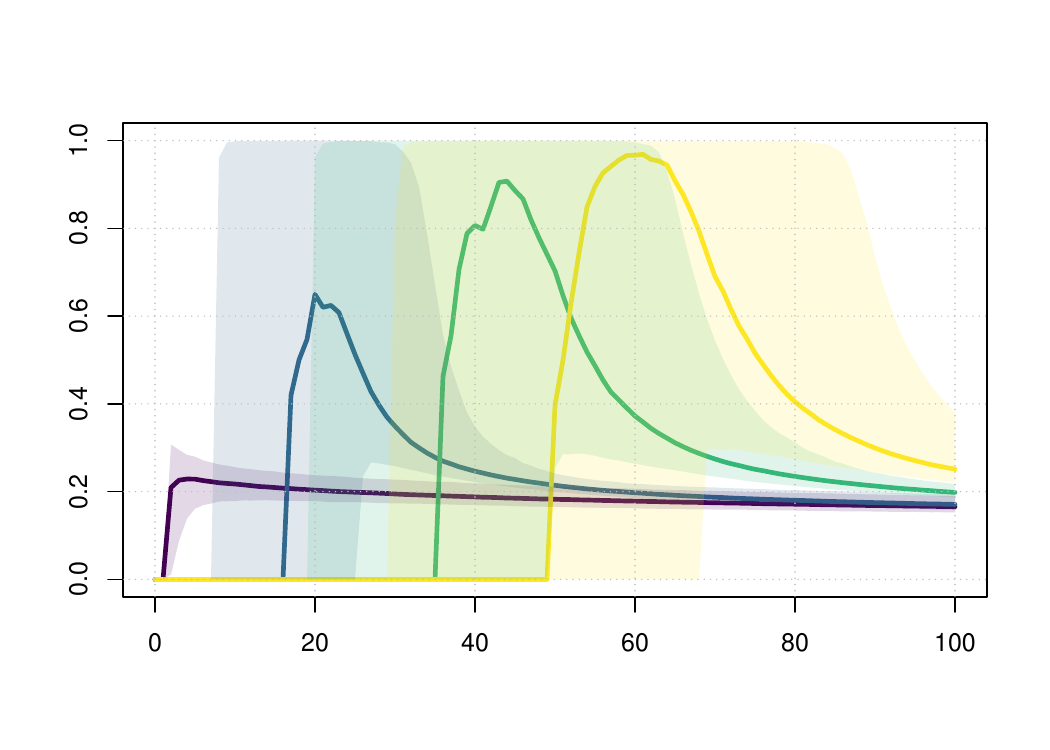}
    &
    \includegraphics[trim={1cm 1.5cm 0.5cm 2cm},clip,width=.43\textwidth]{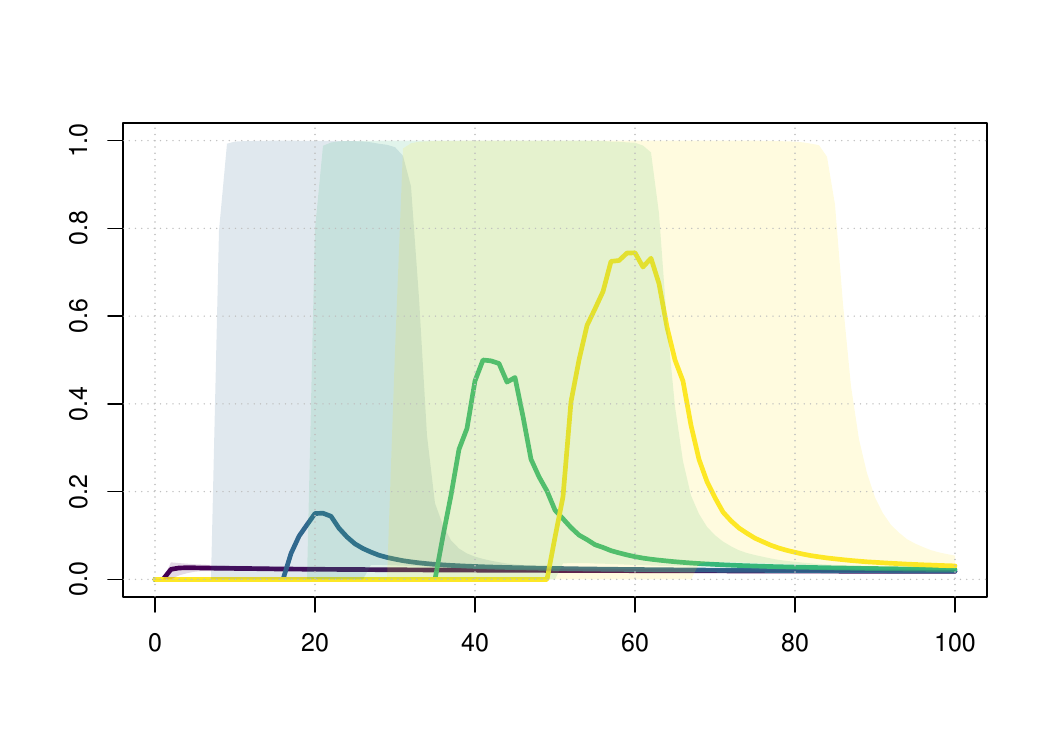}
 \\ \rotatebox[origin=l]{90}{\hspace{.9cm}$\tau=-0.2$} &
    \includegraphics[trim={1cm 1.5cm 0.5cm 2cm},clip,width=.43\textwidth]{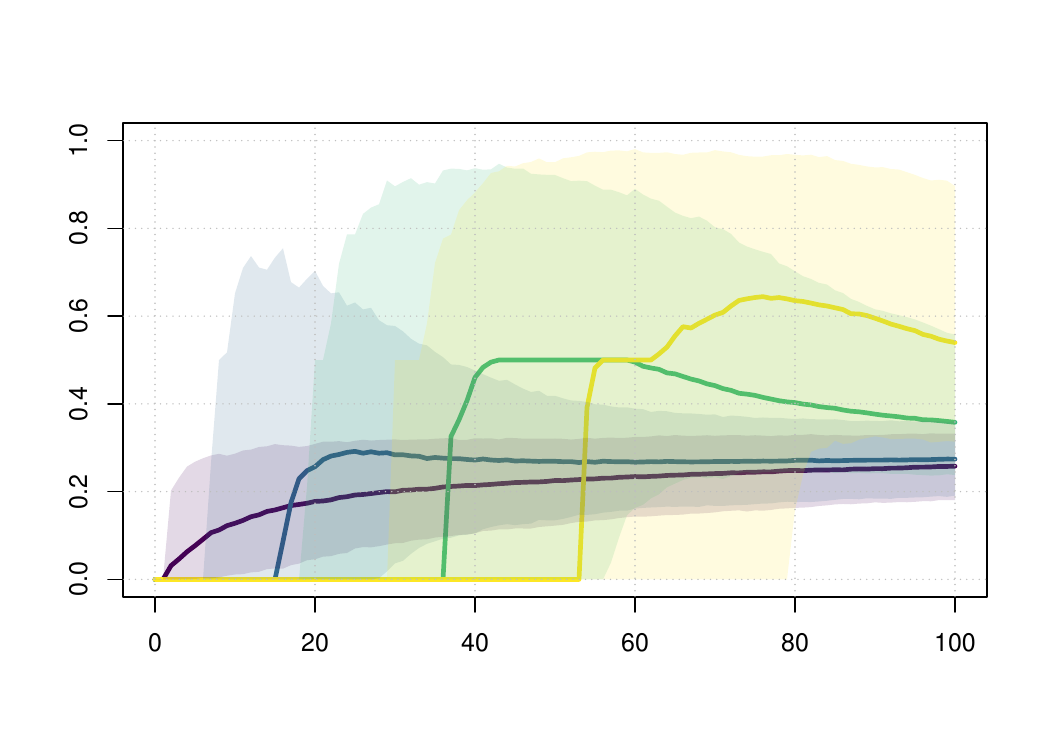}
    &
    \includegraphics[trim={1cm 1.5cm 0.5cm 2cm},clip,width=.43\textwidth]{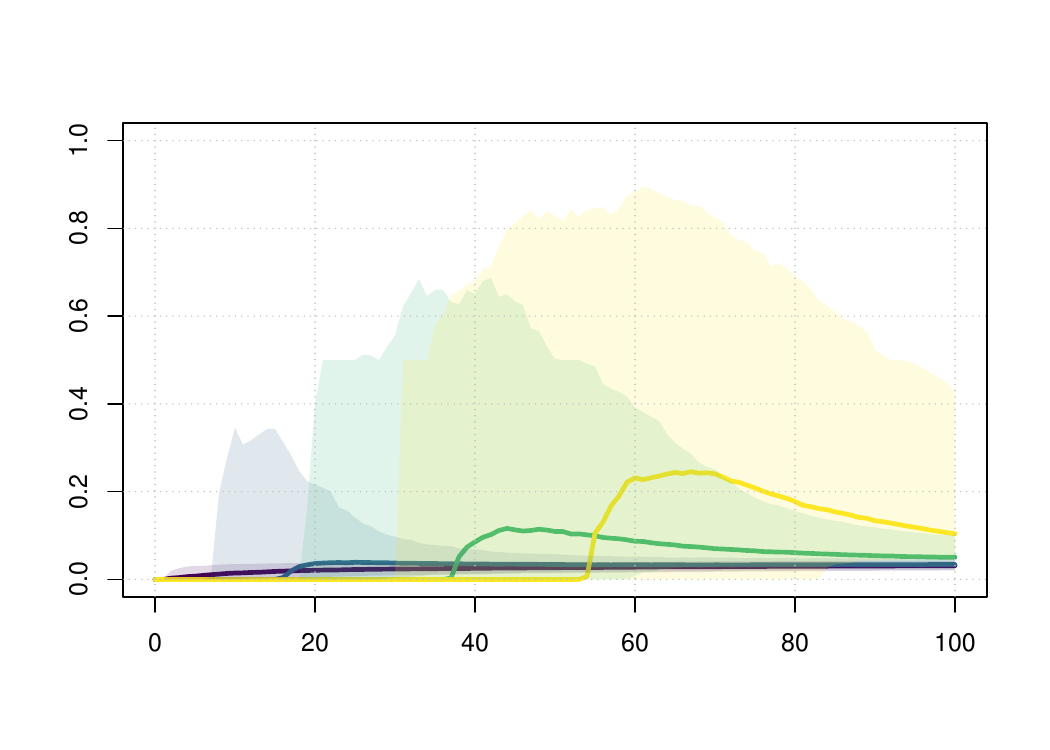}
 \\ \rotatebox[origin=l]{90}{\hspace{.9cm}$\tau=0.2$} &
    \includegraphics[trim={1cm 1.5cm 0.5cm 2cm},clip,width=.43\textwidth]{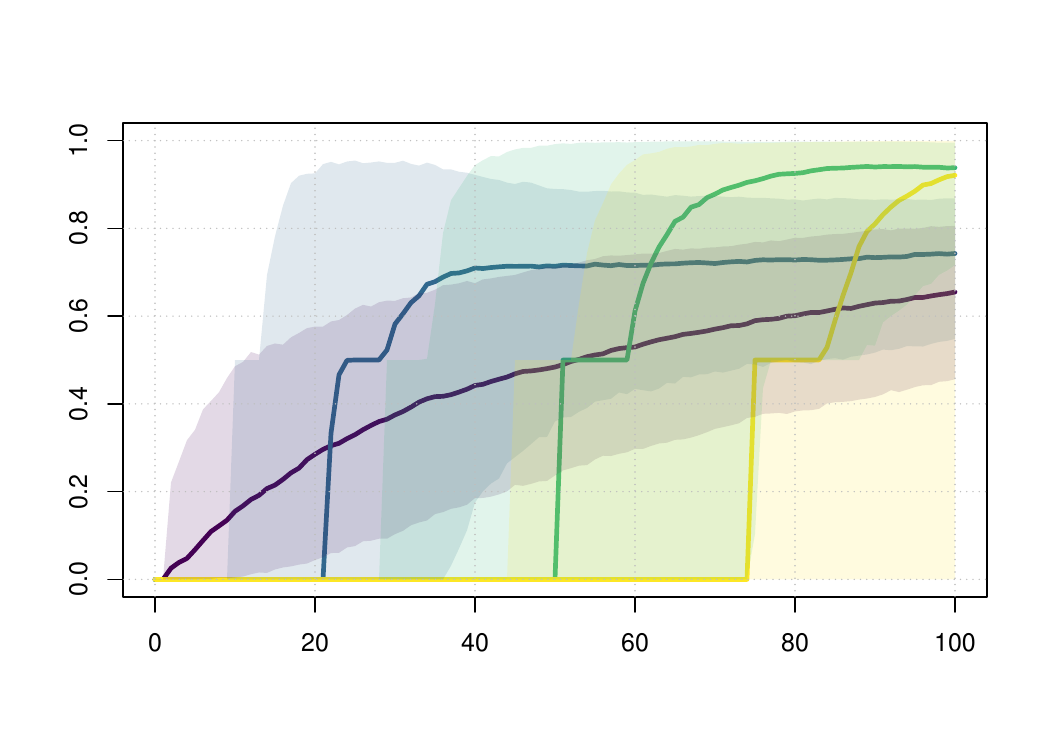}
    &
    \includegraphics[trim={1cm 1.5cm 0.5cm 2cm},clip,width=.43\textwidth]{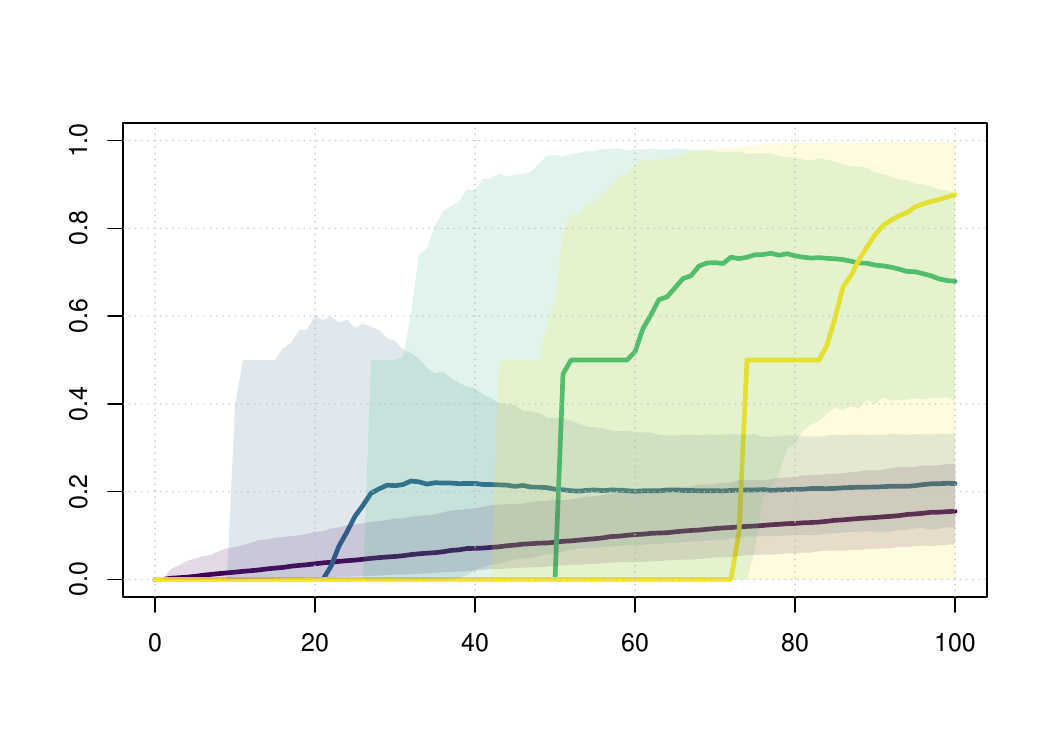}
 \\ \rotatebox[origin=l]{90}{\hspace{.9cm}$\tau=0.8$} &
    \includegraphics[trim={1cm 1.5cm 0.5cm 2cm},clip,width=.43\textwidth]{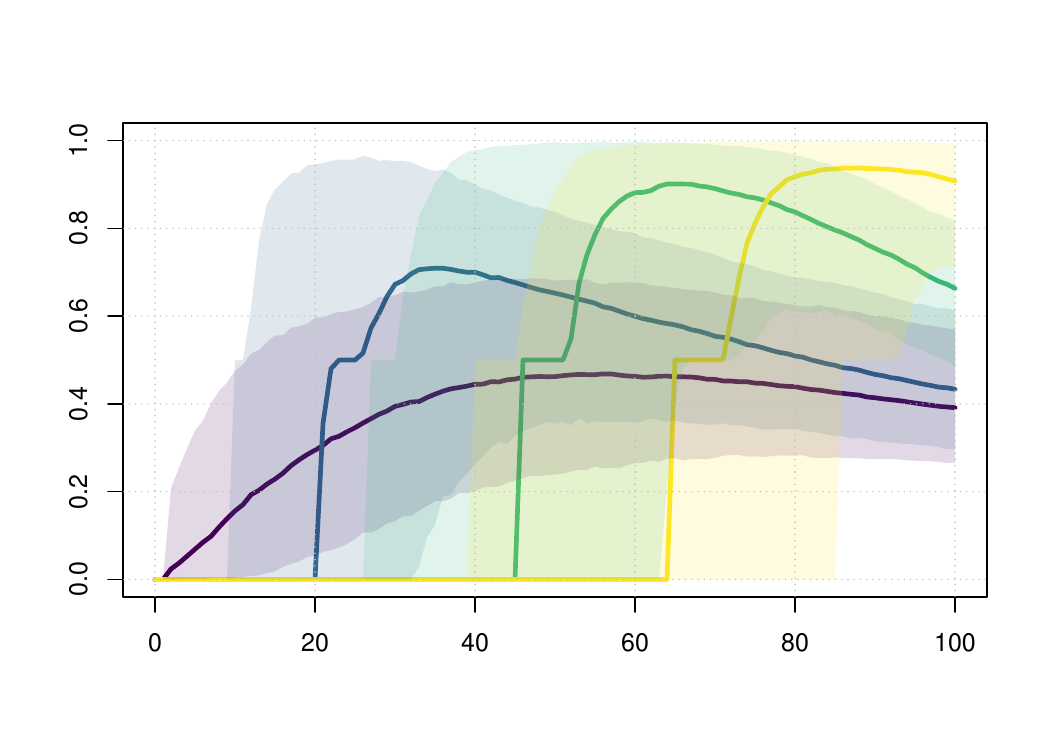}
    &
    \includegraphics[trim={1cm 1.5cm 0.5cm 2cm},clip,width=.43\textwidth]{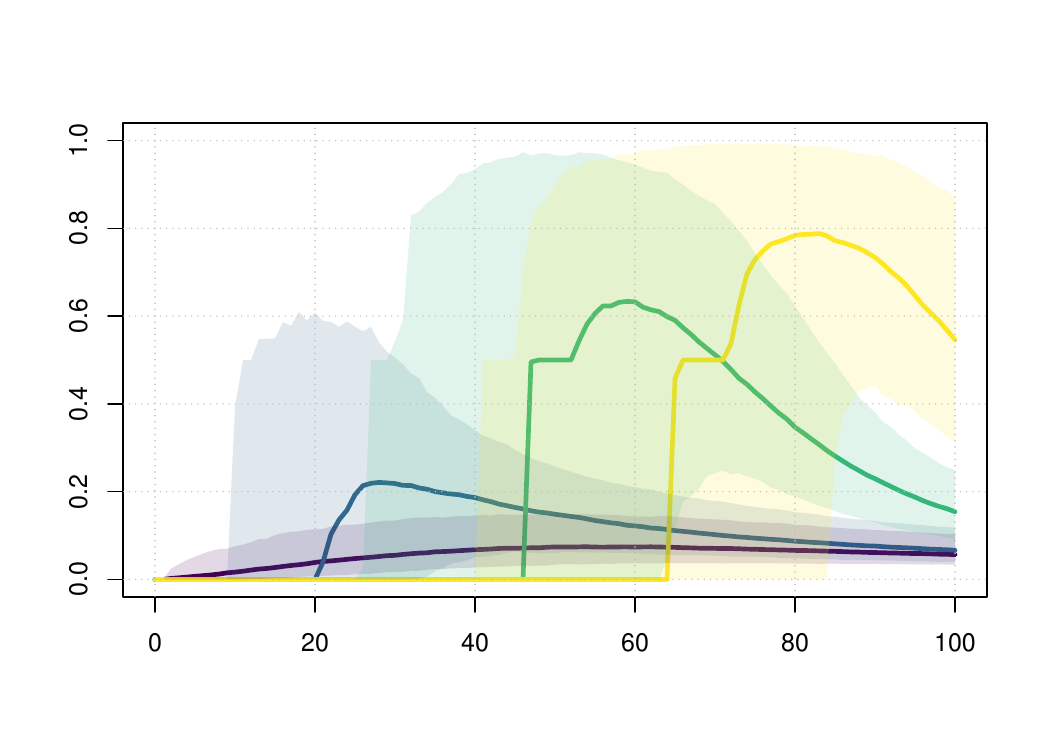}
    \end{tabular}
    \caption{\legendsparse{1/4}}
\label{BEPLSSparse3}
\end{figure}

\end{document}